\documentclass[letterpaper,onecolumn,10pt]{article}

\usepackage{multicol}
\usepackage{amsmath,amssymb}
\usepackage{times}
\usepackage{fullpage}
\usepackage[figure,vlined,linesnumbered]{algorithm2e}
\usepackage{color}
\usepackage{rotating}
\usepackage{soul}
\usepackage{helvet}
\usepackage{url}
\usepackage{refcount}
\usepackage[numbers,sort]{natbib}
\usepackage{inconsolata}
\usepackage{subfigure}
\usepackage{float}

\usepackage{enumitem}
\usepackage{multirow}
\usepackage{fancyhdr}
\usepackage{verbatim}
\usepackage{array}
\usepackage{pdfsync}
\usepackage[breaklinks,colorlinks,linkcolor=black,citecolor=black,urlcolor=black]{hyperref}
\usepackage{cleveref}

\newcommand{\name}{GHAST}

\newcommand{\todo}[1]{}
\newcommand{\phvname}[1]{{\fontfamily{phv}\selectfont #1}}

\newcommand{\statecondm}{For any adversary state $\advs$ appearing in the execution of ghast protocol, let $\advs\l$ be the last adversary state and $\event$ be the event updates $\advs\l$ to $\advs$}
\newcommand{\stateconds}{For any adversary state $\advs$ appearing in the execution of ghast protocol}
\DeclareMathOperator*{\argmax}{arg\,max}
\newcommand{\true}{\mathsf{True}}
\newcommand{\false}{\mathsf{False}}
\renewcommand{\l}{_{-1}}

\newcommand{\eqdef}{:=}

\newcommand{\proto}{\Pi}
\newcommand{\order}{{\cal C}}
\newcommand{\state}{{\cal B}}
\newcommand{\env}{{\cal Z}}
\newcommand{\adv}{{\cal A}}
\newcommand{\secp}{\kappa}

\newcommand{\digest}{{\sf digest}}
\newcommand{\prf}{{\cal H}}
\newcommand{\mine}{{\sf H}}
\newcommand{\verify}{{\sf H.ver}}

\newcommand{\block}{\mathbf{b}}
\newcommand{\graph}{\mathbf{B}}
\newcommand{\genesisblock}{\mathbf{g}}

\newcommand{\uniongraphs}{{\cal U}}
\newcommand{\rmax}{r_{\sf max}}
\newcommand{\negl}{{\sf negl}}
\newcommand{\view}{{\sf View}}

\newcommand{\rcon}{{r_{\sf con}}}

\newcommand{\delay}{d}
\newcommand{\difficulty}{{\eta_{\sf d}}}
\newcommand{\heavyw}{\eta_{\sf w}}
\newcommand{\advan}{\eta_{\sf a}}
\newcommand{\timerw}{\eta_{\sf t}}
\newcommand{\timerdiff}{\eta_{\sf b}}

\newcommand{\chain}[1]{{\mathrm{Chain}\!\left(#1\right)}}
\newcommand{\tree}[1]{{\mathrm{SubT}\!\left(#1\right)}}
\newcommand{\treew}[1]{{\mathrm{SubTW}\!\left(#1\right)}}
\newcommand{\sibtreew}[1]{{\mathrm{SibSubTW}\!\left(#1\right)}}
\newcommand{\child}{\mathrm{Chldn}}

\newcommand{\parent}{\mathsf{parent}}
\newcommand{\weight}{\mathsf{weight}}
\newcommand{\past}{\mathsf{past}}

\newcommand{\adp}{\mathrm{Adapt}}

\newcommand{\mainchild}{\mathbf{c}}

\newcommand{\kah}{{s_{\sf h}}}
\newcommand{\kam}{{s_{\sf m}}}

\newcommand{\event}{e}
\newcommand{\hgen}{\mathsf{hGenRls}}
\newcommand{\mgen}{\mathsf{mGen}}
\newcommand{\mrls}{\mathsf{mRls}}
\newcommand{\allrec}{\mathsf{Arvl}}

\newcommand{\advs}{\mathcal{S}}
\newcommand{\flagb}{\mathbf{f}}
\newcommand{\speset}{\mathbf{S}}
\newcommand{\spevalue}{v}
\newcommand{\vcmp}{\mathbf{C}} 

\newcommand{\tmpset}{\mathbf{T}}
\newcommand{\gmax}{\graph^{\sf max}}
\newcommand{\gmin}{\graph^{\sf min}}
\newcommand{\ggen}{\graph^{\sf gen}}
\newcommand{\gdta}{\graph^{\Delta}}

\newcommand{\eventweight}{\Delta}
\newcommand{\pot}{P}
\newcommand{\cpot}{{\tilde P}}
\newcommand{\potwith}{P_\mathsf{with}}
\newcommand{\potadv}{P_\mathsf{adv}}
\newcommand{\potspe}{P_\mathsf{sp}}

\newcommand{\tip}{\mathrm{Tip}}
\newcommand{\nxt}{\mathrm{Next}}

\newcommand{\eblock}{\block_{\sf e}}
\newcommand{\type}{\mathsf{type}}
\newcommand{\tw}{\mathrm{TotalW}}
\newcommand{\ti}{\tilde}

\newcommand{\varm}{M}
\newcommand{\varh}{H}
\newcommand{\varf}{F}
\newcommand{\vars}{S}
\newcommand{\varsm}{S^{\sf M}}
\newcommand{\varsh}{S^{\sf H1}}
\newcommand{\varsf}{S^{\sf H2}}

\newcommand{\rd}{r_{\Delta}}
\newcommand{\probd}{\phi}
\newcommand{\probf}{\ti\phi}

\newcommand{\tmpi}{\gamma}

\newtheorem{lemma}{Lemma}[section]
\newtheorem{theorem}[lemma]{Theorem}
\newtheorem{claim}[lemma]{Claim}

\newtheorem{definition}[lemma]{Definition}

\newenvironment{proof}{\vspace{-0.05in}\noindent{\bf Proof.}}%
{\hspace*{\fill}$\Box$\par}
{\hspace*{\fill}$\Box$\par\vspace{4mm}}
\newenvironment{proofof}[1]{\smallskip\noindent{\bf Proof of #1.}}%
{\hspace*{\fill}$\Box$\par}

\begin{document}
\title{\bf {\name}: Breaking Confirmation Delay Barrier in Nakamoto Consensus via Adaptive Weighted Blocks}
\author{
{\rm Chenxing Li}\\
Tsinghua University
\and
{\rm Fan Long}\\
University of Toronto
\and
{\rm Guang Yang}\\
Conflux Foundation
} 
\date{}
\maketitle

\begin{abstract}
	Initiated from Nakamoto's Bitcoin system, blockchain technology has demonstrated great capability of building secure consensus among decentralized parties at Internet-scale, i.e., without relying on any centralized trusted party. Nowadays, blockchain systems find applications in various fields. But the performance is increasingly becoming a bottleneck, especially when permissionless participation is retained for full decentralization.

	In this work, we present a new consensus protocol named GHAST (Greedy Heaviest Adaptive Sub-Tree) which organizes blocks in a Tree-Graph structure (i.e., a directed acyclic graph (DAG) with a tree embedded) that allows fast and concurrent block generation. GHAST protocol simultaneously achieves a logarithmically bounded liveness guarantee and low confirmation latency. More specifically, for maximum latency $d$ and adversarial computing power bounded away from 50\%, GHAST guarantees confirmation with confidence $\ge 1-\varepsilon$ after a time period of $O(d\cdot \log(1/\varepsilon))$. When there is no observable attack, GHAST only needs $3d$ time to achieve confirmation at the same confidence level as six-block-confirmation in Bitcoin, while it takes roughly $360d$ in Bitcoin.
\end{abstract}


\section{Introduction}

Blockchain systems like Bitcoin provide secure, decentralized, and consistent ledgers at Internet-scale. 
Such ledgers are initially designed for
cryptocurrencies, but now have evolved to become a powerful abstraction that
fuels innovations on many real-world applications across financial systems~\cite{DeloitteFinancial}, supply
chains~\cite{IBMSupplyChain}, and health cares~\cite{DeloitteHealthCare}.

A novel aspect of blockchain systems is permissionless. It allows anyone to join or leave the system freely without getting approval from some centralized or distributed community. During the blockchain protocol execution, it is not necessary for any participant to be aware of others, once the protocol message from the other participants can be relayed on time. In order to prevent malicious behavior in a permissionless setting, the blockchain systems limit the rate for constructing new blocks using the idea of computational puzzles, which is called \emph{proof-of-work}. To construct a valid block accepted by the blockchain protocol, the participants need to set a proper nonce in block to make the hash value of such block fall in a prescribed bit-string set. (e.g., the set contains bit-string with 70 leading zeros.) So all the participants need to try a large number of different nonces before finding a valid block. 

The robustness of a blockchain system requires a majority of computing power is held by honest participants. So in the long term, the honest participants will generate more blocks than the attacker. Based on this fact, the blockchain protocol directs the participants to organize blocks and select a sequence of blocks as history. For example, Bitcoin adopts \emph{Nakamoto
consensus} protocol~\cite{bitcoin} which operates on a tree of blocks and selects the longest branch as its correct history. All the honest participants are required to append new blocks to the longest branch. Ideally, all the blocks generated by the honest participants will extend the longest chain in Nakamoto consensus. 
If we assume the attacker controls at most $\beta$ computing power ($\beta<1/2$) in total, for each block, we can compute the risk that such block is kicked out of history under the optimal attack strategy in the future. If the risk does not exceed a given threshold, we say a block is \emph{confirmed}. Once a block is confirmed, we have high confidence for the transactions (e.g., payment message) carried in such block is recorded on the blockchain irreversible. 
A protocol that lacks consideration of all possible attack strategies may provide an incorrect way to estimate confirmation risk. And further, a confirmed block may be kicked out of the history frequently. It is regarded as a security flaw. So a rigorous security analysis handling all the possible attack cases is necessary for a blockchain protocol.


The network propagation delay brings issues in reaching consensus. Since a block can not be relayed to other participants instantly, the participants sometimes have an inconsistent view of the current block sets. The influence of network delay depends on the protocol parameters. A blockchain protocol can adjust the block generation interval by adjusting the difficulty of finding a valid block. If the block generation interval is much higher than the time propagating a message in network, with a high probability, no one will generate a new block when the participants have inconsistent local blocks. The blockchain protocol works in a synchronized network. For the opposite situation, there will be a considerable amount of blocks which generated when the participants have an inconsistent view and may cause the participants to diverge. For example, the participants may regard different branches as the longest branch. For a consensus protocol running in a low block generation interval, the protocol design must deal with the inconsistency view carefully. 

A high block generation interval results in a bad performance. A higher block generation interval means fewer blocks are generated in a given time interval. So the consensus protocol has a low throughput. A high block generation interval also results in a high block confirmation delay. 
For example, Nakamoto consensus requires the block propagation delay $\delay$ (i.e., the delay of
one block propagating to all participants in the P2P network) must be
significantly smaller than the block generation interval of $1/\lambda$. 
%
Otherwise, a large number of blocks will be generated in the scenario that another block is in propagation. These blocks
will not contribute to the growth of the longest chain. Once the chain growth of the longest chain slowing down, it requires less cost for the attacker to construct a side chain competing with the confirmed history. 
Furthermore, the confirmation of a block has to wait for several subsequent blocks, 
since in a permissionless consensus system the agreement is only observable through mined blocks.
Therefore Nakamoto consensus has to 
operate with very low block generation rate (e.g., 1 MB blocks per 10 minutes in
Bitcoin) and suffer from unsatisfactory throughput and confirmation latency (e.g., 6 blocks or equivalently 60 minutes in Bitcoin).

Performance becomes one of the major obstacles that impede
the adoption of blockchain techniques. To resolve the performance issues in Nakamoto consensus, several new protocols are proposed in the last five years. Some of them have a rigorous security analysis. Garay et al~~\cite{garay2015bitcoin} first provide a rigorous security analysis for Nakamoto consensus in a synchronized network model. They prove two properties of Nakamoto consensus, the common prefix and the chain quality. Pass et al~~\cite{pass2017analysis} consider the effect of network delay and provide an analysis in an asynchronized network model with a prior maximum network delay $\delay$. They prove an additional property, the chain-growth. Several subsequent works \cite{BitcoinNG,Fruitchain,yu2018ohie} built their security analysis on the top of the basic properties in Nakamoto consensus. These works achieve a good performance in throughput. However, since they build security on the top of Nakamoto consensus, they can not achieve a better confirmation latency than Nakamoto consensus. 

In the same period, Sompolinsky et al~\cite{GHOST} introduce GHOST (Greedy Heaviest-Observed Sub-Tree), which uses another way to select the branch. 
However, they only analysis the behavior of GHOST under some attack cases. Later, Kiayias et al~\cite{kiayias2017trees} provide a security analysis for GHOST in a low block generation rate under a synchronized network model. 

Unfortunately, Natoli et al~\cite{natoli2016balance} point out GHOST is vulnerable in a liveness attack when the block generation rate is high. The liveness attack is not aimed to re-ordering the confirmed blocks but tries to prevent from confirming the new blocks. They provide an attack strategy called the balance attack for a high block generation rate. We will introduce the details of this attack in Section 2.
%


We notice that a high block generation rate helps reduce confirmation delay. Given a time interval of $T$, let random variable $X$ denote the ratio of newly generate blocks between malicious blocks and honest blocks. Since the proof-of-work protocols always assume the attacker has less computing power compared to the honest participants, the expectation of $X$ is less than 1. The higher block generation rate, the lower variance $X$ will have. So it is less likely for an attacker to generate more blocks than honest participants in a given time interval. The blocks can gain an advantage in subtree weight compared to the attacker's side chain quickly and achieve a lower confirmation delay. However, the existing protocols built on the top Nakamoto consensus can not break the barrier of confirmation delay in Nakamoto consensus and GHOST suffers a liveness issue in a high block generation rate. 

\subsection{Our contributions}

This paper presents the {\name} (Greedy Heaviest Adaptive Sub-Tree) consensus protocol which achieves a nearly optimal confirmation delay in a normal case with rigorous security analysis. This protocol is designed based on the GHOST protocol with a high block generation rate. In order to resolve the liveness issue in GHOST protocol, GHAST slows down the block generation rate to defense the liveness attack. More precisely, when detecting a divergence of computing power, the block weight distribution is adaptively changed. Only a small fraction of blocks selected randomly are marked as ``heavy blocks'' and other blocks generated under this circumstance are valid but have zero weight. The block generation rate remains to keep a high throughput.  In other words, only the heavy blocks are taken into consideration in the branch selection. 

This work is the first design of a high-throughput BlockDAG consensus protocol (among all DAG-based consensus proposals including a bunch of GHOST-like protocols) that has a rigorous security analysis and liveness 
against an attacker with the ability to manipulate network delay. 
Furthermore, our protocol is also the first consensus protocol that provides both efficiency and robustness: 1) fast confirmation when there is no observable attack, i.e., the agreed history is immutable against covert attacks; 
and 2) polynomially bounded worst-case liveness when there is an active attack with  
$49\%$ block generation power as well as the ability to arbitrarily manipulate the delay of every block within the maximum propagation delay bound $\delay$ (recall that blocks exceeding this bound are counted as malicious).

\paragraph{Liveness guarantee}

We prove that {\name} guarantees security and liveness in the presence of an active attacker who has the power of manipulating communication delays of every block to every participant. Similar to the framework in the analysis of Nakamoto consensus, we have the following assumptions. 

\begin{itemize}
    \item The block generation rate of all the participants is $\lambda$. 
    \item The adversary controls $\beta$ computing powers among all the participants. ($\beta<1/2$) In other words, the block generation rate of the adversary is $\beta\lambda$. 
    \item There is a maximum latency $\delay$ within which a block will be propagated to all honest nodes. 
\end{itemize}

Once a block is received by all the honest participants, its order in history will be consistent among all the honest participants and become unchangeable after time $O(\log(\varepsilon))$, with probability $1-\varepsilon$. More precisely, we have the following theorem. 

\begin{theorem}[Informal]
    For every risk tolerance $\varepsilon>0$ and fixed system parameters $\lambda,\delay,\beta$, let $\delta\eqdef 1-\beta/(1-\beta)$, if $\lambda\delay\ge 5+0.8\log(1/\delta)$, there exist appropriate parameters such that {\name} guarantees that every block broadcast before time $t$ is confirmed with confidence $\ge (1-\varepsilon)$ by time $t+\delay\cdot O\left(\frac{\log(1/\delta\varepsilon)}{\delta^3}\right)$. 
\end{theorem}

A formal version is given in theorem~\ref{thm:final}.

\paragraph{Low confirmation delay}

{\name} achieves fast confirmation time in the absence of an observable attack. Here an “unobservable attack” includes both cases of attacks that happened in the future and covertly withholding blocks without attempting to influence the current state (but withheld blocks may be released in the future). That is, as long as the attacker is not actively inﬂuencing the Block-TG consensus system, transactions can be confirmed quickly and become immutable once confirmed even if under the fast conﬁrmation rule.
We provide a concrete method in estimating the block confirmation risk and runs an experiment. The system parameters in the experiment can tolerate liveness attacks from a powerful attacker that controls 40\% of the network computation power. Conflux blockchain system running the GHAST protocol result in~\cite{conflux-sys}shows that GHAST can obtain the same confidence as waiting for six blocks in Bitcoin in $3\delay$. While Bitcoin requires $360\delay$ and Prism requires $23\delay$.

\subsection{Main techniques}

\paragraph{GHOST protocol and Tree-Graph structure.} 
The GHAST consensus protocol adopts the GHOST protocol proposed in \cite{GHOST} as the backbone of our protocol. We call the branch selected by GHOST protocol \emph{pivot chain}. 

Borrowing ideas from previous works \cite{inclusive,PHANTOM,SPECTRE}, GHAST organizes the block in the Tree-Graph structure. Blocks in are Tree-Graph structure linked with two types of directed edges. Each block has one outgoing \emph{parent edge} to indicate its parent block under GHOST protocol. And it may have multiple outgoing \emph{reference edges} to show generation-before relationship between blocks. The parent edge and reference edges of a block are immutable. The reference edges also reflect the local Tree-Graph of the block's miner when generating such block.

\paragraph{Structured GHOST protocol.}
In structured GHOST protocol, only $1/\heavyw$ of blocks are weighted blocks that would count in the chain selection process, where $\heavyw$ is a protocol parameter. These blocks are selected randomly based on their hash value. During the chain selection process, all the un-weighted blocks are skimmed. It is equivalent to slows down block generation rate $\heavyw$ times. Kiayias et al~\cite{kiayias2017trees} prove that GHOST protocol has no liveness issue when the block generation rate is low enough. Their proof is based on a synchronized network model, which is different from our model. Our further analysis implies that this claim also holds in a partially synchronized network model. 

\paragraph{Consensus with two strategies.}
The GHAST consensus protocol operates with two strategies, an optimistic strategy following the GHOST protocol and a conservative strategy following the structured GHOST protocol. We adopt an adaptive weight mechanism to incorporate two strategies into one framework. The blocks under the GHOST protocol have block weight one and the weighted blocks under structured GHOST protocol have block weight $\heavyw$. So the expected block weight does not vary while switching the strategies. In normal scenarios, the GHAST consensus protocol adopts the optimistic strategy. When a serious liveness attack happens, the GHAST consensus protocol switches to the conservative strategy. All the block headers include an immutable strategy bit to indicate the strategy it adopts to make the miners reach a consensus for its block weight.

\paragraph{Enforced strategy choices.}
We found that if an attacker can fill the strategy bit arbitrarily, the confirmation delay will be much worse than our expectations. So the GHAST consensus protocol determines the strategy bit of each block based on its ``past graph''. The \emph{past graph} of a block refers to the set of all its reachable blocks following the parent edges and reference edges recursively. When an honest miner generates a new block, its current local Tree-Graph is the same as the past graph of this new block. If the past graph reflects an liveness attack is happening, the block should follow the conservative strategy. Otherwise, it should follow the optimistic strategy.

We define a concrete rule to decide the consensus strategy from the past graph. A block whose strategy bit is inconsistent with its past graph will be regarded as an invalid block and dropped by all the honest participants. So the consensus strategy choices are enforced by the consensus protocol. Since the strategy bit can be inferred from its past set, it can be omitted from the block header. 

Notice that an attacker can still manipulate the strategy bit by ignoring some blocks in its past set. But its ability to delay block confirmation can be significantly reduced. 

\paragraph{Detecting liveness attack.}

The GHAST consensus protocol provides a deterministic algorithm to detect if there is an active liveness attack given a past set and decide the consensus strategy.
The GHAST consensus protocol detects liveness attack following one idea: 
\begin{align*}
    &\textit{Whether there exists an old enough block in the branch selected by GHOST rule,} \\ 
    &\textit{its best child doesn't have a dominant advantage in subtree weight compared to its sibling blocks.}
\end{align*}


Recalling that GHOST protocol selects the pivot chain by picking the child with maximum subtree weight recursively. In the Tree-Graph in an honest participant's view, if a block in the pivot chain doesn't have a child with a dominant advantage in subtree weight, other participants may have a different opinion in picking the next block in pivot chain. If such a block has been generated for a long time, we suspected that a liveness attack is happening.

\paragraph{Partially-synchronized clock}

Pass et al~\cite{pass2017analysis} mentions that the Bitcoin protocol can be used as a partially-synchronized clock. And their subsequent works Fruitchain \cite{Fruitchain} and Hybrid Consensus \cite{Hybrid} show two examples that use Bitcoin protocol as a fundamental service. The GHAST protocol also runs a stand-alone blockchain following Bitcoin protocol, which we called \emph{the timer chain}. Each block in Tree-Graph structure includes the hash value of the longest branch leaf block in the timer chain. The height of the included timer block represents an imprecise timestamp. Given a local Tree-Graph and a block in this Tree-Graph, if the timestamp difference between the given block and maximum timestamp in the local Tree-Graph exceeds a threshold $\timerdiff$, we regard such block as an old enough block in the Tree-Graph. 

Notice that the GHAST consensus protocol only uses the timer chain to decide whether a block is old enough. The order of Tree-Graph blocks does not rely on their timestamp. The blocks confirmation in Tree-Graph does not need to wait for the confirmation of their timer block. In a normal scenario, a block in Tree-Graph is usually confirmed earlier than its timer block.

\paragraph{Embedding timer chain into Tree-Graph.} 

Some consensus protocols have multiple proof-of-work tasks. For example, the GHAST consensus protocol has two proof-of-work tasks: mining a Tree-Graph block and mining a timer block. In a parallel chain protocol like OHIE~\cite{yu2018ohie} and Prism~\cite{Prism19}, each individual chain has a proof-of-work task. Usually, these protocols require that an attacker can not obtain majority computing power in each proof-of-work task. In order to prevent the attacker from concentrating its computing power on one task, a widely used trick makes the participants work on all the proof-of-work tasks simultaneously. It constructs a block that includes the components (or the digests) of all the tasks. When a block is successfully mined, its hash value decides the block type. 

Following this trick, the timer block and Tree-Graph block have the uniform block format, each of which includes a parent edge, several reference edges, the hash value of the last timer block and other metadata such as transactions digest in the application. The GHAST consensus protocol also regards the timer block as a valid Tree-Graph block. 
\subsection{Related work}



\paragraph{Nakamoto consensus.}
Nakamoto consensus~\cite{bitcoin} is the first blockchain protocol. In Nakamoto Consensus, each block has one predecessor block and all blocks form a tree rooted at the genesis block. 
Pass et al~\cite{pass2017analysis} build a round based analysis framework for Nakamoto consensus in an asynchronized network model with a prior known maximum network delay $\delay$. Given the adversary computing power threshold $\beta$, network delay $\delay$ and block generation rate $\lambda$, they show several properties of Nakamoto consensus when $1-\lambda\delay>\beta/(1-\beta)$.

Some other blockchain systems like LiteCoin, Bitcoin Cash, Bitcoin Gold and Bitcoin SV tried to increase the throughput by tuning system parameters in Nakamoto consensus. However, Sompolinsky et al~\cite{GHOST} give the tradeoff between increasing throughput of Nakamoto consensus and security threshold $\beta$. The Nakamoto consensus has two parameters related to throughput: the block generation rate $\lambda$ and block size $s$. The throughput of Nakamoto consensus is bounded by $\lambda \cdot s$. In a network with limited bandwidth $b$, the propagation delay $\delay$ is lower bounded by $s/b$. Since the previous analysis requires $1-\lambda\delay>(1+\delta)\beta/(1-\beta)$, the throughput is upper bounded by $\frac{(2-\beta)\cdot s}{1-\beta}$. 


\paragraph{FruitChain.} 
%
%
FruitChain~\cite{Fruitchain} organizes blocks hierarchically to decouple these functionalities. It packs transactions first into
fruits (i.e., micro blocks) and then packs fruits into blocks (i.e., macro blocks). Both types of blocks are required solutions for proof-of-work puzzle. But only the macro blocks are maintained following Nakamoto consensus. 


The mining rewards in FruitChain are mainly distributed via micro blocks. It mitigates some problems like selfish mining~\cite{selfishmining}. 
In a selfish mining attack, the attacker manipulates the longest chain by withholding its newly generated block accordingly to increase the ratio of its block in the blockchain. This makes the attacker receive more mining reward. In FruitChain, since the mining rewards are distributed according to micro blocks, it is no need to manipulate the macro blocks and the attacker can not apply this apply strategy over micro blocks because they are not maintained by Nakamoto consensus. 

The drawback of FruitChain is that the block confirmation is built on the top of macro blocks, which follows Nakamoto consensus. The micro blocks become irreversible only if the macro block packing it is confirmed. So the confirmation delay of FruitChain is as worse as Nakamoto consensus. 

\paragraph{Bitcoin-NG}
Bitcoin-NG~\cite{Fruitchain} also organizes blocks hierarchically. In Bitcoin-NG, the key blocks (i.e., macro blocks) are organized following Nakamoto consensus. Once a miner generates a macro block, it is allowed to generate a sequence of light blocks (i.e., micro blocks) until the next miner generates a macro block. Each time a miner trying to generate a macro block, it should try to include all the micro blocks generated by the owner of the previous macro block. Similar with FruitChain, the micro blocks are not taken into consideration in branch selection of macro blocks. Bitcoin-NG has the same drawback as FruitChain since its security is also built on the Nakamoto Consensus. 

\paragraph{Hybrid consensus}
Hybrid consensus~\cite{Hybrid} extends the idea in Bitcoin-NG. In Bitcoin-NG, a leader is chosen periodically based on the mining of macro blocks. Hybrid consensus picks a small quorum from the miner of macro blocks. Unlike Bitcoin-NG, a miner is not included in the quorum at the time of its generation. Recalling that in Nakamoto consensus, the longest branch truncating the last $k$ blocks has consistent property. So Hybrid consensus picks quorum from the truncated branch. As the blocks in the truncated branch become irreversible, the chosen quorum will not change. This is different with Bitcoin-NG. The quorum runs a PBFT protocol to commit transactions. 
%

In Hybrid consensus, the security threshold $\beta$ drops to $1/3$ to guarantee the attacker can not control more than 1/3 nodes in a selected quorum, which is the requirement of PBFT protocol. Hybrid also requires honestness has some stickiness, i.e., it takes a short while for an adversary to corrupt a node. So the adversarial can not corrupt the whole quorum instantly once a quorum is selected. However, such an assumption shows that the selected quorum is the single point of failure for the whole consensus protocol. If an attacker continues to DDoS attack the newly selected quorum, Hybrid consensus protocol will crash. 

\paragraph{OHIE}
Another approach in increasing the throughput is running several parallel chains. In OHIE~\cite{yu2018ohie}, the participants mines on the hundreds of parallel chains simultaneously. When mining a block in OHIE, the miner needs to include the parent block hash of the current block in each chain. Once a valid proof-of-work puzzle is solved, the block hash determines which chain the new block belongs to. The parent blocks for each chain are selected following Nakamoto consensus. Each individual chain has a low block generation rate to match the security requirement in Nakamoto consensus. 
The parallel chain remains a low block generation rate for security and all the chains achieve a high throughput collaboratively. 

However, such design increases the cost in metadata extremely. In order to achieve desirable performance, OHIE runs 640 parallel chains and generates 64 blocks per second. In security analysis of OHIE, a block will be confirmed in OHIE only if it is confirmed in the individual chain it belongs to. So its confirmation time is worse than the Nakamoto consensus. 
\paragraph{Prism}
Prism~\cite{Prism19} also runs parallel chains. Unlike OHIE, the parallel chains in Prism do not carry transactions. So Prism does not need to order the blocks in parallel chains. Prism has three types of blocks: transaction blocks that only pack transactions (like fruit in FruitChain), proposer blocks that pack transaction blocks and voter blocks that run in parallel chains. The voter blocks will vote for the proposer blocks and pick a leader block for each height. Prism orders the leader blocks according to their height. The leader block not necessarily appears in the longest branch. So the block confirmation in Prism does not depend on the confirmation of proposer blockchain. 

The most clever point in Prism is that the confirmation of a leader block does not need to wait for its voters become irreversible. Though the delay for one voter block becomes irreversible is as worse as Nakamoto consensus, Prism claims that reverting a majority of voter chains at the same time is much more difficult than reverting one voter chain. So even if a few voter chains are reverted, as long as the leader block receives a majority votes, it is not reverted by the attacker. Prism is the first proof-of-work consensus protocol that breaks confirmation barrier in Nakamoto consensus. Our work has a better performance than Prism. 

Prism still has some drawbacks. Similar to OHIE, parallel chains increase the amount of metadata in the protocol significantly. Prism only provides a security analysis in a synchronized network model, which is doubted unrealistic by~\cite{pass2017analysis}. 


\paragraph{GHOST}
Since the throughput in Nakamoto consensus is upper bounded by security threshold $\beta$, Sompolinsky et al~\cite{GHOST} introduce GHOST, which uses another way to select the branch. \footnote{Though the structured variants of Nakamoto consensus resolves the throughput issue, GHOST is proposed earlier than their work. }. Instead of measuring the length of branches, GHOST defines the subtree weight to measure the number of blocks in the subtree rooted at each block. For each block, GHOST regards its child block with maximum subtree weight as the best child and breaks ties by block hash. Started with the genesis block, GHOST visits the best child recursively to select the branch. GHOST claims that once all the honest nodes mining under the subtree of one block, the growth of its subtree weight will not be undermined on the decreasing of block generation interval. Thus GHOST claims it resolves the security issue of Nakamoto consensus in a high block generation rate (a low block generation interval). 
However, Sompolinsky et al only analysis the behavior of GHOST under some attack cases, lack of a rigorous security analysis with a practical result. 


Kiayias et al~\cite{kiayias2017trees} provides a security analysis for GHOST in a low block generation rate under a synchronized network model. When the block generation rate in GHOST is low enough to matches the requirement in Nakamoto consensus. GHOST has the same security properties as Nakamoto consensus and can have the same confirmation. Kiffer et al~\cite{kiffer2018better} try to provide a similar analysis under an asynchronized network model. 



\paragraph{Liveness attack for GHOST}

Natoli et al~\cite{natoli2016balance} first point out GHOST is vulnerable facing a liveness attack when the block generation rate is high. The liveness attack is not aimed to re-ordering the confirmed blocks, but tries to prevent from confirming the new blocks. 

They provide a liveness attack strategy called the balance attack. Suppose the total block generation rate is $\lambda$ and the attacker is able to delay the message communication with time $\delay$. The attacker splits the honest miner into two groups with similar mining power. 
Figure~\ref{fig:examples:balance-attack} presents one example of such attacks.
The example has the following settings: 1) the total block generation rate of
honest participants is $\lambda$; 2) honest participants are divided into two
groups with equal computation power (group \phvname{X} and group \phvname{Y} in
Figure~\ref{fig:examples:balance-attack}); 3) blocks will transmit instantly
inside each group, but the propagation between these two groups has a
delay of $d$. In Figure~\ref{fig:examples:balance-attack}, each of the two
groups extend their own subtree following the GHOST rule. Note that recent
generated blocks within the time period of $d$ are in-transit blocks (gray 
blocks in Figure~\ref{fig:examples:balance-attack}), which are only
visible by the group that generates it. Therefore each group will believe its
own subtree is larger until one group generates sufficiently more blocks than
the other to overcome the margin caused by the in-transit blocks.
In normal scenarios, one of the two groups will get lucky to enable the
blockchain to converge. However, an attacker can mine under two subtrees
simultaneously to delay the convergence. The attacker can strategically
withhold or release the mined blocks to maintain the balance of the two
subtrees as shown in Figure~\ref{fig:examples:balance-attack}. Previous work
has shown that, if the margin caused by in-transit blocks is significant, i.e.,
$\lambda d>1$, an attacker with little computation power can stall the
consensus progress~\cite{yu2018ohie}.

\begin{figure}
    \centering
    \includegraphics[width=0.5\linewidth]{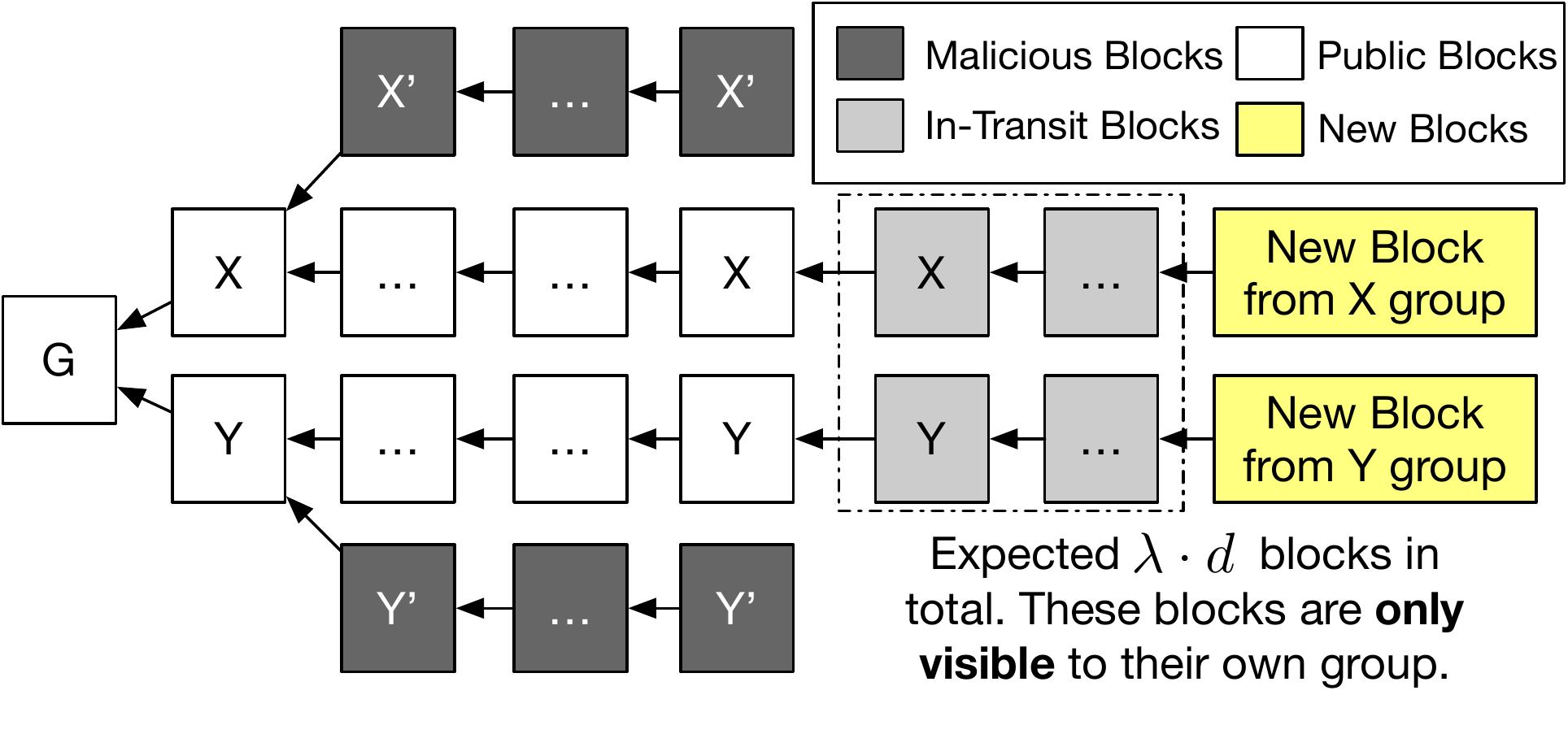}
    \caption{An example for balance attack.}
    \label{fig:examples:balance-attack}
    \vspace{-8mm}
\end{figure}


\paragraph{DAG-based structures} 
To improve the throughput and the confirmation speed, researchers have
explored several alternative structures to organize blocks. Inclusive
blockchain~\cite{inclusive} extends the Nakamoto consensus and GHOST to DAG and
specifies a framework to include off-chain transactions. In
PHANTOM~\cite{PHANTOM}, participating nodes first find an approximate
$k$-cluster solution for its local block DAG to prune potentially malicious
blocks. They then obtain a total order via a topological sort of the remaining
blocks. Unfortunately, when the block generation rate is high, inclusive
blockchain and PHANTOM are all vulnerable to liveness attacks. Unlike {\name} they
cannot achieve both the security and the high performance.


Some protocols attempt to obtain partial orders instead of total orders for
payment transactions.  SPECTRE~\cite{SPECTRE} produces a non-transitive
partial order for all pairs of blocks in the DAG.  Avalanche~\cite{Avalanche}
connects raw transactions into a DAG and uses an iterative random sampling
algorithm to determine the acceptance of each transaction. Unlike {\name}, it
is very difficult to support smart contracts on these protocols without
total orders.

\paragraph{Byzantine fault tolerance.}
ByzCoin~\cite{ByzCoin} and Thunderella~\cite{Hybrid} propose to achieve
consensus by combining the Nakamoto consensus with Byzantine fault tolerance
(BFT) protocols. Algorand~\cite{Algorand}, HoneyBadger~\cite{HoneyBadgerBFT},
and Stellar~\cite{stellar} replace the Nakamoto consensus entirely with BFT
protocols. In practice, all these proposals run BFT
protocols within a confined group of nodes, since BFT protocols only scale up
to dozens of nodes. The confined group is often chosen based on their recent
PoW computation power~\cite{ByzCoin, Hybrid}, their stakes of the
system~\cite{Algorand}, or external hierarchy of trusts~\cite{HoneyBadgerBFT,
stellar}. 

However, these approaches may create undesirable hierarchies among participants
and compromise the decentralization of blockchain systems. Moreover, all of
these approaches except Algorand are also vulnerable to DDoS attacks adaptively
targeting those leader or committee nodes. Algorand is vulnerable to long range
attacks --- an attacker could use a set of old private keys that once hold the
majority of coins to create an alternative transaction history that is
indistinguishable from the real history for new nodes.

\section{Model}\label{sec:model}

We adapt round-based partial synchronous network model similar to \cite{pass2017analysis}. A blockchain protocol is defined as a pair of algorithms $(\proto^{\vec{\eta}}(1^\secp),\order)$ with security parameter $\secp$. $\proto^{\vec{\eta}}$ is parameterized by a list of protocol parameters $\vec{\eta}(\kappa)$ and we use $\proto$ when the context is clear. It maintains the local state $\state$ consists of a set of blocks and prepares new blocks to be resolved proof-of-work puzzle. 
$\order$ orders the blocks in $\state$ deterministically.

The model is directed by an environment $\env(1^\secp)$ with security parameter $\secp$ which interacts with an adversary $\adv$ and a set of participant nodes $\cal N$ activated by $\env$. $\cal N$ contains two types of nodes: the \emph{honest nodes} which follow the blockchain protocol $(\proto,\order)$ and the \emph{corrupted nodes} which are controlled by adversary $\adv$. There is a random function $\prf:\{0,1\}^*\rightarrow\{0,1\}^\secp$ which can be accessed by participant nodes via four oracles $\mine(x):=\prf(x)$, $\verify(x,y):=[\prf(x)=y]$, $\prf^{\sf weight}(x)$ and $\prf^{\sf timer}(x)$. ($\prf^{\sf weight}$ and $\prf^{\sf timer}$ are used to assign each block two random tags. ) The output of $\prf$ is interpreted as an integer in $[0,2^\secp-1]$. Each node (honest or corrupted) is allowed to query oracle $\mine$ once each round.  

\paragraph{Round-based execution}
This protocol proceeds in round to model an atomic time steps (e.g. $10^{-12}$ seconds). For convenience in the security analysis, we re-order the actions of participants and divide them into four phases:
\begin{itemize}[nosep]
	\item \emph{Phase 1}: $\adv$ \emph{corrupts} and \emph{uncorrupts} arbitrary nodes in $\cal N$. It means $\adv$ can switch the corrupted nodes set adaptively between rounds. 
	\item \emph{Phase 2}: $\adv$ delivers blocks to each node. $\env$ delivers other messages (e.g. the contents to be recorded on blockchain) to each node. 
	\item \emph{Phase 3(a)}: For honest nodes, they maintain the local state with input delivered by $\adv$ and $\env$ in phase 2, organize a block to be solved proof-of-work puzzle following $\proto$ and try to solve the puzzle by querying oracle $\mine$. If an honest node constructs a valid block, it delivers to $\adv$ and incorporates the new block into its local state. 
	\item \emph{Phase 3(b)}: For corrupted nodes, $\adv$ gets access to their local state and takes control of access for oracle $\mine$, meaning that adversary is allowed to query oracle $\mine$ with quota the number of corrupted nodes. 
\end{itemize}
Let random variable $\view^{(\proto,\order)}(\env,\adv,\secp)$ denote the joint view of all the participant nodes and the adversary in all rounds.

\paragraph{Block and graph}
A \emph{block} $\block$ is a tuple of $(h_{-1},\vec{h},{\sf m},s,h)$, where $h_{-1}$ is a hash value ($\secp$-bits string) of a previous block (a pointer to this block), $\vec{h}$ is a list of hash values to some other blocks, ${\sf m}$ represents the contents and metadata carried by such block and $s$ is a nonce. $h$ is the hash of block $\block$ satisfying $\prf(h_{-1},\vec{h},{\sf m},s)$. We use $\block.\digest$ to denote the hash $h$. The blocks corresponding to $h_{-1}$ or the hash values in $\vec{h}$ are called {direct dependency blocks} of block $\block$. Following the dependency relation recursively, the other reached blocks are called {indirect dependency blocks}. 

In order to limit the generation rate of blocks, $\proto^{\vec{\eta}}$ only accepts the block whose hash is smaller than $2^\secp/\difficulty$, where $\difficulty$ is the \emph{puzzle difficulty} parameter in $\vec{\eta}$.
In phase 3(a), an honest node following protocol $\proto^{\vec{\eta}}$ checks the validity of incoming blocks and incorporate the valid blocks into local state $\state$. After that, it prepares a new block in format of $(h_{-1},\vec{h},{\sf m},\bot,\bot)$ directed by protocol $\proto^{\vec{\eta}}$. The proof-of-work puzzle refers finding an appropriate nonce $s$ with $\prf(h_{-1},\vec{h},{\sf m},s)<2^\secp/\difficulty$ to make this block valid. The only way to solve puzzle is querying $\prf$ with random $s$ via oracle $\mine$. When an appropriate $s$ is found, we say a block is \emph{generated}. 

Protocol $\proto$ regards a block as valid if its first component $h\l$ is not $\bot$, its hash value is consistent with other components and it solves the proof-of-work puzzle. The \emph{genesis block} $\genesisblock\eqdef(\bot,\bot,\bot,\bot,\prf(\bot,\bot,\bot,\bot))$ is a special valid block which does not satisfy these properties.

A set of valid blocks $\graph$ will be regarded as a \emph{valid graph} if for each block $\block\in\graph$, its direct and indirect dependency blocks are in $\graph$. Since the hash value is unpredictable before querying $\prf$, we can simply claim $\prf$ must output the hash of block $\block$ later than all its direct and indirect dependency blocks. So there should be no cycle in $\graph$. 
The local state $\state$ is a valid graph $\graph$ and we use $\state$ and $\graph$ interchangeably. The following part only focuses on the valid blocks and valid graphs. So we omit term ``valid'' for succinctness.

The blocks constructed by honest nodes are called \emph{honest blocks} and the blocks constructed by adversary are called \emph{malicious blocks}. In our execution model, honest nodes are not informed whether a block is honest or malicious. 

\paragraph{Adversary restriction}
We discuss the model with a restricted environment and adversary.

\begin{definition}[Admissible environment]\label{def:admissible}
	We say that the tuple $(m(\cdot), \beta, d(\cdot), \adv, \env)$ is admissible w.r.t. $(\proto,\order)$ if $\beta<1/2$, $\adv$ and $\env$ are non-uniform probabilistic polynomial-time algorithms, $m(\cdot)$ and $d(\cdot)$ are polynomial functions, and for every $\secp\in \mathbb{N}$ 
	\begin{itemize}[nosep]
		\item $\env$ activates $m=m(\secp)$ participant nodes;
		\item $\adv$ does not modify the contents of delivered message;
		\item $\adv$ always corrupts $\beta\cdot m(\secp)$ corrupted nodes at the same time. \footnote{We assume $\beta\cdot m(\secp)$ always be an integer here. This setting also handle the case that adversary $\adv$ corrupts nodes less than $\lfloor\beta\cdot m(\secp)\rfloor$, because adversary $\adv$ can make some corrupted nodes act like an honest node.}
		\item For any block $\block$, if it appears in local state of one honest node in round $r$, $\adv$ is responsible to make sure all the honest nodes incorporate $\block$ to the local states at and after phase 2 of round $r+d(\secp)$. 
	\end{itemize}
\end{definition}



\paragraph{Metrics} 
Here we discuss the aim of our protocol $(\proto,\order)$. Recalling that random variable $\view^{(\proto,\order)}(\env,\adv,\secp)$ denote the joint view of all the participant nodes in all rounds. The randomness is from the oracle $\prf(\cdot)$ and random coins in $\env,\adv$ and participant nodes. Let $\uniongraphs_r$ collect all the local states of honest nodes in round $r$ (in each phase). Similar with \cite{Prism19}, the protocol executes for a finite round $\rmax$ polynomial in $\secp$. 

Our study focus on the finality of block history. The \emph{history of block} $\block$ in local state $\state$ refers the prefix \footnote{In this paper, a \emph{prefix} of a list could equal to the list itself.} of $\order(\state)$ ended at block $\block$, which is denoted by $\mathrm{Prefix}(\order(\state),\block)$. If $\block\notin\order(\state)$, $\mathrm{Prefix}(\order(\state),\block)=\bot$. Block $\block$ is finalized (or confirmed) if all the honest nodes have consistent history of block $\block$ remain unchanged. Formally, we define $(\varepsilon,\adv,\env,r_0,\secp)$-finalized as follows.

\begin{definition}[Finalization]\label{def:finalization}
	Let $\graph_r$ denote the joint local state of all the honest nodes at round $r$ (after phase 2) 
	in $\view^{(\proto,\order)}(\env,\adv,\secp)$. Round $\rcon$ is $(\varepsilon,\adv,\env,r_0,\secp)$-finalized w.r.t. protocol $(\proto,\order)$ iff 
	$$ \Pr_{\view^{(\proto,\order)}(\env,\adv,\secp)}\left[\forall \block\in \graph_\rcon,
		\left|\bigcup\nolimits_{\substack{r\in\{r_0,\cdots,\rmax\}\\ \state\in \uniongraphs_r}} \mathrm{Prefix}(\order(\state),\block)\right|=1
	\right]\ge 1-\varepsilon-\negl(\secp)$$
\end{definition}
Since the local state of each honest node is a random variable, the finality is defined over a round $\rcon$ other than a block $\block$.  

\begin{definition}[Latency]
	If there exists $r_{\varepsilon}$ such that for any $\rcon\le \rmax-r_{\varepsilon}$, round $\rcon$ is $(\varepsilon,\adv,\env,\rcon+r_{\varepsilon},\secp)$-finalized, then we say $\view^{(\proto,\order)}(\env,\adv,\secp)$ has the $\varepsilon$-latency $r_{\varepsilon}$.
\end{definition}


\section{Protocol}\label{sec:proto}



\subsection{Rephrase GHOST protocol in our framework}\label{sec:ghost}

GHOST proposed in \cite{GHOST} takes a set of blocks in a tree structure as input and outputs a sequence of blocks. Each block in this tree has a non-negative \emph{subtree weight}. For every block $\block$, the subtree weight of $\block$ refers to the total weights of all blocks in the subtree rooted at $\block$. The GHOST starts from the root of the tree and repeatedly proceeds to the child block with maximum subtree weight until reaching a leaf node block. Then the path of blocks will be the chain output by GHOST.

Now, we formalize the GHOST \cite{GHOST} with our notations. Each block under GHOST doesn't have the component of block hash values list. They can be represent by $\block=(h_{-1},\bot,{\sf m},s,h)$. The genesis block $\genesisblock$ is the root of tree in GHOST. For any other valid block $\block$, $h_{-1}$ should be the digest of a valid predecessor block $\block_{-1}$, which is called \emph{parent block} of $\block$. Since two different blocks never have the same digest with negligible exception, we denote the parent block by $\block.\parent$. (For genesis block, $\genesisblock.\parent=\bot$.) Started at any block, following the parent block recursively gives a chain of blocks ended as the genesis block. Every two consecutive blocks in this chain have a parent/child relation. It is called \emph{chain of block} $\block$ and defined as 
\begin{equation}\label{eq:def:chain}
	\mathrm{Chain}(\block):=\left\{\begin{array}{ll}
		\block & \block=\genesisblock  \\
		\mathrm{Chain}(\block.\parent)\circ \block  & \text{otherwise}
	\end{array}\right.
\end{equation}

In a graph $\graph$, each block has exactly one outgoing edge except the genesis block with no outgoing edge and there is no cycle because of unpredictable of digest computation. So all the blocks organize in a tree rooted at $\genesisblock$. We use $\tree{\graph,\block}$ to denote the subtree rooted at block $\block$ in $\graph$. 
\begin{equation}\label{eq:def:subtree}
	\tree{\graph,\block}:=\{\block'\in \graph:\block'\in \chain{\block}\}.
\end{equation}

The subtree weight for each block $\block\in\graph$ refers the total block weight of blocks $\tree{\graph,\block}$. In GHOST protocol, all the blocks have the same weight 1 ($\forall \block\in \graph, \block.\weight=1$) \footnote{A blockchain system uses puzzle difficulty as block weight. Since our model has a static puzzle difficulty, we simply set block weight be 1. } and the subtree weight is formulated as
\begin{equation}\label{eq:def:subtreew}
	\treew{\graph,\block}=\sum_{\block'\in \tree{\graph,\block}} \block'.\weight.
\end{equation}

The \emph{children of block} $\block$ refers all the blocks $\block'$ which regard $\block$ as its parent block. In addition, we filters out the blocks with weight 0. This rule is not activated in GHOST protocol since all the block have weight 1. But the following design will introduce zero-weight block.
\begin{equation}\label{eq:def:children}
	\child(\graph,\block):=\{\block'\in \graph:\block'.\parent=\block\wedge\block'.\weight>0\}.
\end{equation}

Among all the children block, GHOST chooses the one with the largest subtree weight and break tie by choosing the block with minimum block hash. Formally, it can be described by 
\begin{equation}\label{eq:def:bestchild}
		\mathrm{BestChild}(\graph,\block):=\argmax_{\block'\in \child(\graph,\block)} \treew{\graph,\block'}. 
\end{equation}
(Note: 1. When $\child(\graph,\block)=\emptyset$, let $\mathrm{BestChild}(\graph,\block)=\bot$; 2. When there are multiple children having maximum subtree weight, this function breaks tie by returning the block with minimum block digest. \footnote{The original work \cite{GHOST} didn't mention how to break ties when two subtree have the same weight. However, breaking tie with block digest is a common setting in previous work. \cite{pass2017analysis}}) 

Started with the genesis block $\genesisblock$, GHOST recursively choose the best child until reaching a block without children. All reached blocks organize a chain called \emph{pivot chain} of graph $\graph$. $\mathrm{Pivot}(\graph)$ is defined formally in figure~\ref{eq:def:chain}.

\begin{algorithm}[H]
	\small
	\SetNlSty{}{}{}
	\DontPrintSemicolon
	\SetKw{To}{ to }
	\SetKw{In}{ in }
	\SetKw{Or}{ or }
	\SetKw{And}{ and }
	\SetKwInOut{Input}{Input}\SetKwInOut{Output}{Output}
	\SetKw{Define}{define}
	\Input{A graph $\graph$}
	\Output{A sequence of blocks $\mathbf{L}$}
	\let\oldnl\nl
	\newcommand{\nonl}{\renewcommand{\nl}{\let\nl\oldnl}}
	Initialize $\mathbf{L}$ with empty list\;
	Initialize $\block$ with genesis block $\genesisblock$\;
	$\mathbf{L} \longleftarrow \mathbf{L}\circ \block$\;
	\While {$\mathrm{Child}(\graph,\block) \neq \emptyset$} {
		$\block \longleftarrow \mathrm{BestChild}(\graph,\block)$\;
		$\mathbf{L} \longleftarrow \mathbf{L}\circ \block$\;
	} 
	\Return $\mathbf{L}$
	\vspace{-2.5mm}
	\caption{\small The definition of $\mathrm{Pivot}(\graph)$.}
	\label{fig:pivot}
	\vspace{-4mm}
\end{algorithm}

Now we describe the protocol $(\proto^{\vec{\eta}}_{\sf GHOST},\order_{\sf GHOST})$. $\proto^{\vec{\eta}}$ first initiates the local state of all participant nodes with genesis block $\genesisblock$. In phase 3(a) of each round, upon receiving the blocks delivered by adversary $\adv$, $\proto^{\vec{\eta}}_{\sf GHOST}$ directs the honest nodes check their validity and append them into $\graph$ while making sure $\graph$ be valid. Then each honest node computes the pivot chain by $\mathrm{Pivot}(\graph)$ over their local state $\graph$, sets $h_{-1}$ be the digest of last block in $\mathrm{Pivot}(\graph)$, prepares block $\block_{\sf new}=(h_{-1},\bot,{\sf m},\bot,\bot)$ and tries to solve proof-of-work puzzle by querying $\prf(h_{-1},\bot,{\sf m},s)$ with random $s$. After that, the model enters phase 4. In ordering the graph $\graph$, order algorithm $\order_{\sf GHOST}(\state)$ simply returns the pivot chain of local state $\mathrm{Pivot}(\graph)$.

\subsection{Tree-graph structure}

We adopt the ideas from previous works \cite{SPECTRE,PHANTOM} which allow each block refers multiple predecessor blocks and organize blocks in the structure of directed acyclic graph instead of tree. A valid blocks can be represented by $\block = (h_{-1},{\vec h},{\sf m},s,h)$, in which ${\vec h}$ contains a list of digests (pointers) of other valid blocks. The blocks pointed by ${\vec h}$ are called \emph{reference blocks} of $\block$. (Genesis block $\genesisblock$ should have empty $\vec{h}$). Started with block $\block$, by following the reference blocks repeatedly, we can reach all the direct and indirect dependency blocks of $\block$. These blocks (not including block $\block$) organize a valid graph, which is called the \emph{past graph} of block $\block$ and denoted by $\block.\past$. 

The parent block digest $h_{-1}$ of all blocks organizes blocks in a tree structure and the reference block digests ${\vec h}$ organize blocks in a directed acyclic graph structure. So we call it Tree-graph structure and denote the protocol by $(\proto^{\vec{\eta}}_{\sf TG},\order_{\sf TG})$. Compared with $\proto^{\vec{\eta}}_{\sf GHOST}$, $\proto^{\vec{\eta}}_{\sf TG}$ has an additional requirement for block validity. It requires a valid block $\block$ chooses the pivot chain tip of $\block.\past$ as its parent block. So the chain of block $\block$ is consistent with the pivot chain in graph $\block.\past$.  (a.k.a. $\mathrm{Pivot}(\block.\past)\circ\block = \mathrm{Chain}(\block)$.) In organizing new blocks $
\block_{\sf new}=(h_{-1},{\vec h},\sf m,\bot,\bot)$ to be solved proof-of-work puzzle, $\proto^{\vec{\eta}}_{\sf TG}$ prepares $h_{-1}$ and $\sf m$ in the same way as $\proto^{\vec{\eta}}_{\sf GHOST}$, and includes the digests of tip blocks in graph $\state$ into ${\vec h}$ to make sure $\block_{\sf new}.\past=\state$. 

The ordering algorithm $\order_{\sf TG}$ is defined formally in figure~\ref{fig:order}. It initializes a list $\textbf{L}$ with only genesis block $\genesisblock$ and visits the blocks in the pivot chain sequentially. In each round, when the algorithm $\order_{\sf TG}$ reaches block $\block_{\sf next}$ with parent block $\block$, it collects all the blocks in $\block_{\sf \next}.\past$ but not appended to $\textbf{L}$ in the last round (not in $\block.\past\cup\{\block\}$), topological sorts them with a deterministic function $\mathrm{TopoSort}(\cdot)$ and appends the result to $\textbf{L}$. Next the algorithm appends $\block_{\sf next}$ to $\textbf{L}$ and continues to visit the next block. The detailed implementation of function $\mathrm{TopoSort}(\cdot)$ is not necessary in subsequent analysis.


\begin{algorithm}[!htp]
	\small
	\SetNlSty{}{}{}
	\DontPrintSemicolon
	\SetKw{To}{ to }
	\SetKw{In}{ in }
	\SetKw{Or}{ or }
	\SetKw{And}{ and }
	\SetKwInOut{Input}{Input}\SetKwInOut{Output}{Output}
	\SetKw{Define}{define}
	\Input{A valid graph $\graph$ (An alias of local state $\state$)}
	\Output{A sequence of blocks $\mathbf{L}$}
	\let\oldnl\nl
	\newcommand{\nonl}{\renewcommand{\nl}{\let\nl\oldnl}}
	$\mathbf{P} \longleftarrow \mathrm{Pivot}(\graph)$\;
	$\block \longleftarrow \text{Pull the first block from }\mathbf{P}$\;
	Assert $\block$ is genesis block $\genesisblock$\;
	$\mathbf{L}\longleftarrow \block$\;
	\While {$\mathbf{P}$ is not empty list} {
		$\block_{\sf next} \longleftarrow \text{Pull the first block from }\mathbf{P}$\;
		$\mathbf{L} \longleftarrow \mathbf{L}\circ \mathrm{TopoSort}(\block_{\sf next}.\past\backslash\left(\{\block\}\cup\block.\past\right))$\;
		$\mathbf{L} \longleftarrow \mathbf{L}\circ \block_{\sf next}$\;
		$\block\longleftarrow\block_{\sf next}$\;
	} 
	\Return $\mathbf{L}$
	\vspace{-2.5mm}
	\caption{\small The definition of $\mathrm{\order_{\sf TG}}(\state)$.}
	\label{fig:order}
	\vspace{-4mm}
\end{algorithm}

\subsection{GHAST}


%

Now we introduce the Ghast protocol $(\proto^{\vec{\eta}}_{\sf GHAST},\order_{\sf GHAST})$. It is the same as $(\proto^{\vec{\eta}}_{\sf TG},\order_{\sf TG})$ except the definition of block weight. In $(\proto^{\vec{\eta}}_{\sf GHAST},\order_{\sf GHAST})$, the protocol may assign different block weight to each block in order to solve the liveness issues in Ghost. \footnote{This design is totally different from the mechanism called \emph{difficulty adjustment} in Bitcoin. Difficulty adjustment mechanism changes the block weights and puzzle difficulty in react to the change of computing power (the frequency in querying oracle $\prf$). In the safety analysis model with constant computing power, Bitcoin and some other protocols like Ghost and Prism \cite{Prism19} doesn't adjust the puzzle difficulty and block weight. But Ghast protocol may assign different block weight to each block. } 
The following design requires four more parameters in $\vec{\eta}$: $\heavyw$, $\advan$, $\timerw$ and $\timerdiff$. 

\subsubsection{Assign block weight accordingly}

Previous work \cite{yu2018ohie} figures out an adversary with a low computing power in ratio (e.g. $\beta=0.2$) has the ability to defer the finality of a block when the block generation interval ($m/\difficulty$ rounds in average) is much smaller than the maximum message delay ($\delay$ rounds). In order to solve this problem, we propose a structured GHOST protocol which slows down the weighted block generation rate. We assign each block $\block=(h_{-1},\vec{h},{\sf m},s,h)$ a random tag computed by oracle $\prf^{\sf weight}(h):=\prf({\sf weight},h)$ where ${\sf weight}$ represents a fixed bit-string to make $\prf^{\sf weight}(h)$ be independent with the outputs from $\prf$ in solving proof-of-work puzzle. \todo{How to present it formaly.} Let $\heavyw>0$ be a parameter which specifies the ratio in slowing down generation rate. We set $\block.\weight=\heavyw$ when $\prf^{\sf weight}(h)<2^\secp/\heavyw$ and $\block.\weight=0$ otherwise. It means that a valid block will be assigned with weight $\heavyw$ with probability $1/\heavyw$. Thus the generation rate of the blocks with non-zero weight reduces $\heavyw$ times. 

%
When the block generation rate is low enough to solve issues in Ghost, the confirmation latency is as worse as nakamoto consensus. So the protocol switch between an optimistic consensus strategy and a conservative strategy accordingly. In the optimistic strategy, all the blocks have the same weight 1. When the protocol detects a serious attack happens, it switches to the conservative strategy in which only $1/\heavyw$ of blocks is assigned weight $\heavyw$. 

In a partial synchronous network, it is impossible to make all the honest nodes switch the settings simultaneously. So the protocol determines the consensus strategy for each block individually and incorporates the blocks generated in two strategies into one Tree-Graph. We set a deterministic function $\adp(\graph)$ which is responsible for detecting the presence of a liveness in a local state $\graph$. It takes the past-set of a block as input and outputs $\mathsf{opt}$ or $\mathsf{con}$ to indicate the consensus strategy for this block. For the block $\block$ with $\adp(\block.\past)=\mathsf{opt}$, $\block.\weight=1$. If $\adp(\block.\past)=\mathsf{con}$, the block weight equals to 0 or $\heavyw$ depending on $\prf^{\sf weight}(h)$. 

\subsubsection{Detect liveness attack}

Function $\adp(\cdot)$ is parameterized by a positive integer $\advan$ in $\vec{\eta}$.
%
We define a concept called $\advan-$\emph{dominant child}. In a given graph $\graph$, when block $\block$ has a child block $\block'$ whose subtree weight is at least $\advan$ larger than the subtree weight of the other blocks, we say block $\block'$ is $\advan-$dominant child of block $\block$ in graph $\graph$. Specially, if block $\block$ has only one child block, the subtree weight of its $\advan-$dominant child block should be at least $\advan$. Function $\adp(\cdot)$ starts with the genesis block $\genesisblock$, visits the $\advan-$dominant child repeatedly until reaching a block without $\advan-$dominant child. Let block $\block_{\sf c}$ be the last visited block. 

For another graph $\graph'$, if the total block weight of symmetric difference between graph $\graph$ and graph $\graph'$ is less than $\advan$, the pivot chain of $\graph$ and graph $\graph'$ must have the common prefix ended at $\block_{\sf c}$. \todo{Lemma XXX reason it formally.} Intuitively, for another honest node whose local state is not much different from $\graph$, it should also agree that block $\block_{\sf c}$ is in pivot chain. 
%

After getting the chain ended at $\block_{\sf c}$, function $\adp(\cdot)$ accesses a \emph{block age speculation} function $\mathrm{Old}(\graph,\block)$. (It is defined in~\ref{sec:blockage}.) It conjectures whether a block $\block$ is old enough (has been generated for a sufficient long time) at the time point that an honest node has local graph $\graph$. If $\mathrm{Old}(\graph,\block_{\sf c})$ output $\false$, we have relative high confidence that honest nodes have almost the same pivot chain and there is no liveness issue. So $\adp(\graph)$ outputs $\sf opt$ when $\mathrm{Old}(\graph,\block_{\sf c})=\false$ and outputs $\sf con$ otherwise.

The block age speculation function is required the following two properties: 
\begin{enumerate}
	\item  If a block is old enough, all the blocks in its past set should also be old enough. Formally, for any block $\block,\block'$ and graph $\graph$ with $\mathrm{Old}(\graph,\block)=\true$ and $\block'\in\block.\past$, it should be $\mathrm{Old}(\graph,\block')=\true$.
	\item  Once a block becomes old enough, it will always be old enough in the future. Formally, for any graph $\graph,\graph'$ and block $\block$ with $\mathrm{Old}(\graph,\block)=\true$ and $\graph \subseteq \graph'$, it will be $\mathrm{Old}(\graph',\block)=\true$. 
\end{enumerate}

Now we give a formal definition for function $\adp(\cdot)$. Let $\sibtreew{\graph,\block}$ returns the maximum subtree weight of the siblings of $\block$ (the other children of its parent block). 
\begin{equation}\label{def:sibsubtw}
    \sibtreew{\graph,\block}:=\max_{\block'\in \child(\graph,\block.\parent)\backslash\{\block\}}\mathrm{SubTW(\graph,\block')}.
\end{equation}
(If $\child(\graph,\block.\parent)\backslash\{\block\}$ is empty set, $\sibtreew{\graph,\block}$ returns 0.)

Notice that $\block$ is the $\advan-$dominant child iff $\mathrm{SubTW}(\graph,\block)-\sibtreew{\graph,\block}\ge \advan$. The function $\adp(\cdot)$ can be expressed in an equivalent form: 
\begin{definition}\label{def:adaptive}
	Function $\adp(\graph)$ outputs $\sf con$ if one of the following two conditions satisfied:
	\begin{itemize}[nosep]
		\item $\exists \block\in \mathrm{PivotChain}(\graph), \mathrm{Old}(\graph,\block.\parent)=\true \wedge \mathrm{SubTW}(\graph,\block)-\sibtreew{\graph,\block}< \advan.$
		\item For the last block $\block'$ in $\mathrm{PivotChain}(\graph)$, $\mathrm{Old}(\graph,\block')=\true$
	\end{itemize}
	For other cases, $\adp(\graph)$ outputs $\sf opt$. Specially, if $\graph=\emptyset$, $\adp(\graph)$ outputs $\sf opt$. The function $\mathrm{Old}(\graph,\block)$ is defined in definition~\ref{def:oldenough}.
\end{definition}

\subsubsection{Block age speculation function}\label{sec:blockage}

Block age speculation function $\mathrm{Old}(\graph,\block)$ parameterized by $\timerw,\timerdiff$ conjectures if a block has been generated for a long enough time when an honest node see graph $\graph$. 
Let graph $\graph$ be the local state of an honest node at round $r$. Informally, we want to achieve two properties with some integers $r_2<r_1$:
\begin{itemize}[nosep]
	\item If block $\block$ is generated by an honest node, $\mathrm{Old}(\graph,\block)$ outputs $\true$ if $\block$ is generated before round $r-r_1$ and outputs $\false$ if $\block$ is generated later than $r-r_2$ (with negligible exception).
	\item If block $\block$ is generated by adversary, $\mathrm{Old}(\graph,\block)$ outputs $\true$ if $\block$ is generated before round $r-r_1$ (with negligible exception).
\end{itemize}

Note that we allow $\mathrm{Old}(\graph,\block)$ to falsely think a newly generated malicious block has been withheld for a long time. Because in some attack cases, an honest node can not distinguish between a newly generated block and an old block. For example, the adversary generates block $\block_1$ when the blockchain system just launched, withholds block $\block_1$ for a long time, and then generates another block $\block_2$ with the same parent blocks and direct dependency blocks. Then block $\block_1$ and block $\block_2$ have the same past set and difference block age. \todo{A simple example with verbose description.} 

We start with an idea that the protocol runs another blockchain under nakamoto consensus \cite{bitcoin} as \emph{timer chain}. The previous work \cite{pass2017analysis} shows that the longest branch in nakamoto consensus excluding last $x$ blocks will be agreed by all the honest nodes (except with exponentially small probability in $x$), and the growth rate on chain length has a lower bound and an upper bound. So the timer chain excluding last several blocks grows steadily and never be reverted. A block in tree graph structure can indicate its earliest possible generation time by including a block hash in timer chain. Counting the block height difference between the included block and the newest block in the timer chain, we can speculate the age of a block. 

The timer chain is constructed by picking a sequence of blocks in tree-graph structure. Formally, we assign each block $\block=(h_{-1},\vec{h},{\sf m},s,h)$ another random tag computed by oracle $\prf^{\sf timer}(h):=\prf({\sf timer},h)$. The blocks with $\prf^{\sf timer}(h)<\timerw\cdot 2^\secp$ are called \emph{timer blocks}. Let $\mathrm{Timer}(\graph)$ denote all the timer blocks in $\graph$. For each timer block $\block$, we defines its height in timer chain as
\begin{equation}\label{eq:timerheight}
	\mathrm{TimerHeight}(\block):=\max_{\block'\in\mathrm{Timer}(\block.\past)}\mathrm{TimerHeight}(\block')+1.
\end{equation}
Specially, if $\block.\past$ has no timer block, let $\mathrm{TimerHeight}(\block)=1$. For a graph $\graph$, $\mathrm{MaxTM}(\graph)$ returns maximum timer height. 
\begin{equation}\label{eq:maxtimerheight}
	\mathrm{MaxTH}(\graph):=\max_{\block'\in\mathrm{Timer}(\graph)}\mathrm{TimerHeight}(\block').
\end{equation}

\begin{definition}[Block age speculation function]\label{def:oldenough}
	Given graph $\graph$ and block $\block$, $\mathrm{Old}(\graph,\block)$ outputs $\true$ iff 
	$$\mathrm{MaxTH}(\graph) - \mathrm{TimerHeight}(\block)\ge \timerdiff.$$
\end{definition}

Note that function $\mathrm{Old}(\graph,\block)$ is also well-defined even if $\block\notin \graph$,

\subsubsection{Summary for the GHAST protocol}

Formally, the protocol $(\proto_{\sf GHAST},\order_{\sf GHAST})$ is the same as $(\proto_{\sf TG},\order_{\sf TG})$ excepts the way in assigning block weights. Given function $\adp(\graph)$ defined in definition~\ref{def:adaptive}, the block weights in the GHAST protocol is defined as

\begin{equation}\label{eq:blockweight}
	\block.\weight:=\left\{\begin{array}{ll}
		1      & \adp(\graph)={\sf opt} \\
		\heavyw & \adp(\graph)={\sf con} \wedge \prf^{\sf weight}(\block.\digest)<2^\secp/\heavyw \\
		0	   & \adp(\graph)={\sf con} \wedge \prf^{\sf weight}(\block.\digest)\ge 2^\secp/\heavyw
	\end{array}	\right.
\end{equation}






\section{Security}

Now we analysis the security of Ghast protocol. Formally, given $\vec{\eta}:=(\difficulty,\heavyw,\advan,\timerw,\timerdiff)$, we analysis the finalization for $(m,\beta,\delay,\adv,\env)$ which is admissible w.r.t. $(\proto_{\sf GHAST}^{\vec{\eta}},\order_{\sf GHAST})$. 

\subsection{Skeleton of proofs}

We let the adversary $\adv$ maintains an additional \emph{adversary state} $\advs$ for security analysis only. The adversary state is updated only on the following four types of events:
\begin{itemize}[nosep]
	\item \emph{Honest block generation (denoted by $\hgen$).} An honest node constructs a valid block $\block$, incorporates it into local state and sends it to adversary $\adv$. 
	\item \emph{Malicious block generation (denoted by $\mgen$).} The adversary constructs a valid block $\block$. 
	\item \emph{Malicious block release (denoted by $\mrls$).} For a block $\block$ constructed by the adversary, the first time it appears in the local state of an honest node. 
	\item \emph{Block received by all honest nodes (denoted by $\allrec$).} If a block $\block$ is incorporated to the local state of an honest node at round $r$, an event happens at round $r+d$ which guarantees block $\block$ appears the local state of all the honest nodes after this time. 
\end{itemize}

We use a tuple $\event:=(r,\block,t)$ with $t\in \{\hgen,\mgen,\mrls,\allrec\}$ to denote a type $t$ event happens at round $r$ corresponding to block $\block$. A function $\psi(\advs\l,\event)=\advs$ directs adversary updates $\advs\l$ to $\advs$ when event $\event$ happens.


We define a \emph{global potential value} $\cpot(\advs,\graph)$ to quantify the adversary's power in changing the history of blocks in $\graph$. 
We show that the history of blocks in $\graph$ will be consistent among all the honest participants become unchangeable unless $\cpot(\advs,\graph)$ larger than a threshold in the future in theorem~\ref{thm:hisunchange}.
In other words, for any graph $\graph$, the finalization of blocks in graph $\graph$ can be reduced to the upper bound for $\cpot(\advs,\graph)$.

We impose an \emph{event value} on every events to upper bound its influence on global potential value.~(Theorem~\ref{thm:cpotdiff}.) The sum of event values naturally implies an upper bound for the global potential value. So the finalization property can be derived from the random variable for the sum of event values. 

We will introduce the intuitions in designing adversary state update function $\psi(\advs\l,\event)$, potential function $\cpot(\advs,\graph)$, and event value in section~\ref{sec:intuition}. Then we will define these functions formally in section~\ref{sec:concepts}. It is inconsistent with the intuition in some details since the intuitions are illustrative. Section~\ref{sec:properties} proves that our design achieves our aim in proof skeleton.

\subsection{Intuitions}\label{sec:intuition}

\paragraph{Common pivot chain} 
Recalling that the ordering algorithm $\mathrm{\order_{\sf GHAST}}(\state)$ is the same as $\mathrm{\order_{\sf TG}}(\state)$ except the definition of block weight. From the definition of $\mathrm{\order_{\sf TG}}(\state)$ in figure~\ref{fig:order}, the block order is only determined by the pivot chain $\mathrm{Pivot}(\state)$. (We use $\graph$ and $\state$ interchangeably.) Given two graphs $\graph_1,\graph_2$, if block $\block$ appears in the pivot chain of graph $\graph_1$ and graph $\graph_2$, block $\block$ and the blocks in its past-set $\block.\past$ must have the same history in graph $\graph_1$ and $\graph_2$. Because the execution of $\mathrm{\order_{\sf TG}}(\graph_1)$ and $\mathrm{\order_{\sf TG}}(\graph_2)$ are the same before block $\block$ appended to $\mathbf{L}$. So the history of a block will remain unchanged if it is kept in the pivot chain. 

A pivot chain block $\block$ will be kept in the pivot chain only if its subtree weight is always no less than the its sibling blocks. We notice that when all the honest nodes are trying to construct block in the subtree of block $\block$, the subtree weight of $\block$ will increase faster than all its siblings in expectation, no matter what adversary $\adv$ does. To study this case, we define \emph{common pivot chain} which refers the intersection of the local state pivot chains of all honest nodes. As long as a block $\block$ lies on the common pivot chain, all the newly generated honest blocks will fall into the subtree of $\block$ and contribute to the subtree weight of $\block$. So block $\block$ will accumulate subtree weight advantages compared to siblings blocks. Intuitively, if a block $\block$ can stay in a common pivot chain for a sufficient long time, block $\block$ will be finalized. 

\paragraph{Special status} 
Some known attack can make the honest nodes disagree on choosing the best child for an old pivot chain block. In section~\ref{sec:ghost}, we try to handle this attack by switching to a conservative setting. Here, we want all the honest nodes switch to the conservative setting when the last block of common pivot chain is old enough. Formally, let $\mathrm{P}$ be the common pivot chain, $\block_{\sf tip}$ be the last block of the common pivot chain, $\state$ be the local state of an honest node.  We want $\adp(\state)={\sf con}$ when $\mathrm{Old}(\state,\block_{\sf tip})=\true$. 

Let $\block_{\sf tipc}$ be the next block in $\mathrm{Pivot}(\state)$. ($\block_{\sf tipc}=\bot$ if $\block_{\sf tipc}$ is the last block in $\mathrm{Pivot}(\state)$.) By definition~\ref{def:adaptive}, when $\mathrm{Old}(\state,\block_{\sf tip})=\true$, $\adp(\state)={\sf con}$ if $\mathrm{SubTW}(\state,\block_{\sf tipc})-\sibtreew{\state,\block_{\sf tipc}}< \advan$ or $\block_{\sf tipc}=\bot$.
For the case $\block_{\sf tipc}\neq\bot$, since $\block_{\sf tipc}$ is not in the pivot chain, there exists local state $\state'$ of another honest node satisfying $\mathrm{SubTW}(\state',\block_{\sf tipc})-\sibtreew{\state',\block_{\sf tipc}}\le 0$. Let $\tmpset$ include blocks which are in subtree of $\block_{\sf tip}$ and have been received by only a part of honest nodes. Let $w$ be the total block weight in $\tmpset$. The total block weight of symmetric difference between $\mathrm{SubTW}(\state,\block_{\sf tip})$ and $\mathrm{SubTW}(\state',\block_{\sf tip})$ is at most $w$. If $w<\advan$, we have 
$$ \mathrm{SubTW}(\state,\block_{\sf tipc})-\sibtreew{\state,\block_{\sf tipc}}\le w+\mathrm{SubTW}(\state',\block_{\sf tipc})-\sibtreew{\state',\block_{\sf tipc}}<\advan.$$

So we can guarantee $\adp(\state)={\sf con}$ if $\mathrm{Old}(\state,\block_{\sf tip})=\true$ and $w< \advan$. For the case $w\ge \advan$, we can't provide such a guarantee. We define such case as \emph{special status}. 

Let $w_{\sf h}$ and $w_{\sf m}$ denote the total weight of honest blocks and malicious blocks in $\tmpset$. So $w=w_{\sf h}+w_{\sf m}$. $w_{\sf h}$ is upper bounded by the total weight of honest blocks generated in the past $d$ rounds. All honest nodes generate $(1-\beta)\cdot n\cdot \difficulty/2^\kappa$ block weight per round in average. So we can find an threshold $\xi$ such that $w_{\sf h}\ge \xi$ with a low probability. So if the adversary wants to trigger special status, it needs to release blocks in subtree of $\block_{\sf tip}$ with total weight at least $\advan-\xi$ in consecutive $d$ rounds.

\paragraph{Block potential value} 
For each block $\block$ in the common pivot chain, we introduce a \emph{block potential value} $\pot(\advs,\block)$ to quantify the ability of adversary in changing or choosing the next common pivot chain block $\block$. 
Given round number $r$, if there exists a block $\block_r$ in the common pivot chain and $\block_r.\past$ contains all the blocks released no later than $r$. If we want to keep $\block_r$ in the common pivot chain, for each block in $\mathrm{Chain}(\block_r.\parent)$, the adversary should not be able to change its next common pivot chain block. We set global potential value $\cpot_r(\advs)$ the maximum block potential values of blocks in $\mathrm{Chain}(\block_r.\parent)$. 

\paragraph{Event value}
Recalling that we impose event values to upper bounds event influence on global potential value. The random variable for the sum of event values is relevant to block finalization.

Noticing that solving the proof-of-work puzzle is essentially a Markov process. Given $r_1\le r_2$, whether a $\hgen$ event or $\mgen$ event happens in round $r_2$ is independent with events before round $r_1$ and adversary strategies. Since block generation time cannot be manipulated, the attacker’s strategy becomes almost transparent when only looking at block generation events. Therefore, we only impose non-zero event values on block generation events to skim complicate adversary strategies. Usually, $\hgen$ events have negative values and $\mgen$ events have positive values. The event values also depend on the background when such event happens. (For example, whether such event happens in special status.) 



%

\subsection{Concepts}\label{sec:concepts}

\paragraph{Notations}
We define some tool functions and notations used in the following discussion. 
For any set of blocks $\mathbf{T}$ , $\tw(\mathbf{T}):=\sum_{\block\in\mathbf{T}}\block.\weight$ returns the total weight of blocks in $\mathbf{T}$. 
For any two blocks $\block_1,\block_2$ in graph $\graph$, $\block_1\in\chain{\block_2}$ is equivalent to $\block_2\in\tree{\graph,\block_1}$ according to notations in section~\ref{sec:ghost}. We denote their relation by $\block_1\preceq\block_2$. If furthermore there is $\block_1\neq\block_2$, then we write $\block_1\prec \block_2$. If $\vcmp$ is a chain of blocks in which the adjacent blocks have parent/child relation, we use $\tip(\vcmp)$ to denote the last block in $\vcmp$, and use $\nxt(\vcmp,\block)$ to denote the next element (child block) of block $\block$ in $\vcmp$. Specially, if $\block=\tip(\vcmp)$ or $\block\notin\vcmp$,  $\nxt(\vcmp,\block)$ returns $\bot$. We emphasize that the function $\tree{\graph,\block}$ and $\treew{\graph,\block}$ are well-defined when $\graph$ is only a subset of valid blocks other than a valid graph. For event $\event=(r,\block,t)$, we use $\event.{\sf block}$ to denote block $\block$. 

\subsubsection{Adversary State}\label{sec:adversary state}

The adversary state $\advs$ is a tuple of $$(\ggen,\gmax,\gmin,\gdta,\mathbf{M},\flagb,\vcmp,\speset,v).$$

$\ggen,\gmax$ and $\gmin$ are three graphs, which include all the generated blocks, all the released blocks (the blocks whose $\mrls$ or $\hgen$ event has happened) and all the blocks released $d$ rounds before (the blocks whose $\allrec$ event has happened). $\gdta:=\gmax\backslash\gmin$ denotes the blocks which may be received by only a part of honest nodes. $\mathbf{M}$ is a set of malicious blocks in $\ggen$, so we can distinguish the honest block and malicious blocks. $\speset$ and $v$ are a set of blocks and a real number relevant to special status. $\flagb$ is the flag block. It equals to $\bot$ to represents there is no flag block. $\vcmp$ is a variant of common pivot chain. It will be formally defined later. 

\todo{Clearify chain property of $\mathbf{C}$.}

For convenient, symbols $\ggen$, $\gmax$, $\gmin$, $\gdta$, $\mathbf{M}$, $\flagb$, $\vcmp$, $\speset$, and $v$ denote corresponding components of $\advs$. The symbols with subscript denote components of adversary state with the same subscript in the context. For example, $\gmax\l$ denotes the first component of $\advs\l$. For event $\event=(r,\block,t)$, we use $\event.{\sf round}$, $\event.{\sf block}$ and $\event.{\sf type}$ to denote its three components.

Traverse function $\psi(\advs\l,\event)=\advs$ updates $\advs\l$ to $\advs$ when an event $\event$ happens. $\psi$ is parameterized by two non-negative integers $\kam,\kah$, which is for security analysis only. $\kam$ and $\kah$ satisfy $2\kah+2\kam\le \heavyw$. \todo{and discuss how to choose $\kam$ and $\kah$ later XXX}. The following part introduces how the traverse function maintain each component of the adversary state.

\paragraph{The first five sets}
$\ggen$, $\gmax$ and $\gmin$ are initiated with a set with genesis block $\{\genesisblock\}$. $\mathbf{M}$ is initiated with an empty set. Upon event $\event$ happens, $\ggen\l$, $\gmax\l$, $\gmin\l$ and $\mathbf{M}\l$ are updated to $\ggen$, $\gmax$, $\gmin$ and $\mathbf{M}$ according to $\event.{\sf type}$. 
\begin{itemize}[nosep]
    \item $\hgen$ event: add $\event.{\sf block}$ to $\ggen\l$ and $\gmax\l$
    \item $\mgen$ event: add $\event.{\sf block}$ to $\ggen\l$ and $\mathbf{M}\l$
    \item $\mrls$ event: add $\event.{\sf block}$ to $\gmax\l$
    \item $\allrec$ event: add $\event.{\sf block}$ to $\gmin\l$
\end{itemize}

Then we set $\gdta=\gmax\backslash\gmin$. 

$\gmax$ includes all the blocks appears in the local state of an honest node, and the blocks in $\gmin$ must appear in the local state of all the honest nodes in an admissible environment (by definition~\ref{def:admissible}). Thus we can claim an honest node local state $\state$ must always be $\gmin\subseteq \state\subseteq\gmax$. Since all the blocks are added to $\ggen$ when it is generated, there must be $\gmax\subseteq \ggen$.

\begin{claim}\label{clm:graphrel}
    \stateconds, and any honest node local state $\state$ at the same time, there must be $\gmin \subseteq \state \subseteq \gmax \subseteq \ggen$. 
\end{claim}


\paragraph{Speical status}
%
%
The special status function $\mathrm{Spe}(\gdta,\vcmp)$ outputs $\true$ or $\false$ to indicate whether $\advs$ is in special status. $\mathrm{Spe}(\advs)$ represents $\mathrm{Spe}(\gdta,\vcmp)$ in case we don't care about which components are accessed. The special status is defined as follows.
\begin{definition}[Special Status]\label{def:special}
    Given $\gdta$ and chain $\vcmp$,  
    $\mathrm{Spe}(\gdta,\vcmp)$ returns $\true$ if one of the following three conditions satisfied:
    \begin{enumerate}
        \item $\tree{\gdta\cap \mathbf{M},\tip(\vcmp)}\ge \kam$
        \item $|\left\{\block \in \gdta\backslash \mathbf{M}\;|\;\block.\mathsf{weight}=1\right\}|\ge \kah$
        \item $|\left\{\block \in \gdta\backslash \mathbf{M}\;|\;\block.\mathsf{weight}=\heavyw\right\}|\ge 3$
    \end{enumerate}
\end{definition}

\paragraph{Flag block}
The flag block $\flagb$ is a block in $\gmax$ or equals to $\bot$ representing no flag block. It is initiated with $\bot$. Given adversary state $\advs\l$ and event $\event$, traverse function updates $\flagb\l$ by the following rules:
\begin{enumerate}
    \item If $\event.{\sf block}.\weight=\heavyw$, $\event.{\sf type}=\hgen$, $\mathrm{Spe}(\advs\l)=\false$, and $\gdta\l\backslash \mathbf{M}\l$ has no block with block weight $\heavyw$, then $\flagb=\event.{\sf block}$.
    \item If $\event.{\sf block}.\weight=\heavyw$, $\event.{\sf type}=\hgen$ and $\flagb\l\neq\bot$, then $\flagb=\bot$.
    \item If $\event.{\sf block}=\flagb\l\neq\bot$ and $\event.{\sf type}=\allrec$, then $\flagb=\bot$.
    \item For other cases, let $\flagb=\flagb\l$.
\end{enumerate}
The definition shows that a block can only become a flag block when it is generated, and it is no longer a flag block when its $\allrec$ event happens. So a flag block $\flagb$ must be in $\gdta$. 

\begin{claim}\label{clm:flagb}
    \stateconds, if $\flagb\neq \bot$, $\flagb$ must be an honest block in $\gdta$ with block weight $\heavyw$.
\end{claim}

\paragraph{Variant of common pivot chain} 
The chain $\vcmp$ is defined on $\gmax,\gmin$, $\vcmp\l$ and $\flagb$. \footnote{Notice that $\flagb$ depends on $\mathrm{Spe}(\advs\l)$ other than $\mathrm{Spe}(\advs)$. So we don't have a recursive dependency here.} 
First, we define function $\mathrm{Adv}((\gmax,\gmin,\flagb),\block)$ which provides a lower bound for subtree weight difference between block $\block$ and its maximum sibling blocks in an honest node local state, with the assumption that the flag block (if exists) has been received by all the honest nodes. $\mathrm{Adv}((\gmax,\gmin,\flagb),\block)$ can be represented by $\mathrm{Adv}(\advs,\block)$ for simplicity.
$$\mathrm{Adv}((\gmax,\gmin,\flagb),\block):=\treew{\gmin\cup\{\flagb\},\block}-\sibtreew{\gmax,\block}.$$
Here we regard $\{\bot\}$ as an empty set because $\bot$ presents ``no such a block''.

\begin{lemma}\label{lma:onebest}
    For any two graphs $\gmin\subseteq \gmax$, let $\flagb\in \gdta$ or $\flagb=\bot$. For any block $\block$, there exists at most one block $\block'\in\child(\gmax,\block)$ with $\mathrm{Adv}((\gmax,\gmin,\flagb),\block')>0$. 
\end{lemma}

The previous lemma shows that every blocks will have at most one child block $\block'$ with $\mathrm{Adv}((\gmax,\gmin,\flagb),\block')>0$. It is proved in appendix~\ref{sec:proof:lma:onebest}. Based on this property, we specifies the rule in updating $\vcmp\l$. $\vcmp$ is initiated with a genesis block $\genesisblock$. Given $\gmax$ and $\gmin$, $\vcmp\l$ is updated in the following steps:
\begin{enumerate}
    \item Started with the genesis block, we recursively visit its child block $\block$ in graph $\gmax$ which satisfies \\$\mathrm{Adv}((\gmax,\gmin,\flagb),\block)>0$, until we reach a block without such child block. These blocks organize the chain $\vcmp$.
    \item If there is a block $\block'\in \vcmp$ satisfying $\tip(\vcmp\l)\prec \block'$ and $\mathrm{Adv}((\gmax,\gmin,\flagb),\block')\le \kam+\kah$, we cut off the suffix started with $\block'$ from $\vcmp$.
\end{enumerate}
So we have the following claims for chain $\vcmp$.

\begin{claim}\label{clm:chainc}
    For any adversary state $\advs\l$ appearing in the blockchain protocol execution, and let $\advs=\psi(\advs\l,\event)$ for some event $\event$. Let block $\block$ be the last block of $\vcmp$ and block $\block\l$ be the last block of $\vcmp\l$. We have 
    \begin{enumerate}[nosep]
        \item For any block $\block'\in \vcmp$, $\mathrm{Adv}(\advs,\block')>0$. Specially, if $\tip(\vcmp\l)\prec \block'$, then $\mathrm{Adv}(\advs,\block')>\kam+\kah$.
        \item For any block $\block'\in \child(\gmax,\tip(\vcmp))$, $\mathrm{Adv}(\advs,\block')\le \kam+\kah$. Specially, if $\block'\preceq \tip(\vcmp\l)$, then $\mathrm{Adv}(\advs,\block')\le 0$.
        \item For any block $\block_{\sf c}$, if all the block $\block'$ in $\chain{\block_{\sf c}}$ satisfying $\mathrm{Adv}(\advs,\block')>\kam+\kah$, then $\block_{\sf c}\in \vcmp$.
    \end{enumerate}
\end{claim}

\paragraph{The set and value related to special status}
The set of blocks $\speset$ and value $\spevalue$ are initialized with empty set and 0. $\speset$ records all the malicious blocks have been taken into consider in special status. It is defined on $\gmax,\gmin,\vcmp,\mathbf{M}$ and $\speset\l$. Recalling that the special status consider the blocks in $\tmpset:=\tree{\gdta,\tip(\vcmp)}$. 
Upon the events happens, we set $\speset=\speset\l\cup (\tmpset\cap\mathbf{M})$. The value $v$ traces the total block weight in $\speset$ and increases by up to $\kam$ in each update. Formally, $\spevalue$ is initiated by 0 and let $\spevalue:=\spevalue\l+\min\left\{\kam,\tw( \speset\backslash\speset\l)\right\}$. We summarize several properties that are immediately induced from the definition.

\begin{claim}\label{clm:spesv}
    For any adversary state $\advs\l$ appearing in the blockchain protocol execution, and let $\advs=\psi(\advs\l,\event)$ for some event $\event$.  Let $\tmpset:=\tree{\gdta,\tip(\vcmp)}$, we have the following claims
    \begin{enumerate}[nosep]
        \item $\speset\l\subseteq\speset$
        \item $0\le \spevalue-\spevalue\l\le \kam$
        \item $\spevalue-\spevalue\l\le \tw(\speset)-\tw(\speset\l)$
        \item $\speset\backslash\speset\l\subseteq\tmpset\cap \mathbf{M}\subseteq \speset$.
    \end{enumerate}
\end{claim}

\subsubsection{Potential value}

\paragraph{Block potential value}
Now we provide a formal definition for \emph{block potential value} over the adversary state defined above. The block potential value $\pot(\advs,\block)$ is defined for a single block $\block$ over some adversary state $\advs$. Formally, the potential value $\pot(\advs,\block)$ is $\bot$ by default and it is an integer when $\block$ is old enough in $\gmin$ and agreed by all honest nodes, i.e. $\mathrm{Old}(\gmax,\block)=\true$ and $\block\in \vcmp$. 
When not being $\bot$, 
the potential value $\pot(\advs,\block)$ is the summation of three components $\potwith,\potadv,\potspe$ defined as follows. 

\begin{itemize}
    \item $\potwith(\advs,\block)$ is the total weight of blocks withheld by adversary under the subtree of $\block$:
    \begin{align*}
        \potwith(\advs,\block)&:=\treew{\ggen\backslash\gmax,\block}
    \end{align*}

    \item $\potadv(\advs, \mainchild)$ roughly corresponds to the volatility of the pivot child of $\block$. Let $\mainchild:=\nxt(\vcmp,\block)$, $\mathbf{N}:=\{\block'\in \gmax\backslash\mathbf{M}\;|\;\block'.\weight=1\}$.
    \begin{align*}
        \potadv(\advs, \mainchild) 
        &:=\left\{
            \begin{array}{ll}
                \kah+\kam-\mathrm{Adv}(\advs,\mainchild)-\min\{\tw(\tree{\gdta,\mainchild}\cap\mathbf{N}),\kah\} & \mainchild\neq \bot\\
                0 & \mainchild= \bot\\
            \end{array}
        \right.
    \end{align*}

    \item $\potspe(\advs, \block)+\spevalue$ measures the total cost for triggering the special status. Let $\mainchild=\nxt(\vcmp,\block)$.
    \begin{align*}
        \potspe(\advs ,\block)
        &:=
        \left\{\begin{array}{ll}
            \tw((\tree{\gdta,\mainchild} \cap \mathbf{M})\backslash\speset) & \mainchild\neq \bot \\
            0 & \mainchild = \bot
        \end{array}\right.
    \end{align*}
\end{itemize}

And the potential value for block $\block$ is defined as follows

\begin{align*}
    \pot(\advs,\block) &:=
    \left\{
        \begin{array}{ll}
            \potwith(\advs ,\block)+\potadv(\advs ,\nxt(\vcmp,\block))+\potspe(\advs ,\block) & \block \in \vcmp \wedge \mathrm{Old}(\gmin,\block)=\true \\
        \bot & \text{Otherwise}
        \end{array}
    \right.
\end{align*}

\paragraph{Global potential value}
The global potential function $\cpot(\advs,\graph)$ is defined on the adversary state $\advs$ and a graph $\graph$. Given adversary state $\advs$ and graph $\graph$, let $\vcmp':=\{\block\in \vcmp|\graph\nsubseteq \block.\past\}$. It returns the maximum block potential values in $\vcmp'$. 
\begin{equation}\label{eq:cpot}
    \cpot(\advs,\graph):=\max_{\block\in\vcmp \,\wedge\, \graph\nsubseteq \block.\past}\pot(\advs,\block).
\end{equation}
Specially, if $\vcmp'$ is empty set or all the blocks in $\vcmp'$ have a block potential value $\bot$, $\cpot(\advs,\graph)$ returns $0$. 

\subsubsection{Event value}

In order to upper bound the global potential value $\cpot(\advs,\graph)$,
we investigate how an event influences the potential and prove its upper bound by a case-by-case analysis.
More specifically, given an event $\event$ happens when the adversary state is $\advs\l$,
we introduce the \emph{event value} of $\event$ with respect to the adversary state $\advs\l$.
Formally, the event value of $\event$ is denoted by $\eventweight(\advs\l,\event)$ and defined as follows: 
\begin{itemize}
    \item If $\event.\mathsf{type}\in\{\mrls,\allrec\}$, $\eventweight(\advs\l,\event):=0$.
    \item If $\event.\mathsf{type}=\mgen$, $\eventweight(\advs\l,\event):=\event.\mathsf{block}.\mathsf{weight}$.
    \item If $\event.\mathsf{type}=\hgen$ and $\event.\mathsf{block}.\mathsf{weight}=0$, $\eventweight(\advs\l,\event):=0$
    \item If $\event.\mathsf{type}=\hgen$ and $\event.\mathsf{block}.\mathsf{weight}=1$, then
    $$\eventweight(\advs\l,\event):=\left\{
    \begin{array}{ll}
        -1 & \mathrm{Spe}(\advs\l)=\false \\
        0  & \mathrm{Spe}(\advs\l)=\true
    \end{array}\right.$$
    \item If $\event.\mathsf{type}=\hgen$ and $\event.\mathsf{block}.\mathsf{weight}=\heavyw$, then $\eventweight(\advs\l,\event)$ is defined as follows
    \begin{itemize}
        \item If $\gdta\l$ doesn't have an honest block with block weight $\heavyw$ and $\mathrm{Spe}(\advs\l)=\false$, $\event.\mathsf{block}.\mathsf{weight}:=2\kah+2\kam-\heavyw$. (If $\flagb=\event.\mathsf{block}$, it must be in this case.) 
        \item If $\gdta\l$ doesn't have an honest block with block weight $\heavyw$ and $\mathrm{Spe}(\advs\l)=\true$, $\event.\mathsf{block}.\mathsf{weight}:=0$. 
        \item If $\gdta\l$ has an honest block with block weight $\heavyw$ and $\flagb\l=\bot$, $\event.\mathsf{block}.\mathsf{weight}:=0$. 
        \item If $\gdta\l$ has an honest block with block weight $\heavyw$ and $\flagb\l\neq\bot$, $\event.\mathsf{block}.\mathsf{weight}:=\heavyw+\kam$. (If $\flagb\l\neq\bot$ and $\flagb=\bot$, it must be in this case.) 
    \end{itemize}
    (Note that a flag block must be an honest block with block weight $\heavyw$ and it must belong to $\gdta\l$. So $\gdta\l$ doesn't have an honest block with block weight $\heavyw$ only if $\flagb\l=\bot$.)
\end{itemize}

\subsection{Properties}\label{sec:properties}

\paragraph{Notations}

Let $\event_n$ denote the $n^{th}$ event since the ghast protocol launched and $\advs_{n-1}$ denote the adversary state when event $\event_n$ happens. The traverse function $\psi(\advs_{n-1},\event_n)=\advs_n$ updates $\advs_{n-1}$ to $\advs_n$. Variable $N(r)$ denotes the number of events happens before round $r$. In other words, the last event before round $r$ is $\event_{N(r)}$. 

Let $\delta\eqdef1-\beta/(1-\beta)$ and $\lambda\eqdef m(d+1)/\difficulty$. So $1-\delta$ represents the ratio of computing power between the honest participants and the adversary, $\lambda$ approximately represents the expectation of block generated in a maximum network delay. 

All the notations in the previous sub-sections are inherited here. 
$\\$

The proof of theorems are in the appendix.


First, we shows the sufficient conditions for block finalization. Theorem~\ref{thm:hisunchange} shows that the history of blocks in $\gmin_{N(r_0)}$ will be consistent among all the honest participants and remain unchanged as long as $\cpot(\advs,\gmin_{N(r_0)})<-\heavyw$ holds and all the blocks $\block$ satisfying $\gmin_{N(r_0)}\nsubseteq\block.\mathsf{past}$ must be old enough. This implies that the analysis for block finalization can be reduced to the analysis of potential value and ``not old enough'' blocks. 

\begin{theorem}\label{thm:hisunchange}
    In execution of ghast protocol, for any $r_0$ and $r_1$, we have 
    $$ \forall \ti\block \in\gmin_{N(r_0)}, \left|\bigcup\nolimits_{\substack{r\in\{r_1,r_1+1,\cdots,\rmax\}\\\state\in \uniongraphs_r}} \mathrm{Prefix}(\order_{\sf GHAST}(\state),\ti\block)\right|=1$$
    as long as both of the following conditions satisfied for any $n>N(r_1)$,
    \begin{itemize}
        \item $\cpot(\advs_n,\gmin_{N(r_0)})<-\heavyw$
        \item For any block $\block\in \ggen_n$ with $\gmin_{N(r_0)}\nsubseteq\block.\mathsf{past}$, it will be $\mathrm{Old}(\gmin_n,\block)=\true$. 
    \end{itemize}
    (Note that $\uniongraphs_r$ collects all the local states of honest nodes in round $r$.)
\end{theorem}


We study the difference of the addition of potential value and special value between two adjacent adversary state. For the block potential value, we have the following theorem. This theorem shows that when the adversary state is updated from $\advs\l$ to $\advs$, we provide an upper bound for block potential values except one special case: $\pot(\advs\l,\block)=\bot$, $\pot(\advs,\block)\neq\bot$ and $\mathrm{Old}(\gmin\l,\block)=\false$. 

\begin{theorem}\label{thm:potcase:final}
    \statecondm. 

    For any block $\block$ with $\pot(\advs\l,\block)\neq \bot$ and $\pot(\advs,\block)\neq\bot$, we have 
    $$ (\pot(\advs,\block)+\spevalue)-(\pot(\advs\l,\block)+\spevalue\l)\le \eventweight(\advs\l,\event).$$
    
    For any block $\block$ with $\pot(\advs\l,\block)=\bot$, $\pot(\advs,\block)\neq\bot$ and $\mathrm{Old}(\gmin\l,\block)=\true$, we have 
    $$ (\pot(\advs,\block)+\spevalue)-(\pot(\advs\l,\tip(\vcmp\l))+\spevalue\l)\le \eventweight(\advs\l,\event).$$
\end{theorem}

Since the global potential value takes the maximum block potential over a block set, theorem~\ref{thm:potcase:final} derives a similar property for the global potential value in theorem~\ref{thm:cpotdiff}. 



\begin{theorem}\label{thm:cpotdiff}
    \statecondm. For any graph $\graph$ 
    we have the following inequality 
    $$ (\cpot(\advs,\graph)+\spevalue)-(\cpot(\advs\l,\graph)+\spevalue\l)\le \eventweight(\advs\l,\event)$$
    as long as $\mathrm{Old}(\gmin\l,\genesisblock)=\true$ for genesis block $\genesisblock$ and the following holds for every block $\block\in \{\block'\in\vcmp|\graph\nsubseteq \block'.\past\}$, 
    $$ \mathrm{Old}(\gmin,\block)=\false \vee \mathrm{Old}(\gmin\l,\block)=\true \vee \cpot(\advs,\graph)\neq\pot(\advs,\block)$$
\end{theorem}

This theorem shows that when an event $\event$ happens, the value of $\cpot(\advs\l,\graph)+\spevalue\l$ grows at most the event value $\eventweight(\advs\l,\event)$. The exceptional case is that a block $\block$ just becomes old enough at $\advs$ and at the same time $\cpot(\advs, \block)$ determines $\cpot(\advs,\graph)$. 
So for any event index $n$, let $n'$ denote the last event triggering the exceptional case. The global potential value $\cpot(\advs_n,\graph)$ is upper bounded by the summation of two terms: 1) the global potential value after event $\event_{n'}$, a.k.a. $\cpot(\advs_{n'},\graph)$; 2) The summation of event values (removing the influence from special value), a.k.a. $\sum_{i=n'+1}^{n} \left( \eventweight(\advs_{i-1},\event_i) - (\spevalue_i-\spevalue_{i-1})\right)$. For the first term, theorem~\ref{thm:cpotwhenold} provides an upper bound. For the second term, theorem~\ref{thm:eventdec} show that the summation of tends to negative infinity in a long enough time. 

So if we want to prove $\cpot(\advs_n,\gmin_{N(r_0)})<-\heavyw$ for all the sufficient large $n$, (the first sufficient condition for block finalization given in theorem~\ref{thm:hisunchange}), we only need to show that the exceptional case of theorem~\ref{thm:cpotdiff} never happens after a time point. Recalling that $\cpot(\advs_n,\gmin_{N(r_0)})$ takes the maximum block potential weight among blocks with $\block.\past \nsubseteq \gmin_{N(r_0)}$. Since theorem~\ref{thm:mustoldenough} proves that for all the blocks satisfying $\block.\past \nsubseteq \gmin_{N(r_0)}$ must be old enough after a time point, this guarantee that the exceptional case will never happen for $\gmin_{N(r_0)}$ after that. This is also the second sufficient condition for block finalization.

\begin{theorem}\label{thm:cpotwhenold}
    Given $(m,\beta,\delay,\adv,\env)$ which is admissible w.r.t. $(\proto_{\sf GHAST}^{\vec{\eta}},\order_{\sf GHAST})$. Let 
    $$w(\varepsilon) \eqdef  4\lambda\cdot \max\left\{\frac{140}{\delta^2}\cdot\log\left(\frac{9000}{\varepsilon\delta^2}\right),\frac{8(\timerdiff+4)}{\delta}\right\}.$$
    
    Let ${\ti \graph}_r$ denote all the blocks with $\mathrm{Old}(\gmin_{N(r)},\block)=\false$. When $\lambda \ge 0.8\log(500/\delta)$, $\timerw= 2\lambda/\delta$ and $\heavyw=30\lambda/\delta$, for any $r_2\ge 0$ and $\varepsilon>0$, we have
    $$ \Pr\left[ \exists N(r_2)<n\le N(r_2+1), \exists \block \in {\ti \graph}_{r_2}, \pot(\advs_n,\block)\ge w(\varepsilon) \right]\le \varepsilon.$$
\end{theorem}

\begin{theorem}\label{thm:eventdec}
    Given $(m,\beta,\delay,\adv,\env)$ which is admissible w.r.t. $(\proto_{\sf GHAST}^{\vec{\eta}},\order_{\sf GHAST})$. When $\lambda \ge {0.8\log(500/\delta)}$ and $\heavyw=30\lambda/\delta$, for any round $r_1<r_2$, $\rho>0$ and $\varepsilon>0$, if 
    $$\frac{r_2-r_1}{\delay+1}\ge \max\left\{(3+\rho/\heavyw)\cdot\frac{600}{\delta^2},\log\left(\frac{4}{\varepsilon}\right)\cdot\frac{3000}{\delta^3},\log\left(\frac{404}{\varepsilon}\right)\cdot\frac{200}{\delta}\right\},$$
    we have 
    $$ \Pr\left[\forall r_2,\sum_{i=N(r_1)+1}^{N(r_2)}\eventweight(\advs_{i-1},\event_i)-\left(\spevalue_{N(r_2)}-\spevalue_{N(r_1)}\right)\ge -\rho\middle| \view_{r_1}\right]\le \varepsilon.$$
(Let $\view_r$ denote the joint view in $\view^{(\proto_{\sf GHAST},\order_{\sf GHAST})}(\env,\adv,\secp)$ before round $r$. )

\end{theorem}

\begin{theorem}\label{thm:mustoldenough}
    Given $(m,\beta,\delay,\adv,\env)$ which is admissible w.r.t. $(\proto_{\sf GHAST}^{\vec{\eta}},\order_{\sf GHAST})$. When $\beta\ge 0.1$, $\timerw\ge 2\lambda/\delta$ and $r_\Delta$ satisfying 
    $$r_\Delta\ge\frac{\timerw\difficulty}{m}\cdot 
    \max\left\{
    \frac{128}{\delta^2}\cdot \log\left(\frac{8400}{\varepsilon\delta^2}\right),
    \frac{8(\timerdiff+2)}{\delta}
    \right\}
    $$
    for any $r$, we have 
    $$\Pr\left[\exists n\ge N(r+r_\Delta), \exists \block \in \ggen_n, \gmin_{N(r)}\nsubseteq\block.\past \wedge  \mathrm{Old}(\gmin_n,\block)=\false\right]\le \varepsilon.$$
\end{theorem}


Summarize all the previous proofs, finally we gives the finalization properties of the GHAST protocol under any possible attack strategies. 

\begin{theorem}\label{thm:final}
    Given $(m,\beta,\delay,\adv,\env)$ which is admissible w.r.t. $(\proto_{\sf GHAST}^{\vec{\eta}},\order_{\sf GHAST})$.
    When $\beta\ge 0.1$, $\lambda \ge 0.8\log(500/\delta)$, $\timerw= 2\lambda/\delta$ and $\heavyw=30\lambda/\delta$, $\view^{(\proto_{\sf GHAST},\order_{\sf GHAST})}(\env,\adv,\secp)$ has the $\varepsilon$ latency

    $$\delay \cdot O\left(\max\left\{\frac{\log\left(\frac{1}{\varepsilon\delta}\right)}{\delta^3}+\frac{\timerdiff}{\delta^2}\right\}\right)\;\text{rounds}.$$
\end{theorem}
\begin{proof}
    Given round $r_0$, let $n$ denote the largest event index with
    $$ (\cpot(\advs_{n}, \gmin_{N(r_0)})+\spevalue_{n})-(\cpot(\advs_{n-1}, \gmin_{N(r_0)})+\spevalue_{n-1})>\eventweight(\advs_{n-1},\event_n).$$

    We define 
    $$ w_1(\varepsilon)\eqdef 4\lambda\cdot \max\left\{\frac{140}{\delta^2}\cdot\log\left(\frac{9000}{\varepsilon\delta^2}\right),\frac{8(\timerdiff+4)}{\delta}\right\}.$$
    $$ w_2(\varepsilon)\eqdef \max\left\{(4+w_1(\varepsilon)/\heavyw)\cdot\frac{600}{\delta^2},\log\left(\frac{4}{\varepsilon}\right)\cdot\frac{3000}{\delta^3},\log\left(\frac{404}{\varepsilon}\right)\cdot\frac{200}{\delta}\right\}.$$
    $$ w_3(\varepsilon)\eqdef \frac{\timerw\difficulty}{m}\cdot 
    \max\left\{
    \frac{128}{\delta^2}\cdot \log(\frac{8400}{\varepsilon\delta^2}),
    \frac{8(\timerdiff+2)}{\delta}
    \right\}.$$

    According to theorem~\ref{thm:cpotdiff}, since $(\cpot(\advs_{n}, \gmin_{N(r_0)})+\spevalue_{n})-(\cpot(\advs_{n-1}, \gmin_{N(r_0)})+\spevalue_{n-1})>\eventweight(\advs_{n-1},\event_n)$, there must exists block $\block$ with $\mathrm{Old}(\gmin_{n-1},\block)=\false$, $\mathrm{Old}(\gmin_{n},\block)=\true$ and $\cpot(\advs_n, \gmin_{N(r_0)})=\pot(\advs_n,\block)$. According to theorem~\ref{thm:cpotwhenold}, since $\mathrm{Old}(\gmin_{n-1},\block)=\false$, with probability $1-\varepsilon$, 
    $\pot(\advs_n,\block)\le w(\varepsilon).$

    According to the choice of $n$, for any $n'>n$, we claim 
    $\cpot(\advs_{n'}, \gmin_{N(r_0)})\le w(\varepsilon) +\sum_{i=n+1}^{n'} \eventweight(\advs_{i-1},\event_i)-(v_{n'}-v_n).$
    So $\cpot(\advs_{n'}, \gmin_{N(r_0)})< -\heavyw$ if 
    \begin{equation}\label{ieq:event_sum}
        \sum_{i=n+1}^{n'} \eventweight(\advs_{i-1},\event_i)-(v_{n'}-v_n) < -\heavyw-w(\varepsilon).
    \end{equation}
    
    According to theorem~\ref{thm:eventdec}, with probability $1-\varepsilon$, for any event $n'$ which is $w_2(\varepsilon)$ rounds later than event $n$, inequality~\ref{ieq:event_sum} holds. 

    Let round $r_{\sf young}$ denote the largest round that exists $\block$ such that $\gmin_{N(r_0)}\nsubseteq \block.\past$ and $\mathrm{Old}(\gmin_{N(r_{\sf young})},\block)=\false$.  According to theorem~\ref{thm:mustoldenough}, with probability at least $1-\varepsilon$, 
    $$ r_{\sf young}\le r_0+w_3(\varepsilon).$$

    Notice that $n<N(r_{\sf young}+1)$ by definition, so we claim for any $n'>N(r_0+w_3(\varepsilon)+w_2(\varepsilon))$, there will be 
    \begin{itemize}
        \item $\pot(\advs_{n'},\gmin_{N(r_0)})<-\heavyw$
        \item For any block $\block\in \ggen_{n'}$ with $\gmin_{N(r_0)}\nsubseteq\block.\mathsf{past}$, it will be $\mathrm{Old}(\gmin_{n'},\block)=\true$. 
    \end{itemize}
    
    According to theorem~\ref{thm:hisunchange}, the history of blocks in $\gmin_{N(r_0)}$ will not be changed after round $r_0+w_2(\varepsilon)+w_3(\varepsilon)$. 

    In all, the blocks released before round $r_0-\delay$ will be confirmed after round $r_0+w_2(\varepsilon)+w_3(\varepsilon)$ with probability $1-3\varepsilon$. Notice that $w_2(\varepsilon)+w_3(\varepsilon)+\delay=\delay\cdot O\left(\max\left\{\frac{\log\left(1/(\varepsilon\delta)\right)}{\delta^3},\frac{\timerdiff}{\delta^2}\right\}\right)$, we have proved this theorem. 
\end{proof}



\newcommand{\ablock}{\mathbf{a}}
\newcommand{\cblock}{\mathbf{c}}

\section{Confirmation Policy}

In order to show the low confirmation delay of {\name} protocol, we provide a concrete method in estimating the block confirmation risk. 

Consider a block $\block'$ in Tree-Graph structure and suppose $\block'$ is in the past set of a pivot block $\block$. Then the finalization of $\block'$ can be reduced to block $\block$. Given a local state $\state$ and block $\block\in \mathrm{Pivot}(\state)$, if block $\block$ has more subtree weight than one of its siblings, block $\block$ may be kicked out of the pivot chain and the history of block $\block'$ may change. So in this section, we study the probability that block $\block$ is kicked out of pivot chain in the future under the assumption that the blocks in $\chain{\block.\parent}$ are always the common pivot chain. 
%


For simplicity, we ignore the network delay and assume all the blocks are delivered to all the honest nodes at once. But we still allow the attacker to withhold blocks. The confirmation policy considering the network delay can have a similar idea, except assuming an honest participant can not see the newly generated honest block in past time $2\delay$.

The confirmation rule consists of two parts. First, we estimate the confirmation risk under an assumption that for the first $\theta$ blocks $\block'$ generated later that block $\block.\parent$, if $\block'$ is in the subtree of block $\block'.\parent$, there will be $\mathrm{Adapted}(\block')=\false$. $\theta$ can be an arbitrary positive integer. Second, we estimate the probability that such assumption is break. By the union bound, we get the confirmation risk finally. 

\subsection{Confirmation risk under an assumption}\label{sec:confirmation_risk}

Here, we define two variables $m$ and $n$ for a Tree-Graph $\graph$. $m$ and $n$ shadow the same symbols defined in execution model and security analysis. 
\begin{align}
	m&\eqdef \mbox{An }\textit{upper bound }\mbox{for the number of honest blocks generated later than }\block.\parent \\
	n&\eqdef \mbox{A }\textit{lower bound }\mbox{for subtree weight advantage compared between }\block\mbox{ and its siblings}\\
	&\mbox{(When computing the subtree weight of }\block\mbox{'s sliblings, only honest blocks are taken into account.)}\notag
\end{align}

The confirmation risk is computed conditioned on given $m$ and $n$. In estimation for value $m$, notice that the blocks in $\block.\parent.\past$ must be generated earlier than $P(\block)$, we can just count the number of blocks in $\block\backslash\block.\parent.\past$ as the value of $m$. In estimation for value $n$, we try to distinguish the malicious blocks as many as possible. Then we compare the subtree weight of $\block$ and its maximum siblings for the normal case. If we fail to distinguish malicious block, we will get a more conservative confirmation risk estimation. 

Let $K$ denote the total weight of malicious block in the subtree of block $\block.\parent$, $T$ denote the number of blocks generated after the generation of $\block.\parent$. Although $K$ and $T$ are fixed in the view of execution model, but a participants can not know the exact number of how many malicious blocks generated in a given time interval. So we regard $K$ and $T$ as random variables and discuss their probability distribution conditioned on $m,n$. 

Let random variable $X_0$ denote the current difference between subtree weight of block $\block$ and the maximum subtree weight of $\block$'s sibling blocks (including the blocks withheld by the attacker). Then it will be $X_0\ge n-K$. Let $X_i$ denote the difference after $i$ blocks generation. So block $\block$ will be kicked out of pivot chain only if $X_i<0$. Notice that $X_{i+1}-X_i$ is independent with $K$ and $T$ because their randomness are from different time interval. So we can regard $K$ and $T$ as fixed value in discussing $X_{i+1}-X_i$.

We have assume that for the first $\theta$ blocks $\block'$ generated later than $\block.\parent$, if block $\block'$ is in the subtree of $\block.\parent$, its block weight must be 1.
\footnote{But we don't discuss the probability conditioned on the assumption. If we want to learn the probability of event $E$, the probability under the assumption $A$ means $\Pr[E\wedge A]$ and the probability conditioned on the assumption $A$ means that $\Pr[E|A]$. }
Let $\ti\theta = \max\{ \theta-T,0\}$. 
When $i\le \ti\theta$, each time an honest node generate a block, $X_i$ will increase by 1. Each time an malicious node generates a block, $X_i$ will decrease by at most one. We have the following claim: 

\begin{align*}
	X_{i}-X_{i-1}\ge \left\{\begin{array}{ll}
		1  & \mbox{The }i\mbox{-th block is generated by honest nodes.}\\
		-1 & \mbox{The }i\mbox{-th block is generated by the attacker. }
	\end{array}\right.
\end{align*}

When $i>\ti\theta$, the $i$-th block could be a normal block or an adapted block. Since the block weight never exceed $\heavyw$, we have $-\heavyw\le X_{i+1}-X_{i}\le \heavyw$. Since the normal block and adapted block have the same block weight expectation, similarly, we have

\begin{align*}
	\mathbb{E}[X_{i}-X_{i-1}|X_{i-1},\cdots,X_0]\ge \left\{\begin{array}{ll}
		1  & \mbox{The }i\mbox{-th block is generated by honest nodes.}\\
		-1 & \mbox{The }i\mbox{-th block is generated by the attacker. }
	\end{array}\right.
\end{align*}

Since the adversary shares $\beta$ computing power, we have $\mathbb{E}[X_{i}-X_{i-1}|X_{i-1},\cdots,X_0]=1-2\beta$. Thus 
\begin{align*}
 \forall s>0, \mathbb{E}[e^{s(X_{i-1}-X_{i})}|X_{i-1},\cdots,X_0] \le 
\left\{\begin{array}{ll}
	e^{s\beta}+e^{-s(1-\beta)}  & i\le \ti\theta\\
	\frac{1}{\heavyw}\cdot\left(e^{s\heavyw\beta}+e^{-s\heavyw(1-\beta)}\right) & i> \ti\theta
\end{array}\right.
\end{align*}

We use $g_1(s)$ to denote $e^{s\beta}+e^{-s(1-\beta)}$ and $g_2(s)$ to denote $\frac{1}{h}\cdot\left(e^{sh\beta}+e^{-sh(1-\beta)}\right)$. We have 
\begin{align*}
	\forall s>0, \mathbb{E}[e^{s(X_{0}-X_{i})}|X_{i-1},\cdots,X_0] \le g_1(s)^{\min\{i,\ti\theta\}}\cdot g_2(s)^{\max\{0,i-\ti\theta\}}.
\end{align*}

According to Markov's inequality, 
\begin{align*}
	\forall c>0,\forall i\ge 0, \Pr[X_0-X_i\ge c] \le \min_s g_1(s)^{\min\{i,\ti\theta\}}\cdot g_2(s)^{\max\{0,i-\ti\theta\}}\cdot e^{-tc}.
\end{align*}

Recalling $X_0\ge n-K$, applying the union bound, 
\begin{align*}\label{eq:partial_risk}
	\Pr[\exists i\ge 0, X_i\le 0] \le \min\left\{1,\sum_{i=1}^{+\infty} \min_s \left(g_1(s)^{\min\{i,\ti\theta\}}\cdot g_2(s)^{\max\{0,i-\ti\theta\}}\cdot e^{s(K-n)}\right)\right\}.
\end{align*}

Let random variable $K'$ denote the number of malicious blocks generated from the creation of $\block.\parent$ to the current time. So $T\le m+K'$. During the time interval that honest nodes generates $m+1$ blocks, the number of malicious blocks follows the negative binomial distribution with $m+1$ successes and $1-\beta$ success probability. Thus 
$$ \forall k>0, \Pr[K'\ge k]\le I_{1-\beta}(k,m+1).\quad (I\mbox{ is regularized beta function.}).$$

When $K'+m \le \theta$, all the $K'$ malicious blocks will be normal blocks with weight one except the blocks not in subtree of $\block.\parent$. Thus $K\le K'$.

Let $p(K,T)$ equal the right hand side of inequality (\ref{eq:partial_risk}). Notice that $p(K,T)$ is non-decreasing in terms of $K$ and $T$ and $p(n,T)=1$. For any given $t\le \theta$, the confirmation risk will be 

\begin{align*}
	   \mathbb{E}_{K,T}[p(K,T)] \le & \Pr[T>t] + \mathbb{E}_{K}[p(K,t)] \\ 
	\le & \Pr[K'+m>t] + p(0,t) + \sum_{k=1}^{n} \Pr[K\ge k]\cdot (p(k,t)-p(k-1,t)) \\ 
	\le &  \Pr[K'+m>t\vee K'+m>\theta] + p(0,t) + \sum_{k=0}^{n-1} \Pr[K'\ge k]\cdot p(k,t) \\
	\le &I_{1-\beta}(t-m+1,m+1) + p(0,t) + \sum_{k=0}^{n-1} I_{1-\beta}(k,m+1)\cdot p(k,t)
\end{align*}

\subsection{Risk for breaking the assumption}

The previous computation assumes the GHAST weight adaption is not triggered under the subtree of $\block.\parent$ during the generation of $\theta$ blocks. In this subsection, we provide a concrete method to estimate the risk that the attacker breaks this assumption. The $\theta^{th}$ block generation time after  $\block.\parent$ is called the \emph{deadline}. 

If a block $\cblock$ in subtree of $\block.\parent$ satisfying $\mathrm{Adapted}(\cblock)=\true$. There must exists block $\ablock\in \mathrm{Chain}(\cblock)$ such that 
\begin{enumerate}
    \item $\mathrm{MaxTH}(\cblock.\past)-\mathrm{TimerHeight}(\ablock.\parent)> \timerdiff.$ 
    \item $\treew{\cblock.\past,\ablock}-\mathrm{SibSubTW}(\cblock.\past,\ablock)\le \advan.$
\end{enumerate}

Let event $E_1(\ablock)$ and $E_2(\ablock)$ denote that there exists block $\cblock$ generated before the deadline in subtree of $\block.\parent$ satisfying the first and the second condition respectively. Then the risk that the assumption is broken is no more than
\begin{align}
    \Pr[\exists \ablock,E_1(\ablock)\wedge E_2(\ablock)]\le \min\{\Pr[\exists \ablock, E_1(\ablock)],\Pr[\exists \ablock, E_2(\ablock)]\}.
\end{align}

Now, we will discuss how to compute $\Pr[E_1(\ablock)]$ and $\Pr[E_2(\ablock)]$ for fixed block $\ablock$. 

Let $Z$ denote the maximum timer height among all the generated blocks at the creation of $\block.\parent$ (including the blocks withheld by the adversary).  A necessary condition for $E_1(\cblock)$ is that there are $\timerdiff-(Z-\mathrm{TimerHeight}(\ablock.\parent))$ timer blocks among the first $\theta$ blocks after the generation of $\block.\parent$. Thus 
\begin{align}
    \Pr[E_1(\ablock)] \le \Pr[\mathrm{B}(\theta,1/\timerw)\le \timerdiff-(Z-\mathrm{TimerHeight}(\ablock.\parent)]\qquad \mbox{$\mathrm{B}$ denotes the binomial distribution}
\end{align}

The value of $Z-\mathrm{MaxTH}(\block.\parent)$ is upper bounded by the number of consecutive malicious blocks in the timer chain, which can be estimated by chain-quality property in Nakamoto consensus proved by Pass et al~\cite{pass2017analysis}. We will formally reason it in the later version. 

In the estimation for $\Pr[E_2(\ablock)]$, we only consider the block $\ablock$ with $\ablock\in \chain{\block.\parent}$. Since $\cblock$ is in the subtree of $\block.\parent$, it will be $\block.\parent.\past\subseteq \cblock.\past$. Thus $\treew{\cblock.\past,\ablock}\ge \treew{\block.\parent.\past,\ablock}.$ Let $w\eqdef  \treew{\block.\parent.\past,\ablock}$ is a known value. So $E_2(\ablock)$ happens only if 
$$ \mathrm{SibSubTW}(\cblock.\past,\ablock)\ge w-\advan.$$

We takes two parameters as input, 
\begin{align*}
	m&\eqdef \mbox{An }\textit{upper bound }\mbox{for the number of honest blocks generated later than }\ablock.\parent\mbox{ and before }\block.\parent \\
	l&\eqdef \mbox{An }\textit{upper bound }\mbox{for the total weight of honest blocks contributes to the subtree of sliblings blocks of }\ablock
\end{align*}

So the malicious blocks must generate at least $w-\advan-l$ block weights between the creation of $\ablock.\parent$ and the deadline. The probability estimation for this events compose of two steps: 1) the number $N$ of blocks that the adversary generates in this time interval; 2) If the adversary is allowed to generated $N=n$ blocks and switch its consensus strategy adaptively, how many block weights it can generated. We can choose proper $n$ and claim that if the adversary generates $w-\advan-l$ block weights during the time interval, the adversary must be either generates more than $n$ blocks during the time interval or generates $w-\advan-l$ block weights in the first $n$ blocks. 

Let $N_1$ be the number of malicious blocks generated between the creation of $\ablock.\parent$ and $\block.\parent$ and $N_2$ be the number of malicious blocks generated between the creation of $\block.\parent$ and the deadline. So $N=N_1+N_2$. $N_1$ is upper bounded by the negative binomial distribution with $m+1$ successes and $\beta$ success probability, $N_2$ is a binomial distribution with $\theta$ trials and $\beta$ success probability. By the union bound, we have 

$$ \forall n_1,n_2,\Pr[N\ge n_1+n_2]\le \Pr_{N_1\sim \mathrm{NB}(m+1,1-\beta)}[N_1\ge n_1]+\Pr_{N_2\sim \mathrm{B}(\theta,\beta)}[N_2 \ge n_2].$$

If the adversary is allowed to generate $n$ blocks during the this interval. Let $Y_i$ denote the total block weight after the generation of $i^{th}$ block. Then we have $Y_0=0$ and 
$$ \mathbb{E}[Y_i-Y_{i-1}|Y_{i-1},\cdots,Y_0]=1.$$

Notice that $Y_i\le \heavyw$, thus 
$$ \forall t>0, \mathbb{E}[e^{t(Y_i-Y_{i-1})}|Y_{i-1},\cdots,Y_0]\le \frac{\heavyw-1+e^{t\heavyw}}{\heavyw}.$$

So we have 
$$ \Pr[Y_n\ge \rho]\le \min_t \left(\left(\frac{\heavyw-1+e^{t\heavyw}}{\heavyw}\right)^n\cdot e^{-t\rho}\right).$$

So an upper bound for probability of $E_2(\ablock)$ is $\Pr[N\ge n]+\Pr[Y_n\ge w-\advan-l]$. 
$\\$

Taking the union bound over all the possible $\ablock$ gives the risk that an attacker generated an adaptive weighted block before the deadline. In order to achieve a better result, we compute the probability of $E_1(\ablock)\wedge E_2(\ablock)$ in batch. Formally, we split $\chain{\block.\parent}$ into several slices. For each slice $\mathbf{S}$, we compute the probability that  
$$\Pr[\exists \cblock \in \mathbf{S}, E_1(\cblock)\wedge E_2(\cblock)] \le \min\{\Pr[\exists \cblock \in \mathbf{S}, E_1(\cblock)],\Pr[\exists \cblock \in \mathbf{S}, E_2(\cblock)]\}.$$
Let $\cblock'$ be the oldest block in $\mathbf{S}$, then $\Pr[\exists \cblock \in \mathbf{S}, E_1(\cblock)]=\Pr[E_1(\cblock')]$. In estimation of $\Pr[\exists \cblock \in \mathbf{S}, E_2(\cblock)]$, we pick the maximum $m,l$ and minimum $w$ among all the blocks $\ablock$ in the slice. (Notice that $m,l,w$ are defined on given block $\ablock$ in the estimation of $\Pr[E_2(\ablock)]$.)
This gives a better result.

\section{Implementation}

A separated work~\cite{conflux-sys} implements this protocol in the Conflux blockchain system. 
Under the experiment with 20Mbps network bandwidth limit per node, Conflux chooses the block generation rate of 4 blocks per second and the block size limit of 300K. Under this parameter, Conflux achieves a block throughput of 9.6Mbps. The experiments shows that all the blocks can be propagated to 99\% full nodes in 15 seconds, which implies $\lambda = 60$. ($\lambda$ is defined in section~\ref{sec:properties}.) Conflux chooses the consensus protocol parameters $\vec{\eta}$ to tolerate liveness attacks from a powerful attacker that controls 40\% of the network computation power. (a.k.a. $\beta = 0.4$, $\delta = 1-\beta/(1-\beta)=1/3$.) We choose the parameter $\heavyw=600$, \footnote{Theorem~\ref{thm:final} requires $\heavyw=30\lambda/\delta$. This requirement derives from lemma~\ref{lma:probf}. With the concrete parameter $\lambda = 60$ and $\delta = 1/3$, we find a smaller solution $\heavyw=600$ which achieves the same results as lemma~\ref{lma:probf} in complexity.} $\timerw = 2\lambda/\delta = 360$ and $\advan=3\heavyw$. 

Waiting for six blocks in Bitcoin has the confirmation risk $2\times 10^{-5}$ with $\beta =0.1$ adversary. To obtain the same confidence as waiting for six blocks in Bitcoin in a short time, we set $\timerdiff=160$ to make the risk output by the computation policy small enough. In the global view, a block can be confirmed in less than three times of network delay. 

\section{Conclusion}

In this work, we design a novel consensus protocol that achieves security and low confirmation delay in a normal scenario. The protocol executes two consensus strategies to achieve efficiency in normal cases and switch to a conservative strategy in defending a liveness attack. We provide a rigorous security analysis to show such design can resolve all the possible liveness issues in GHOST protocol. Any block will become $\varepsilon-$finalized after a logarithm time after its release. 
A separated work \cite{conflux-sys} implements this protocol and shows that the block can be confirmed in less than $3\delay$ in the global view. 












{\footnotesize \bibliographystyle{acm}
	\bibliography{paper}}
\clearpage
\appendix


\section{Potential Values}

\subsection{Preparation}

First we remind the most frequently used notations in this section. Functions $\chain{\cdot}$,$\tree{\cdot}$, $\treew{\cdot}$ are $\sibtreew{\cdot}$ are defined in section~\ref{sec:proto}. Functions $\tw(\cdot)$, $\tip(\cdot)$ and $\nxt$ are defined in the beginning of section~\ref{sec:concepts}. There are four types of events $\hgen$, $\mgen$, $\mrls$ and $\allrec$ defined in section~\ref{sec:adversary state}. Recalling that we use symbols $\ggen$, $\gmax$, $\gmin$, $\mathbf{M}$, $\flagb$, $\vcmp$, $\speset$, $v$ to denote corresponding components of adversary state $\advs$ in the context and use the symbols with subscript ``$\l$'' to denote components of $\advs\l$. These symbols will never be used as ephemeral symbols. $\block_1\preceq\block_2$ represents the relation $\block_1\in\chain{\block_2}$. If furthermore there is $\block_1\neq\block_2$, we write $\block_1\prec\block_2$. Since the local state $\state$ of a participant node is essentially a graph $\graph$, we use $\graph$ to denote the local state here.

We claim the necessary condition and sufficient condition for a block $\block\in\graph$ be in the pivot chain $\mathrm{Pivot}(\graph)$. 

\begin{lemma}\label{lma:ghost}
    For any graph $\graph$ and a block $\block$ in graph $\graph$, let $\genesisblock$ be the genesis block, the sufficient condition and necessary condition of $\block\in\mathrm{Pivot}(\graph)$.
    \begin{itemize}
        \item Sufficient condition: $\forall \block'\in \chain{\block}\backslash\{\genesisblock\}$, $\treew{\graph,\block'}-\sibtreew{\graph,\block'}>0$
        \item Necessary condition: $\forall \block'\in \chain{\block}$, $\treew{\graph,\block'}-\sibtreew{\graph,\block'}\ge 0$
    \end{itemize}
\end{lemma}
\begin{proof}
    Recalling that the pivot chain is a list of blocks which starts with the genesis block $\genesisblock$ and recursively expands the best child (defined in equation~\ref{eq:def:bestchild}) of last block into it. So the necessary and sufficient condition for $\block\in\mathrm{Pivot}(\graph)$ is $\forall \block'\in \chain{\block}\backslash\{\genesisblock\}, \block'=\mathrm{BestChild}(\graph,\block'.\parent)$. ($\genesisblock$ is the genesis block.) The definition of best child in eq.~(\ref{eq:def:bestchild}) shows that the sufficient condition and necessary condition for $\block'=\mathrm{BestChild}(\graph,\block'.\parent)$
    \begin{itemize}
        \item Sufficient condition: $\treew{\graph,\block'}-\sibtreew{\graph,\block'}>0$
        \item Necessary condition: $\treew{\graph,\block'}-\sibtreew{\graph,\block'}\ge 0$
    \end{itemize}
    So we have the necessary condition and sufficient condition for pivot chain.
    \begin{itemize}
        \item Sufficient condition: $\forall \block'\in \chain{\block}\backslash\{\genesisblock\}$, $\treew{\graph,\block'}-\sibtreew{\graph,\block'}>0$
        \item Necessary condition: $\forall \block'\in \chain{\block}\backslash\{\genesisblock\}$, $\treew{\graph,\block'}-\sibtreew{\graph,\block'}\ge 0$
    \end{itemize}

    Specially, the genesis block doesn't have any sibling blocks, so $\sibtreew{\graph,\genesisblock}=0$ always holds. Thus $\treew{\graph,\genesisblock}-\sibtreew{\graph,\genesisblock}\ge 0$ can be also a necessary condition.
\end{proof}
$\\$

We claim some properties which derives from the definitions directly. We do not explicitly refer the previous two claims in use since they are intuitive.

By the definition of $\tree{\graph,\block},\treew{\graph,\block}$ and $\sibtreew{\graph,\block}$ in eq.~(\ref{eq:def:subtree},\ref{eq:def:subtreew},\ref{def:sibsubtw}), when the graph $\graph$ includes more blocks, the outputs of these three functions will be non-decreasing. Formally, we have the following claim. \todo{describe this claim.}

\begin{claim}
    For any set of blocks $\graph$ and a block $\block$, we have $\tree{\graph,\block}\subseteq\graph$. 

    For any set of blocks $\graph_1$ and $\graph_2$ with $\graph_1\subseteq\graph_2$ and a block $\block$, 
    we have $\tree{\graph_2,\block}=\tree{\graph_1,\block}\cup \tree{\graph_2\backslash\graph_1,\block}$ , 
    $\treew{\graph_2,\block}= \treew{\graph_1,\block} + \treew{\graph_2\backslash\graph_1,\block}$ 
    and $\sibtreew{\graph_1,\block}\le \sibtreew{\graph_2,\block}$. 

    For any set of blocks $\graph$ and blocks $\block_1,\block_2$ with $\block_1\preceq\block_2$, 
    we have $\tree{\graph,\block_2}\subseteq\tree{\graph,\block_1}$ and $\treew{\graph,\block_2}\le\treew{\graph,\block_1}$.
\end{claim}

Recalling that protocol $\proto_{\sf GHAST}$ is the same as $\proto_{\sf TG}$ except the block weight. In $\proto_{\sf TG}$, the validity of a block requires $\mathrm{Pivot}(\block.\past)\circ\block=\mathrm{Chain}(\block)$. So the ghast protocol inherits such property. 

\begin{claim}
    For any block $\block$, we have $\mathrm{Pivot}(\block.\past)\circ\block=\mathrm{Chain}(\block)$.
\end{claim}

\begin{claim}
    Let $\mathbf{P}$ be a list of blocks in which any two consecutive blocks are in parent/child relation. (The outputs of $\chain{\block'}$ for some block $\block'$ and the chain $\vcmp$ in the adversary state satisfy such requirement.)
    
    For any $\block\in\mathbf{P}\backslash\{\tip(\mathbf{P})\}$, let $\block_1:=\nxt(\mathbf{P},\block)$, it will be $\block_1\neq\bot$ and $\block_1.\parent=\block$.
\end{claim}

\subsection{Properties for concepts}\label{sec:proof:lma:onebest}

Before touching the potential value, we prove some important properties for concepts first. 

\begin{proofof}{Lemma~\ref{lma:onebest}}
    For any block $\block\in\gmax$ and its children blocks $\block_1,\block_2\in\child(\gmax,\block)$ ($\block_1\neq\block_2$). Since $\sibtreew{\graph,\block}$ (defined in eq.~(\ref{def:sibsubtw}))returns the maximum subtree weight of sibling blocks in $\block$, and block $\block_1$ and $\block_2$ are on sibling relationship, we can claim $\sibtreew{\gmax,\block_1}\ge \treew{\gmax,\block_2}$ and $\sibtreew{\gmax,\block_2}\ge \treew{\gmax,\block_1}$.
    Because block $\flagb$ is in $\gdta$, we have $\gmin\cup\{\flagb\}\subseteq\gmax$. Thus 
    \begin{align*}
        & \mathrm{Adv}((\gmax,\gmin,\flagb),\block_1)+\mathrm{Adv}((\gmax,\gmin,\flagb),\block_2)\\
        \le & \treew{\gmax,\block_1}-\treew{\gmax,\block_2}
        + \treew{\gmax,\block_2}-\treew{\gmax,\block_1}\\
        \le & 0
    \end{align*}
    Thus it can not be $\mathrm{Adv}((\gmax,\gmin,\flagb),\block_1)>0$ and $\mathrm{Adv}((\gmax,\gmin,\flagb),\block_2)>0$.
\end{proofof}

\begin{lemma}\label{lma:eptinpivot}
    \statecondm. If flag block $\flagb\l= \bot$, then $\vcmp\l$ must be a prefix of pivot chain $\mathrm{Pivot}(\event.\mathsf{block})$.  
\end{lemma}
\begin{proof}
    Let $\block:=\event.\mathsf{block}$, $\graph:=\block.\past$. 
    Since $\flagb\l= \bot$, for any block $\block'\in \gmax\l$, $\mathrm{Adv}(\advs\l,\block')=\treew{\gmin\l,\block'}-\sibtreew{\gmax\l,\block'}$. 
    By claim~\ref{clm:chainc}.1, for any block $\block'\in \vcmp\l$, $\mathrm{Adv}(\advs\l,\block')>0$.
    Since $\graph\l$ is the local state of the honest node who generates $\block$, according to claim~\ref{clm:graphrel}, $\gmin\l\subseteq\graph\l\subseteq\gmax\l$. So we have
    $$ \forall \block'\in \vcmp\l, \treew{\graph\l,\block'}-\sibtreew{\graph\l,\block'}\ge \treew{\gmin\l,\block'}-\sibtreew{\gmax\l,\block'}>0.$$
    According to lemma~\ref{lma:ghost}, block $\tip(\vcmp\l)$ must be in $\mathrm{Pivot}(\graph\l)$ and thus $\vcmp\l$ is a prefix of pivot chain $\mathrm{Pivot}(\graph\l)$ when $\flagb\l=\bot$. Since $\mathrm{Pivot}(\graph\l)\circ\block=\chain{{\block}}$, $\vcmp\l$ is also a prefix of $\chain{{\block}}$.
\end{proof}

\begin{lemma}\label{lma:treewbound}
    \stateconds,
    any graph $\graph$ satisfying $\gmin\subseteq\graph\subseteq\gmax$ and any block $\block$ which is not genesis block. We have
    $$ \treew{\graph,\block}-\sibtreew{\graph,\block}\le \mathrm{Adv}(\advs,\block) + \treew{\gdta,\block.\parent}.$$

\end{lemma}
\begin{proof}
    %
    %
    By the definition of $\treew{\cdot}$, for any block $\block$, we have 
    $$\treew{\gmax,\block}-\treew{\gmin,\block}=\treew{\gdta,\block}.$$

    As for the upper bound of $\sibtreew{\gmax,\block}-\sibtreew{\gmin,\block}$ for any block $\block$ except the genesis block $\genesisblock$. 
    If $\child(\gmax,\block.\parent)\backslash\{\block\}$ is empty set, then $\child(\gmin,\block.\parent)\backslash\{\block\}$ will also be empty set and $\sibtreew{\gmax,\block}-\sibtreew{\gmin,\block}=0$. 
    Otherwise let $\block'$ be the block with maximum subtree weight in $\child(\gmax,\block'.\parent)\backslash\{\block'\}$. Then $\sibtreew{\gmax,\block'}=\treew{\gmax,\block'}$. 
    If $\block'\notin \child(\gmin,\block.\parent)\backslash\{\block\}$, it must be $\block'\in\gdta$. So all the blocks in subtree of $\block'$ not appears in $\gmin$, and thus $\tree{\gmin,\block'}=0$. 
    If $\block'\in \child(\gmin,\block.\parent)\backslash\{\block\}$, then $\sibtreew{\gmin,\block'}\ge \treew{\gmin,\block'}$. In summary for the case $\block'$ exists, we have 
    $$ \sibtreew{\gmax,\block'}-\sibtreew{\gmin,\block'}\le \treew{\gmax,\block'}-\treew{\gmin,\block'}=\treew{\gdta,\block'}.$$
    Recalling that $\block'.\parent=\block.\parent$, we have the following result for both cases that $\block'$ exists or not.
    \begin{align*}
       \sibtreew{\gmax,\block}-\sibtreew{\gmin,\block}
       \le \treew{\gdta,\block.\parent}-\treew{\gdta,\block}.
    \end{align*}

    Notice that $\mathrm{Adv}(\advs,\block)\ge \treew{\gmin,\block}-\sibtreew{\gmax,\block}$ by definition. We have
    \begin{align*}
         \mathrm{Adv}(\advs,\block) + \treew{\gdta,\block.\parent}
        \ge & \treew{\gmax,\block}-\sibtreew{\gmin,\block}\\
        \ge & \treew{\graph,\block}-\sibtreew{\graph,\block}
    \end{align*}
\end{proof}

\begin{lemma}\label{lma:flagadv}
    \statecondm. If $\event$ is the $\hgen$ event of a flag block (a.k.a. $\flagb\l=\bot$ and $\flagb=\event.\mathsf{block}$), then we have $\forall \block' \in \chain{\event.\mathsf{block}}, \mathrm{Adv}(\advs,\block')>\heavyw-\kah-\kam$. 
\end{lemma}
\begin{proof}
    We denote block $\event.\mathsf{block}$ by $\block$ and denote graph $\block.\mathsf{past}$ by $\graph\l$. Since $\block$ is an honest block, $\graph\l$ must be the local state of an honest node when adversary state is $\advs$, thus $\gmin\l\subseteq\graph\l\subseteq\gmax\l$ by claim~\ref{clm:graphrel}. 
    Since $\mathrm{Chain}(\block)\backslash\{\block\}=\mathrm{Pivot}(\graph\l)$, according to the necessary condition of pivot chain, for any block $\block'$ in $\mathrm{Chain}(\block)\backslash\{\block\}$, $\treew{\graph\l,\block'}-\sibtreew{\graph\l,\block'}\ge 0$. 
    Since $\block.\parent$ is the last block of the pivot chain $\mathrm{Pivot}(\graph\l)$, $\graph\l$ must have no child block in $\graph\l$. Thus $\sibtreew{\graph\l,\block}=0$. So we have 
    \begin{align*}
        \forall \block' \in \mathrm{Chain}(\block),\treew{\graph\l,\block'}-\sibtreew{\graph\l,\block'} \ge 0.
    \end{align*}

    Since $\flagb\l=\bot$ and $\gmax\l\subseteq\graph\l\subseteq\gmin\l$, according to lemma~\ref{lma:treewbound}, we have
    \begin{align*}
        \forall \block' \in \mathrm{Chain}(\block)\backslash\{\genesisblock\}, \quad & 
        \mathrm{Adv}(\advs\l,\block')+\treew{\gdta\l,\block'.\parent} \\
        \ge &\treew{\graph\l,\block'}-\sibtreew{\graph\l,\block'} \\
        \ge &0.
    \end{align*}

     According to our rule in maintaining the flag block, when $\flagb\l=\bot$ and $\flagb=\event.\mathsf{block}$, $\gdta\l$ has no honest block with block weight $\heavyw$ and $\mathrm{Spe}(\advs\l)=\false$. 
     Let $\tmpset\l:=\tree{\gdta\l,\tip(\vcmp\l)}$, according to the definition of special status (definition~\ref{def:special}), we have $\tw(\tmpset\l\cap \mathbf{M}\l) <\kam$ and $|\left\{\block \in \gdta\l\backslash \mathbf{M}\l\;|\;\block.\mathsf{weight}=1\right\}|< \kah$. 
     Since there is no honest block with block weight $\heavyw$ in $\gdta\l$, $|\left\{\block \in \gdta\l\backslash \mathbf{M}\l\;|\;\block.\mathsf{weight}=\heavyw\right\}|=0$. 
     Thus $ \tw(\tmpset\l)<\kam+\kah.$ 
     Since $\flagb\l=\bot$, according to lemma~\ref{lma:eptinpivot}, $\vcmp\l$ is a prefix of $\mathrm{Chain}(\block)$. For any block $\block'$ in $ \mathrm{Chain}(\block) \backslash \vcmp\l$, we have $\tree{\gdta\l,\block'.\parent}\subseteq\tmpset\l$. Thus, \todo{Recheck calling special status}
     $$ \forall \block'\in \mathrm{Chain}(\block) \backslash \vcmp\l, \treew{\gdta\l,\block'.\parent}\le \tw(\tmpset)<\kam+\kah.$$
    
     Since genesis block $\genesisblock$ must be in $\vcmp\l$, for all the blocks $\block'$ in $\mathrm{Chain}(\block) \backslash \vcmp\l$, the previous two inequalities give $\mathrm{Adv}(\advs\l,\block')>-\kah-\kam$. According to claim~\ref{clm:chainc}.1, for the blocks $\block'$ in $\vcmp\l$, we have $\mathrm{Adv}(\advs\l,\block')>0$. Thus
     $$ \forall \block'\in \mathrm{Chain}(\block), \mathrm{Adv}(\advs\l,\block')>-\kah-\kam.$$
   
     According to the rule in updating adversary state, $\hgen$ event does not modify $\gmin$, so $\gmin\l=\gmin$. 
     By claim~\ref{clm:flagb}, $\flagb\notin \gmin$.
     For all the blocks $\block'$ in $\chain{\block}$, $\tree{\{\block\},\block'}=\{\block\}$. 
     Note that $\block$ and $\flagb$ refers the same block, and the block weight of flag block is $\heavyw$. So we have
     $$ \treew{\gmin\cup\{\flagb\},\block'}=\treew{\gmin\l,\block'} +\heavyw.$$
     
     For any block $\block'\in \chain{\block}$ and $\block'_1\in \child(\gmax,\block'.\parent)\backslash\{\block'\}$, it must be $\treew{\gmax\l,\block'_1}=\treew{\gmax,\block'_1}$ because $\gmax$ is one block $\block$ different from $\gmax\l$ and it is not in subtree of $\block'_1$. Thus 
     $$\forall\block'\in \chain{\block}, \sibtreew{\gmax,\block'}=\sibtreew{\gmax\l,\block'}.$$ 

     Summarize all the previous results, we have
     \begin{align*}
        \forall \block'\in\chain{\block}, \quad
        &\treew{\gmin\cup\{\flagb\},\block'}-\sibtreew{\gmax,\block'}\\
        = &\treew{\gmin\l,\block'}-\sibtreew{\gmax\l,\block'}+\heavyw \\
        > & \heavyw - \kam -\kah
     \end{align*}
\end{proof}

Now we will show that when the adversary state is updated from $\advs\l$ to $\advs$, one of $\vcmp\l$ and $\vcmp$ must be the prefix of another. 

\begin{lemma}\label{lma:tipmove1}
    \statecondm. One of $\vcmp\l$ and $\vcmp$ must be the prefix of another. 
\end{lemma}
\begin{proof}
    We prove this property by contradiction. We assume $\vcmp\l$ is not the prefix of $\vcmp$ and $\vcmp$ is not the prefix of $\vcmp\l$. Let $\block_{\sf c}$ be the last block of common prefix between $\vcmp$ and $\vcmp\l$. Let $\block_1:=\nxt(\vcmp\l,\block_{\sf c})$, $\block_2:=\nxt(\vcmp,\block_{\sf c})$. There should be $\block_1\neq\block_2$, $\block_1\neq\bot$ and $\block_2\neq \bot$. According to claim \ref{clm:chainc}.1, $\mathrm{Adv}(\advs\l,\block_1)>0$, $\mathrm{Adv}(\advs,\block_2)>0$. Since $\block_1$ and $\block_2$ have the same parent block, according to lemma~\ref{lma:onebest}, $\mathrm{Adv}(\advs\l,\block_2)\le 0$, $\mathrm{Adv}(\advs,\block_1)\le 0$. Thus 
    $$\mathrm{Adv}(\advs,\block_1)-\mathrm{Adv}(\advs\l,\block_1)<0\;\;\emph{ and }\;\;\mathrm{Adv}(\advs,\block_2)-\mathrm{Adv}(\advs\l,\block_2)>0.$$

    Recalling that $\mathrm{Adv}(\advs,\block)$ is defined by $\treew{\gmin\cup\{\flagb\},\block}-\sibtreew{\gmax,\block}.$ We discuss in two cases with $\flagb\l=\flagb$ and $\flagb\l\neq\flagb$. 

    \textbf{Case 1:} $\flagb\l=\flagb$, claim~\ref{clm:flagb} guarantees $\flagb\notin \gmin$ and $\flagb\l\notin \gmin\l$ when $\flagb\neq\bot$. So for any block $\block$
    \begin{align*}
        &\mathrm{Adv}(\advs,\block)-\mathrm{Adv}(\advs\l,\block)\\
        =&\left(\treew{\gmin,\block}-\treew{\gmin\l,\block}\right)-\left(\sibtreew{\gmax,\block}-\sibtreew{\gmax\l,\block}\right).
    \end{align*}

    \textbf{Case 1.1:} $\event$ is $\hgen$ or $\mrls$ event. 
    
    It will be $\gmin\l=\gmin$ and $\gmax\l\subseteq \gmax$. So for any block $\block$, $\mathrm{Adv}(\advs,\block)-\mathrm{Adv}(\advs\l,\block)\le 0$. It can not be $\mathrm{Adv}(\advs,\block_2)-\mathrm{Adv}(\advs\l,\block_2)>0$.
    
    \textbf{Case 1.2:} $\event$ is $\allrec$ event. 
    
    It will be $\gmin\l\subseteq\gmin$ and $\gmax\l= \gmax$. So for any block $\block$, $\mathrm{Adv}(\advs,\block)-\mathrm{Adv}(\advs\l,\block)\ge 0$. It can not be $\mathrm{Adv}(\advs,\block_1)-\mathrm{Adv}(\advs\l,\block_1)<0$.
    
    \textbf{Case 1.3:} $\event$ is a $\mgen$ event. 
    
    It will be $\gmin\l=\gmin$ and $\gmax\l= \gmax$. So for any block $\block$, $\mathrm{Adv}(\advs,\block)-\mathrm{Adv}(\advs\l,\block)=0$. It can not be $\mathrm{Adv}(\advs,\block_2)-\mathrm{Adv}(\advs\l,\block_2)>0$.
    $\\$

    \textbf{Case 2:} $\flagb\l\neq\flagb$. There could be three possible sub-cases according to the rule updating $\flagb$. 

    \textbf{Case 2.1:} $\event.\mathsf{type}=\hgen$, $\flagb\l\neq\bot$ and $\flagb=\bot$. 
    
    It will be $\gmin\l\cup\{\flagb\l\}\subseteq\gmin\cup\{\flagb\}$ and $\gmax\l\subseteq \gmax$. So for any block $\block$, $\mathrm{Adv}(\advs,\block)-\mathrm{Adv}(\advs\l,\block)\le 0$. It can not be $\mathrm{Adv}(\advs,\block_2)-\mathrm{Adv}(\advs\l,\block_2)>0$.

    \textbf{Case 2.2:} $\event.\mathsf{type}=\allrec$, $\flagb\l=\event.\mathsf{block}$ and $\flagb=\bot$. 
    
    $\flagb\l$ is included into $\gmin$. So $\gmin\l\cup\{\flagb\l\}=\gmin\cup\{\flagb\}$. An $\allrec$ event does not update $\gmax$, so $\gmax\l=\gmax$. Thus for any block $\block$,  $\mathrm{Adv}(\advs,\block)-\mathrm{Adv}(\advs\l,\block)=0$. It can not be $\mathrm{Adv}(\advs,\block_2)-\mathrm{Adv}(\advs\l,\block_2)>0$.

    \textbf{Case 2.3:}  $\event.\mathsf{type}=\hgen$, $\flagb\l=\bot$ and $\flagb=\event.\mathsf{block}$.
    
    Let $\graph_{\sf f}=\flagb.\mathsf{past}$. According to lemma~\ref{lma:flagadv}, $\forall \block'\in \chain{\flagb}$, $\mathrm{Adv}(\advs,\block')>\heavyw-\kam-\kah\ge \kam+\kah$. (Recalling that we require $\heavyw\ge 2\kam+2\kah$.\todo{Formal recalling}) 
    According to claim~\ref{clm:chainc}.3, $\chain{\flagb}$ is a prefix of $\vcmp$. 
    Since $\flagb\neq\bot$ and $\event$ is the $\hgen$ event of $\flagb$, according to lemma~\ref{lma:eptinpivot}, $\vcmp\l$ is a prefix of $\chain{\flagb}$. 
    Thus $\vcmp\l$ is a prefix of $\vcmp$. 
    $\\$

   As a summary, the cases except case 2.3 are proved by contradiction. For the case 2.3, we show  $\vcmp\l$ is a prefix of $\vcmp$ directly.
\end{proof}

\begin{lemma}\label{lma:tipmove2}
    \statecondm. 
    
    $\tip(\vcmp\l)\prec\tip(\vcmp)$ only if one of the following condition holds:
    \begin{itemize}[nosep]
        \item $\event.\mathsf{type}=\allrec$ and $\flagb\l=\flagb$.
        \item $\event.\mathsf{type}=\hgen$, $\flagb\l=\bot$ and $\flagb=\event.\mathsf{block}$.
    \end{itemize}

    $\tip(\vcmp)\prec\tip(\vcmp\l)$ only if one of the following condition holds:
    \begin{itemize}[nosep]
        \item $\event.\mathsf{type}\in \{\hgen,\mrls\}$ and $\flagb\l=\flagb$.
        \item $\event.\mathsf{type}=\hgen$, $\flagb\l\neq\bot$ and $\flagb=\bot$.
    \end{itemize}
\end{lemma}
\begin{proof}
    If $\tip(\vcmp\l)\prec\tip(\vcmp)$, let $\block:=\nxt(\vcmp,\tip(\vcmp\l))$. 
    According to claim~\ref{clm:chainc}.2, since $\block$ is the child block of $\tip(\vcmp\l)$, $\mathrm{Adv}(\advs\l,\block)\le \kah+\kam$. According to claim~\ref{clm:chainc}.2, since $\tip(\vcmp\l)\prec \block$ and $\block\in\vcmp$, so $\mathrm{Adv}(\advs,\block)> \kah+\kam$. So $\mathrm{Adv}(\advs,\block)-\mathrm{Adv}(\advs\l,\block)>0$. 

    In the proof of lemma~\ref{lma:tipmove1}, there is no block satisfying $\mathrm{Adv}(\advs,\block)-\mathrm{Adv}(\advs\l,\block)>0$ in case 1.1, 1.3, 2.1, 2.2. Thus $\tip(\vcmp\l)\prec\tip(\vcmp)$ only if one of the following condition holds.
    \begin{itemize}[nosep]
        \item $\event.\mathsf{type}=\allrec$ and $\flagb\l=\flagb$.
        \item $\event.\mathsf{type}=\hgen$, $\flagb\l=\bot$ and $\flagb=\event.\mathsf{block}$.
    \end{itemize}
    $\\$

    If $\tip(\vcmp)\prec\tip(\vcmp\l)$, let $\block:=\nxt(\vcmp\l,\tip(\vcmp))$. 
    Since $\block$ is the child block of $\tip(\vcmp\l)$ and $\block\preceq \tip(\vcmp)$, according to claim~\ref{clm:chainc}.1, $\mathrm{Adv}(\advs,\block)\le 0$. Since $\block\in\vcmp\l$, according to claim~\ref{clm:chainc}.2, $\mathrm{Adv}(\advs\l,\block)> 0$. So $\mathrm{Adv}(\advs,\block)-\mathrm{Adv}(\advs\l,\block)<0$. 

    In the proof of lemma~\ref{lma:tipmove1}, there is no block satisfying $\mathrm{Adv}(\advs,\block)-\mathrm{Adv}(\advs\l,\block)<0$ in case 1.2, 1.3, 2.2. It also shows $\vcmp\l$ is a prefix of $\vcmp$ in case 2.3. Thus $\tip(\vcmp)\prec\tip(\vcmp\l)$ only if one of the following condition holds.
    \begin{itemize}[nosep]
        \item $\event.\mathsf{type}\in \{\hgen,\mrls\}$ and $\flagb\l=\flagb$.
        \item $\event.\mathsf{type}=\hgen$, $\flagb\l\neq\bot$ and $\flagb=\bot$.
    \end{itemize}
\end{proof}

\begin{lemma}\label{lma:alwaysheavy}
    \statecondm. Let $\block_{\sf c}:=\tip(\mathrm{Pivot}(\graph\l)\cap \vcmp\l)$. 
    
    If $\event.\mathsf{type}=\hgen$, $\mathrm{Spe}(\advs\l)=\false$ and $\mathrm{Old}(\gmin\l,\block)=\true$, then $\adp(\event.\mathsf{block})={\sf con}$.
\end{lemma}
\begin{proof}
    Let $\block:=\event.\mathsf{block}$ and $\graph\l:=\block.\mathsf{past}$. 
    Since $\graph\l$ is the local state of the honest node who generates $\block$, according to claim~\ref{clm:graphrel}, $\gmin\l\subseteq\graph\l\subseteq\gmax\l$. 
    Since $\mathrm{Old}(\gmin\l,\block)=\true$ and $\gmin\l\subseteq \graph\l$, it can be verified that 
    $$\mathrm{Old}(\graph\l,\block_{\sf c})=\true.$$
    
    Let $\block_{\sf c1}=\nxt(\mathrm{Pivot}(\graph\l),\block_{\sf c})$. 
    If $\block_{\sf c1}=\bot$, $\block_{\sf c}$ will be the last block in the pivot chain of $\graph\l$. Since we have $\mathrm{Old}(\graph\l,\block_{\sf c})=\true$, $\adp(\block)={\sf con}$ according to the second rule in definition~\ref{def:adaptive}.
    
    In the following, we discuss the case $\block_{\sf c1}\neq\bot$ according to whether $\vcmp\l$ is a prefix of $\mathrm{Pivot}(\graph\l)$ or not. Since $\block_{\sf c1}\neq \bot$, it is the child block of $\block_{\sf c}$.
    
    \textbf{Case 1:} $\vcmp\l$ is a prefix of (or equals to) $\mathrm{Pivot}(\graph\l)$.

    %
    In this case, $\block_{\sf c}=\tip(\vcmp\l)$. Since $\block_{\sf c1}$ is the child of $\block_{\sf c}$, 
    %
    according to claim~\ref{clm:chainc}.2, 
    $$\mathrm{Adv}(\advs\l,\block_{\sf c1})\le \kah+\kam.$$ 
    Let $\tmpset\l:=\tree{\gdta\l,\tip(\vcmp\l)}$, according to the definition of special status (definition~\ref{def:special}), we have $\tw(\tmpset\l\cap \mathbf{M}\l) <\kam$, $|\left\{\block' \in \gdta\l\backslash \mathbf{M}\l\;|\;\block'.\mathsf{weight}=1\right\}|< \kah$ and  $|\left\{\block' \in \gdta\l\backslash \mathbf{M}\l\;|\;\block'.\mathsf{weight}=\heavyw\right\}|\le 2$. Thus 
    $$\tree{\gdta\l,\block_{\sf c1}.\parent}=\tw(\tmpset)< \kam+\kah+2\heavyw.$$
    Recalling that we require $\advan\ge 2\kam+2\kah+2\heavyw$. According to lemma~\ref{lma:treewbound}, 
    \begin{align*}
           &\treew{\graph\l,\block_{\sf c1}}-\sibtreew{\graph\l,\block_{\sf c1}} \\
        \le&  \mathrm{Adv}(\advs\l,\block_{\sf c1})+\tree{\gdta\l,\block_{\sf c1}.\parent} \\
        < &2\kam+2\kah+2\heavyw\\
        \le &\advan.
    \end{align*}
    Since we are given $\mathrm{Old}(\graph,\block_{\sf c})=\true$, we have $\adp(\block)={\sf con}$ according to the first rule in definition~\ref{def:adaptive}. 

    \textbf{Case 2:} $\vcmp\l$ is not a prefix of (or equals to) $\mathrm{Pivot}(\graph\l)$.

    In this case, we claim that $\block_{\sf c}$ is not the last block in $\vcmp\l$. 
    Let $\block_{\sf c2}:=\nxt(\vcmp\l,\block_{\sf c})$. 
    According to claim~\ref{clm:chainc}.1, $\mathrm{Adv}(\advs\l,\block_{\sf c2})>0$. 
    Since $\block_{\sf c1}$ and $\block_{\sf c2}$ are in sibling relations and $\sibtreew{\cdot}$ returns the maximum subtree weight among sibling blocks (defined in eq.~(\ref{def:sibsubtw})), 
    we have $\treew{\graph\l,\block_{\sf c1}}\le \sibtreew{\graph\l,\block_{\sf c2}}$ and $\treew{\graph\l,\block_{\sf c2}}\le \sibtreew{\graph\l,\block_{\sf c1}}$. 
    (For the special case $\block_{\sf c2}\notin \graph\l$, these two inequalities also hold.) So we have 
    \begin{align*}
        &\treew{\graph\l,\block_{\sf c1}}-\sibtreew{\graph\l,\block_{\sf c1}}\\
     \le&  \sibtreew{\graph\l,\block_{\sf c2}}-\treew{\graph\l,\block_{\sf c2}} \\
     \le & \sibtreew{\gmax\l,\block_{\sf c2}}-\treew{\gmin\l\cup\{\flagb\l\},\block_{\sf c2}} + \heavyw \\
     < &\heavyw\\
     \le &\advan.
 \end{align*}
 Since we are given $\mathrm{Old}(\graph,\block_{\sf c})=\true$, we have $\adp(\block)={\sf con}$ according to the first rule in definition~\ref{def:adaptive}. 
\end{proof}


\subsection{Case discussions for potential value (Part 1)}\label{sec:potcase}

\paragraph{Common settings}
\statecondm. 
In this sub-section, we study the upper bound of block potential value difference $\pot(\advs,\block)-\pot(\advs\l,\block)$ under the assumption that $\pot(\advs\l,\block)\neq\bot$ and $\pot(\advs,\block)\neq\bot$. All the lemmas in this section assumes $\pot(\advs\l,\block)\neq\bot$ and $\pot(\advs,\block)\neq\bot$. We will also not repeat them in the following lemmas (except lemma~\ref{lma:potcase:final}) since they are not referred outside this section. The upper bound of 

We define symbol $\mainchild\l:=\nxt(\vcmp\l,\block)$ and $\mainchild:=\nxt(\vcmp,\block)$ for given $\vcmp\l$, $\vcmp$ and $\block$ in the context. So $\mainchild$ and $\mainchild\l$ are the child of block $\block$ when they are not $\bot$. According to lemma~\ref{lma:tipmove1}, one of $\vcmp\l$ and $\vcmp$ must be the prefix of another. So if $\mainchild\l\neq\bot$ and $\mainchild\neq\bot$, there must be $\mainchild\l=\mainchild$. These property are frequently used and we do not explicitly refer them in the following. 
We define symbol $\mathbf{N}:=\{\block' \in \gmax\backslash\mathbf{M}\;|\;\block'.\weight=1\}$ for given $\gmax$ and $\mathbf{M}$ in the context 
and we define $\mathbf{N}\l$ similarly. 

Now we start the case discussion for three components $\potwith,\potadv,\potspe$.

\subsubsection{The first component}

\begin{lemma}\label{lma:potcase:1}
    If $\event$ is $\mgen$ event, then $\potwith(\advs,\block)-\potwith(\advs\l,\block)\le \event.{\sf block}.\weight$.

    If $\event$ is $\mrls$ event and $\block\preceq\event.\mathsf{block}$, then $\potwith(\advs,\block)-\potwith(\advs\l,\block)\le -\event.{\sf block}.\weight$.

    For other cases, $\potwith(\advs,\block)-\potwith(\advs\l,\block)=0$.
\end{lemma}
\begin{proof}
    Let $\eblock=\event.\mathsf{block}$. Recalling that $\potwith(\advs,\block)=\treew{\ggen\backslash\gmax,\block}$. 

    If $\event$ is $\mgen$ event, we have $\ggen\backslash\ggen\l=\{\eblock\}$ and $\gmax=\gmax\l$. So 
    $\potwith(\advs,\block)-\potwith(\advs\l,\block)\le \eblock.\weight.$

    If $\event$ is $\mrls$ event, we have $\ggen=\ggen\l$, $\gmax\backslash\gmax\l=\{\eblock\}$ and $\eblock\in\ggen$. 
    So $\ggen\l\backslash\gmax\l$ has one more element $\eblock$ than $\ggen\backslash\gmax$. 
    When $\eblock\preceq \block$, $\eblock\in \tree{\ggen\l\backslash\gmax\l,\block}$ and thus $\potwith(\advs,\block)-\potwith(\advs\l,\block)\le -\eblock.\weight.$ Otherwise $\potwith(\advs,\block)-\potwith(\advs\l,\block)=0.$

    If $\event$ is $\hgen$ event, we have $\ggen\backslash\ggen\l=\{\eblock\}$ and $\gmax\backslash\gmax\l=\{\eblock\}$. Thus $\ggen\l\backslash\gmax\l=\ggen\backslash\gmax$  and $\potwith(\advs,\block)-\potwith(\advs\l,\block)=0.$

    If $\event$ is $\allrec$ event, we have $\ggen=\ggen\l$ and $\gmax=\gmax\l$. Thus  $\potwith(\advs,\block)-\potwith(\advs\l,\block)=0.$
\end{proof}

\subsubsection{The second component}

In discussing the component $\potadv$, for the case $\flagb=\flagb\l$, we discuss the upper bound of  $\potadv(\advs,\mainchild)-\potadv(\advs,\mainchild\l)$ and $\potadv(\advs,\mainchild\l)-\potadv(\advs\l,\mainchild\l)$ respectively and combine them later. 
%



%

\begin{lemma}\label{lma:potcase:2.1.1}
    If $\event$ is $\hgen$ event, $\flagb\l=\flagb$, $\mathsf{Spe}(\advs\l)=\false$ and $\event.\mathsf{block}.\weight=1$, then it will be 
    $$\potadv(\advs,\mainchild\l)-\potadv(\advs\l,\mainchild\l)\le -1.$$
\end{lemma}
\begin{proof}
    Let $\eblock=\event.\mathsf{block}$. If $\mainchild\l=\bot$, we have $\block=\tip(\vcmp\l)$. According to lemma~\ref{lma:eptinpivot}, $\vcmp\l$ is a prefix of $\eblock$. Thus $\block=\tip(\vcmp\l\cap\chain{\eblock})$. Since $\pot(\advs\l,\block)\neq\bot$, we have $\mathrm{Old}(\gmin\l,\block)=\true$. According to lemma~\ref{lma:alwaysheavy}, $\adp(\eblock)={\sf con}$. So $\eblock.\weight\neq 1$ according to equation~\ref{eq:blockweight}. This contradicts the assumption $\eblock.\weight=1$ in this lemma. Thus, $\mainchild\l$ cannot be $\bot$.

    Since $\event.\type=\hgen$ and $\eblock.\weight=1$, $\mathbf{N}$ must has one more element $\eblock$ compared to $\mathbf{N}\l$. 
    Since $\mainchild\l\in\vcmp\l$ and $\vcmp\l$ is a prefix of $\eblock$, it will be $\mainchild\l\preceq \eblock$. Thus $\tree{\gdta,\mainchild\l}$ has one more element $\eblock$ compared to $\tree{\gdta\l,\mainchild\l}$. 
    Since $\mathrm{Spe}(\advs\l)=\false$, according to the definition of special status, $\gdta\l$ has at most $\kah-1$ honest blocks with block weight 1. Thus $|\tree{\gdta\l,\mainchild\l}|\le \kah-1$. So we have 
    $$ \tw(\tree{\gdta,\mainchild\l}\cap\mathbf{N})=\tw(\tree{\gdta\l,\mainchild\l}\cap\mathbf{N}\l)+1\le \kah.$$
    
    Since $\mainchild\l\preceq \eblock$ and $\gmax$ has one more element $\eblock$ than $\gmax\l$, all the sibling blocks of $\mainchild\l$ has the same subtree in $\gmax$ and $\gmax\l$. So we have $\sibtreew{\gmax,\eblock}=\sibtreew{\gmax\l,\eblock}$. We also have $\gmin\cup\{\flagb\}=\gmin\l\cup\{\flagb\l\}$ because $\event.\type\neq \allrec$ and $\flagb=\flagb\l$. Thus
    $$\mathrm{Adv}(\advs,\mainchild\l)-\mathrm{Adv}(\advs\l,\mainchild\l)=0.$$

    Recalling that $\potadv(\advs,\mainchild)=\kah+\kam-\mathrm{Adv}(\advs,\mainchild)-\min\{\kah,\tw(\tree{\gdta,\mainchild}\cap\mathbf{N})\}$. Thus we have 
    $$ \mathrm{Adv}(\advs,\mainchild\l)-\mathrm{Adv}(\advs,\mainchild\l)=-1.$$

\end{proof}

\begin{lemma}\label{lma:potcase:2.1.2}
    If $\event$ is a $\mrls$ event, $\flagb=\flagb\l$, $\mainchild\l\neq\bot$ and $\event.\mathsf{block}\in \tree{\gmax,\block}\backslash\tree{\gmax,\mainchild\l}$. We have
    $$\potadv(\advs,\mainchild\l)-\potadv(\advs\l,\mainchild\l)\le \event.\mathsf{block}.\weight.$$

    For all the other cases satisfying $\flagb=\flagb\l$, 
    $$\potadv(\advs,\mainchild\l)-\potadv(\advs\l,\mainchild\l)\le 0.$$
\end{lemma}
\begin{proof}
    Let $\eblock:=\event.\mathsf{block}$. In this proof, we try to find all the cases with $\potadv(\advs,\mainchild\l)-\potadv(\advs\l,\mainchild\l)>0$. When $\mainchild\l=\bot$, both $\potadv(\advs,\mainchild\l)$ and $\potadv(\advs\l,\mainchild\l)$ equal to $0$. So we only focus on the case with $\mainchild\l\neq\bot$. 


    We study the difference between $\tree{\gdta\l,\mainchild\l}\cap\mathbf{N}\l$ and $\tree{\gdta,\mainchild\l}\cap\mathbf{N}$. Since $\event$ does not a $\hgen$ event, we have $\mathbf{N}=\mathbf{N}\l$. Thus
    \begin{align*}
        &\min\{\tw(\tree{\gdta\l,\mainchild\l}\cap\mathbf{N}\l),\kah\}-\min\{\tw(\tree{\gdta,\mainchild\l}\cap\mathbf{N}),\kah\}\\
    \le & \max\{0,\tw(\tree{\gdta\l,\mainchild\l}\cap\mathbf{N})-\tw(\tree{\gdta,\mainchild\l}\cap\mathbf{N})\} \\
   \le   & \treew{(\gdta\l\backslash\gdta)\cap \mathbf{N},\mainchild\l}
    \end{align*}
    By the definition of $\potadv$, we have
    $$ \potadv(\advs,\mainchild\l)-\potadv(\advs\l,\mainchild\l) \le \mathrm{Adv}(\advs\l,\mainchild\l)-\mathrm{Adv}(\advs,\mainchild\l)+\treew{(\gdta\l\backslash\gdta)\cap \mathbf{N},\mainchild\l}.$$
    Notice that $\gdta\l\backslash\gdta=(\gmax\l\backslash\gmin\l)\backslash(\gmax\backslash\gmin)$. If $\gmin=\gmin\l$, then $\gdta\l\backslash\gdta\subseteq(\gmax\l\backslash\gmax)=\emptyset$. So 
    $\gdta\l\backslash\gdta\neq \emptyset$ only if $\gmin\neq\gmin\l$ and thus $\event$ must be an $\allrec$ event. We can claim $\treew{(\gdta\l\backslash\gdta)\cap \mathbf{N},\mainchild\l}>0$ only if $\event$ is an $\allrec$ event and $\mainchild\l\preceq\eblock$. In this case, $\gmax\l=\gmax$ and $\gdta\subseteq \gdta\l$, so we have
    \begin{align*}
        &\!\!\!\!\!\text{If }\treew{(\gdta\l\backslash\gdta)\cap \mathbf{H},\mainchild\l}>0,\\
        &\mathrm{Adv}(\advs\l,\mainchild\l)-\mathrm{Adv}(\advs,\mainchild\l)+\treew{(\gdta\l\backslash\gdta)\cap \mathbf{H},\mainchild\l} \\
        \le & \treew{\gmin\l,\mainchild\l}-\treew{\gmin,\mainchild\l}+\treew{\gdta\l,\mainchild\l}-\treew{\gdta,\mainchild\l}\\
        = & \treew{\gmax\l,\mainchild\l}-\treew{\gmax,\mainchild\l} \\
        = & 0
    \end{align*}
    Thus we have 
    $$ \potadv(\advs,\mainchild\l)-\potadv(\advs\l,\mainchild\l) \le \max\{0,\mathrm{Adv}(\advs\l,\mainchild\l)-\mathrm{Adv}(\advs,\mainchild\l)\}.$$

    Now we only need to study in which cases $\mathrm{Adv}(\advs\l,\mainchild\l)>\mathrm{Adv}(\advs,\mainchild\l)$. 
    It can only happens when $\sibtreew{\gmax,\mainchild\l}>\sibtreew{\gmax\l,\mainchild\l}$. 
    Thus it must be 
    $$\event.\type\in\{\hgen,\mrls\}\quad\text{ and }\quad\eblock\in\tree{\gmax,\block}\backslash\tree{\gmax,\mainchild\l}.$$ 

    If $\event.\type=\mrls$, since $\gmax$ and $\gmax\l$ differ at one block $\eblock$, $\sibtreew{\gmax,\mainchild\l}-\sibtreew{\gmax\l,\mainchild\l}\le \eblock.\weight$ and thus
    $$\mathrm{Adv}(\advs\l,\mainchild\l)-\mathrm{Adv}(\advs,\mainchild\l)\le \eblock.\weight.$$

    If $\event.\type=\hgen$, recalling that block $\block$ is the parent of $\mainchild\l$ and $\mainchild\l\in\vcmp\l$, we have $\tip(\chain{\eblock}\cap \vcmp\l)=\block$. Since $\pot(\advs\l,\block)\neq\bot$, we have $\mathrm{Old}(\gmin\l,\block)=\true$. According to lemma~\ref{lma:alwaysheavy}, $\adp(\eblock)={\sf con}$. So $\eblock.\weight$ equals to 0 or $\heavyw$. We have three sub-cases as follows.
    \begin{enumerate}[nosep]
        \item  If $\eblock.\weight=0$, then $\gmax$ and $\gmax\l$ differ at one block with zero block weight. So $\sibtreew{\gmax,\mainchild\l}=\sibtreew{\gmax\l,\mainchild\l}$.
        \item  If $\eblock.\weight=\heavyw$ and $\flagb\neq\bot$, since $\eblock$ is an honest block and $\event.\type=\hgen$, according to the rule in updating flag block, we have $\flagb=\bot\neq\flagb$. This contradicts to our assumption. 
        \item  If $\flagb=\bot$, since $\eblock$ is an honest block and $\event.\type=\hgen$, according to lemma~\ref{lma:eptinpivot}, $\vcmp\l$ should be a prefix of $\chain{\eblock}$ and thus $\mainchild\l\in \chain{\eblock}$. This contradicts $\eblock\notin \tree{\gmax,\mainchild\l}$. 
    \end{enumerate}
    In all, if $\event.\type=\hgen$, all the three sub-cases cannot be $\mathrm{Adv}(\advs\l,\mainchild\l)>\mathrm{Adv}(\advs,\mainchild\l)$. 
    $\\$

    As a summary for the whole proof, for the case $\flagb=\flagb\l$, $\potadv(\advs,\mainchild\l)-\potadv(\advs\l,\mainchild\l)> 0$ only if $\event.\type=\mrls$, $\mainchild\l\neq\bot$ and $\event.\mathsf{block}\in \tree{\gmax,\block}\backslash\tree{\gmax,\mainchild\l}$. For this case, we have $$\potadv(\advs,\mainchild\l)-\potadv(\advs\l,\mainchild\l)\le \event.\mathsf{block}.\weight.$$
\end{proof}

\begin{lemma}\label{lma:potcase:2.2.1}
    If $\mainchild\l\neq\bot$ and $\mainchild=\bot$, then
    $$ \potadv(\advs,\mainchild)-\potadv(\advs,\mainchild\l)\le -\kam.$$
    
    For the other cases, 
    $$ \potadv(\advs,\mainchild)-\potadv(\advs,\mainchild\l)\le 0.$$
\end{lemma}
\begin{proof}
    If $\mainchild=\mainchild\l$, then $\potadv(\advs,\mainchild)-\potadv(\advs,\mainchild\l)=0$ holds trivially. 

    If $\mainchild\neq\mainchild\l$, then one of $\mainchild$ and $\mainchild\l$ equals to $\bot$. (Recalling that $\mainchild=\mainchild\l$ if $\mainchild\l\neq\bot$ and $\mainchild\neq\bot$.)

    \textbf{Case 1:} $\mainchild=\bot$ and $\mainchild\l\neq \bot$. 

    Since $\mainchild=\bot$, $\block$ should be $\mathrm{Tip}(\vcmp)$. Since block $\mainchild\l$ is the child block of $\block$ and $\mainchild\l\preceq\tip(\vcmp\l)$, according to claim~\ref{clm:chainc}.2, $\mathrm{Adv}(\advs,\mainchild\l)\le 0$. Since $\min\{\tw(\tree{\gmax,\mainchild\l}\cap\mathbf{N}),\kah\}\le \kah$ and $\kah\ge 0$, we have 
    $$ \potadv(\advs,\mainchild)-\potadv(\advs,\mainchild\l)=-\potadv(\advs,\mainchild\l)\le \mathrm{Adv}(\advs,\mainchild\l)+\kah-\kah-\kam\le -\kam.$$
    
    \textbf{Case 2:} $\mainchild\neq\bot$ and $\mainchild\l= \bot$. 

    Since $\mainchild\l=\bot$, $\block$ should be $\mathrm{Tip}(\vcmp\l)$. Since block $\mainchild$ is the child block of $\block$, we have $\mathrm{Tip}(\vcmp\l)\prec\mainchild$, according to claim~\ref{clm:chainc}.1, $\mathrm{Adv}(\advs,\mainchild)>\kah+\kam$. Since $\min\{\tw(\tree{\gmax,\mainchild}\cap\mathbf{N}),\kah\}\ge 0$, we have 
    $$ \potadv(\advs,\mainchild)-\potadv(\advs,\mainchild\l)=\potadv(\advs,\mainchild)\le \kah+\kam-\mathrm{Adv}(\advs,\mainchild)<0.$$
\end{proof}

\begin{lemma}\label{lma:potcase:2.3.1}
    If $\event.\type=\hgen$, $\event.\mathsf{block}.\weight=\heavyw$, $\flagb\l=\bot$ and $\flagb=\event.\mathsf{block}$, then 
    $$\potadv(\advs,\mainchild)-\potadv(\advs\l,\mainchild\l)\le 2\kah+2\kam-\heavyw.$$ 
\end{lemma}
\begin{proof}
    Let $\eblock:=\event.\mathsf{block}$. 
    If $\event.\type=\hgen$, $\flagb\l=\bot$ and $\flagb=\event.\mathsf{block}$, 
    according to lemma~\ref{lma:flagadv}, 
    for all the block $\block'$ in $\chain{\eblock}$, 
    we have $\mathrm{Adv}(\advs,\block')> \heavyw-\kah-\kam\ge \kah+\kam$. 
    According to claim~\ref{clm:chainc}.3, $\chain{\eblock}$ is a prefix of $\vcmp$.
    Since $\eblock\notin \gmin$, $\eblock$ cannot be an old enough block and thus $\pot(\advs,\eblock)=\bot$. 
    So $\block$ cannot be the last block in $\vcmp$. 
    We claim $$\mainchild\neq\bot.$$

    \textbf{Case 1:}  $\mainchild\l\neq\bot$. It will be $\mainchild=\mainchild\l$

     Recalling that $\mainchild\in\chain{\eblock}$, so the subtree weight for the sibling blocks of $\mainchild$ does not change. We have $\sibtreew{\gmax,\mainchild}=\sibtreew{\gmax\l,\mainchild}$. Since $\event.\type=\hgen$, $\flagb\l=\bot$ and $\flagb=\event.\mathsf{block}$, $\gmin\cup\{\flagb\}$ has one more block $\eblock$ than $\gmin\l\cup\{\flagb\l\}$. So $\treew{\gmin\l\cup\{\flagb\l\},\mainchild}-\treew{\gmin\cup\{\flagb\},\mainchild}=-\eblock.\weight$. Since $\eblock=\flagb$, the block weight of $\eblock$ must be $\heavyw$. Thus 
    $$\mathrm{Adv}(\advs\l,\mainchild)-\mathrm{Adv}(\advs,\mainchild\l)=-\heavyw.$$

    Since $\eblock.\weight=\heavyw$ and $\mathbf{N}$ only contains honest blocks with block weight 1, we have $\tree{\gdta\l,\mainchild}\cap \mathbf{N}\l = \tree{\gdta,\mainchild}\cap \mathbf{N}$. So in this case,
    $$\potadv(\advs,\mainchild)-\potadv(\advs\l,\mainchild\l)\le \mathrm{Adv}(\advs\l,\mainchild)-\mathrm{Adv}(\advs,\mainchild)+\kah=\kah-\heavyw.$$

    \textbf{Case 2:}  $\mainchild\l=\bot$.

    Recalling that $\mainchild\in\chain{\eblock}$ and so $\mathrm{Adv}(\advs,\mainchild)>\heavyw-\kah-\kam$ according to lemma~\ref{lma:flagadv}. Thus 
    $$\potadv(\advs,\mainchild)-\potadv(\advs\l,\mainchild\l)\le \kah+\kam-\mathrm{Adv}(\advs,\mainchild)<2\kah+2\kam-\heavyw.$$
\end{proof}

\begin{lemma}\label{lma:potcase:2.3.2}
    If $\event.\type=\hgen$, $\event.\mathsf{block}.\weight=\heavyw$, $\flagb\l \neq\bot$ and $\flagb=\bot$, we have 
    $$ \potadv(\advs,\mainchild)-\potadv(\advs\l,\mainchild\l)\le \heavyw.$$
\end{lemma}
\begin{proof}
    Let $\eblock:=\event.\mathsf{block}$ and $\graph\l:=\eblock.\past$. Since $\eblock$ is an honest block, according to claim~\ref{clm:graphrel}, $\gmin\l\subseteq\graph\l\subseteq\gmax\l$. Since $\event.\mathsf{block}.\weight=\heavyw$, $\flagb\l \neq\bot$ and $\flagb=\bot$, according to lemma~\ref{lma:tipmove2}, it cannot be $\tip(\vcmp\l)\prec\tip(\vcmp)$. Thus it cannot be $\mainchild\l=\bot\wedge\mainchild\neq\bot$. 

    \textbf{Case 1:} $\mainchild\l\neq\bot$ and $\mainchild\neq\bot$. It will be $\mainchild=\mainchild\l$. 

    Since $\eblock.\weight=\heavyw$ and $\mathbf{N}$ only contains honest blocks with block weight 1, we have $\tree{\gdta\l,\mainchild}\cap \mathbf{N}\l = \tree{\gdta,\mainchild}\cap \mathbf{N}$. So in this case, 
    $$ \potadv(\advs,\mainchild)-\potadv(\advs\l,\mainchild\l)\le \mathrm{Adv}(\advs\l,\mainchild\l)-\mathrm{Adv}(\advs,\mainchild).$$
    Since $\mainchild\in\vcmp$, according to claim~\ref{clm:chainc}.1, $\forall \block'\in \chain{\mainchild},\mathrm{Adv}(\advs,\mainchild)>0$. Since $\flagb=\bot$, we have $\gmin\cup\{\flagb\}=\gmin$. Since $\event.\type=\hgen$, we have $\gmin=\gmin\l$. Recalling that $\gmin\l\subseteq\graph\l\subseteq\gmax\l$, we have
    \begin{align*}
        \forall \block'\in \chain{\mainchild}, \quad & \treew{\graph\l,\block'}-\sibtreew{\graph\l,\block'} \\
        \ge & \treew{\gmin\l,\block'}-\sibtreew{\gmax\l,\block'} \\
        \ge & \treew{\gmin,\block'}-\sibtreew{\gmax,\block'}\\
        >   & 0
    \end{align*}
    According to lemma~\ref{lma:ghost}, $\mainchild\in\mathrm{Pivot}(\graph\l)$. Thus $\mainchild\prec\eblock$. Since $\gmax$ and $\gmax\l$ only differs at block $\eblock$, we have $\sibtreew{\gmax\l,\mainchild}=\sibtreew{\gmax,\mainchild}$. Recalling that $\gmin\l=\gmin$ and $\flagb=\bot$, we have
    $$ \mathrm{Adv}(\advs\l,\mainchild\l)-\mathrm{Adv}(\advs,\mainchild)=\treew{\gmin\l\cup\{\flagb\l\},\mainchild}-\treew{\gmin\cup\{\flagb\},\mainchild}\le \heavyw.$$  

    Combined with the first inequality in this proof, we have proved $\potadv(\advs,\mainchild)-\potadv(\advs\l,\mainchild\l)\le\heavyw$ for this case. 

    \textbf{Case 2:} $\mainchild\l\neq\bot$ and $\mainchild=\bot$. 

    We prove this case by showing that $\mathrm{Adv}(\advs\l,\mainchild\l)\le \heavyw$. If $\mathrm{Adv}(\advs\l,\mainchild\l)> \heavyw$, recalling that $\gmin\l\subseteq\graph\l\subseteq\gmax\l$, we have 
    \begin{align*}
        & \treew{\graph\l,\mainchild\l}-\sibtreew{\graph\l,\mainchild\l}\\
        \ge &\treew{\gmin\l,\mainchild\l}-\sibtreew{\gmax\l,\mainchild\l} \\
        \ge &\treew{\gmin\l\cup \{\flagb\},\mainchild\l} -\heavyw-\sibtreew{\gmax\l,\mainchild\l}\\
        > & 0.
    \end{align*}

    So $\mainchild\l$ is the child with maximum subtree weight of $\block$. If $\block\in\mathrm{Pivot}(\graph\l)$, there must be $\mainchild\l\in \mathrm{Pivot}(\graph\l)$. Thus no matter $\block$ belongs to $\mathrm{Pivot}(\graph\l)$ or not, the sibling blocks of $\mainchild\l$ are not in $\mathrm{Pivot}(\graph\l)$ and $\chain{\eblock}$. So we have $\sibtreew{\gmax\l,\mainchild\l}=\sibtreew{\gmax,\mainchild\l}$. ($\gmax$ and $\gmax\l$ only differs at block $\eblock$.) Recalling that $\gmin=\gmin\l$ and $\flagb=\bot$, we have
    $$ \mathrm{Adv}(\advs\l,\mainchild\l)-\mathrm{Adv}(\advs,\mainchild)=\treew{\gmin\l\cup\{\flagb\l\},\mainchild}-\treew{\gmin\cup\{\flagb\},\mainchild}\le \heavyw.$$  
    
    Thus $\mathrm{Adv}(\advs,\mainchild)\ge \mathrm{Adv}(\advs\l,\mainchild\l)-\heavyw>0$. Since $\pot(\advs,\block)\neq \bot$, we have $\block\in\vcmp$. Since $\mainchild\l\preceq \tip(\vcmp\l)$, $\mathrm{Adv}(\advs,\mainchild\l)>0$ and $\mainchild\l.\parent=\block$, $\mainchild\l$ should also belong to $\vcmp$ according to our rule in maintaining chain $\vcmp$. This contradicts to $\mainchild=\bot$. 

    So in this case, there must be $\mathrm{Adv}(\advs\l,\mainchild\l)\le \heavyw$. And thus
    $$ \potadv(\advs,\mainchild)-\potadv(\advs\l,\mainchild\l)\le -\kah-\kam+\mathrm{Adv}(\advs\l,\mainchild\l)+\kah\le \heavyw-\kam.$$

    \textbf{Case 3:} $\mainchild\l=\bot$ and $\mainchild=\bot$. 

    In this case, $ \potadv(\advs,\mainchild)-\potadv(\advs\l,\mainchild\l)=0<\heavyw$.
\end{proof}

\begin{lemma}\label{lma:potcase:2.3.3}
    If $\event.\type=\allrec$, $\event.\mathsf{block}.\weight=\heavyw$, $\flagb\l=\event.\mathsf{block}$ and $\flagb=\bot$, we have
    $$ \pot(\advs,\mainchild)-\pot(\advs\l,\mainchild\l)=0.$$
\end{lemma}
\begin{proof}
    Let $\eblock:=\event.\mathsf{block}$. According to lemma~\ref{lma:tipmove2}, we have $\vcmp\l=\vcmp$ if  $\event.\type=\allrec$, $\flagb\l=\event.\mathsf{block}$ and $\flagb=\bot$. So 
    $$\mainchild\l=\mainchild.$$ 
    Since $\event$ is the $\allrec$ event of $\flagb\l$ and $\flagb=\bot$, we have $\gmin\l\cup\{\flagb\l\}=\gmin\cup\{\flagb\}$. Since $\eblock.\weight=\heavyw$ and $\tree{\gdta,\mainchild}\cap \mathbf{N}$ only contains honest blocks with block weight 1, we have $\tree{\gdta\l,\mainchild\l}\cap \mathbf{N}\l = \tree{\gdta,\mainchild}\cap \mathbf{N}$. So in this case, 
    $$ \pot(\advs,\mainchild)-\pot(\advs\l,\mainchild\l)=0.$$
\end{proof}

\subsubsection{The third component}

Recalling that the third component of potential value is defined as follows. Let $\mainchild=\nxt(\vcmp,\block)$, 
\begin{align*}
    \potspe(\advs ,\block)
    &:=
    \left\{\begin{array}{ll}
        \tw(\tree{\gdta,\mainchild} \cap \mathbf{M}\backslash\speset) & \mainchild\neq \bot \\
        0 & \mainchild = \bot
    \end{array}\right.
\end{align*}


Here we define another intermediate potential value for the third component. Let $\mainchild\l=\nxt(\vcmp\l,\block)$, 
\begin{align*}
    \potspe'(\advs,\advs\l,\block)
    &:=
    \left\{\begin{array}{ll}
        \tw(\tree{\gdta,\mainchild\l} \cap \mathbf{M}\backslash\speset\l) & \mainchild\l\neq \bot \\
        \tw(\tree{\gdta,\block} \cap \mathbf{M}\backslash\speset\l) & \mainchild\l = \bot
    \end{array}\right.
\end{align*}


\begin{lemma}\label{lma:potcase:3.1.1}
    If $\event.\type=\mgen$ and one of the following properties hold,
    \begin{itemize}[nosep]
        \item $\mainchild\l\neq\bot$ and $\mainchild\l\preceq\event.\mathsf{block}$
        \item $\mainchild\l=\bot$ and $\block\preceq\event.\mathsf{block}$
    \end{itemize}
    then we have 
    $$ \potspe'(\advs,\advs\l,\block)-\potspe(\advs\l,\block)\le \event.\mathsf{block}.\weight.$$
    For all the other cases, 
    $$ \potspe'(\advs,\advs\l,\block)-\potspe(\advs\l,\block)\le 0.$$
\end{lemma}
\begin{proof}
    Let $\eblock=\event.\mathsf{block}$. In this proof, we try to figure out the cases with $ \potspe'(\advs,\advs\l,\block)-\potspe(\advs\l,\block)> 0.$

    \textbf{Case 1:} $\mainchild\l=\bot$.

    In this case, we have $\block=\tip(\vcmp\l)$. According to claim~\ref{clm:spesv}, $\tree{\gdta\l,\block}\cap\mathbf{M}\l\subseteq\speset\l$. Since $\mathbf{M}$ and $\mathbf{M}\l$ only differs at blocks which have not been generated when the adversary state is $\advs\l$, we have $\tree{\gdta\l,\block}\cap\mathbf{M}\l=\tree{\gdta\l,\block}\cap\mathbf{M}.$
    \begin{align*}
        & \potspe'(\advs,\advs\l,\block)-\potspe(\advs\l,\block) \\
        = & \tw(\tree{\gdta,\block}\cap\mathbf{M}\backslash\speset\l)\\
        \le & \tw\left(\left(\tree{\gdta,\block}\cap\mathbf{M}\right)\backslash\left(\tree{\gdta\l,\block}\cap\mathbf{M}\right)\right)\\
        = & \tw(\tree{\gdta\backslash\gdta\l,\block}\cap\mathbf{M})
    \end{align*}
    $\tree{\gdta\backslash\gdta\l,\block}\cap\mathbf{M}\neq\emptyset$ only if $\event.\type=\mrls$ and $\block\preceq\eblock$. Since $\gdta$ and $\gdta\l$ can only differ at block $\eblock$, we have 
    $$ \potspe'(\advs,\advs\l,\block)-\potspe(\advs\l,\block)\le \eblock.\weight.$$

    \textbf{Case 2:} $\mainchild\l\neq\bot$.
    
    Since $\mathbf{M}$ and $\mathbf{M}\l$ only differs at blocks which have not been generated when the adversary state is $\advs\l$, we have $\tree{\gdta\l,\mainchild\l}\cap\mathbf{M}\l=\tree{\gdta\l,\mainchild\l}\cap\mathbf{M}.$

    \begin{align*}
        & \potspe'(\advs,\advs\l,\block)-\potspe(\advs\l,\block) \\
        = & \tw(\tree{\gdta,\mainchild\l}\cap\mathbf{M}\backslash\speset\l)-\tw(\tree{\gdta\l,\mainchild\l}\cap\mathbf{M}\backslash\speset\l) \\
        \le & \tw(\tree{\gdta\backslash\gdta\l,\mainchild\l}\cap\mathbf{M}\backslash\speset\l) \\
        \le & \tw(\tree{\gdta\backslash\gdta\l,\mainchild\l}\cap\mathbf{M}).
    \end{align*}

    Similar with case 1, $\tree{\gdta\backslash\gdta\l,\mainchild\l}\cap\mathbf{M}\neq\emptyset$ only if $\event.\type=\mrls$ and $\block\preceq\eblock$. And we have 
    $$ \potspe'(\advs,\advs\l,\block)-\potspe(\advs\l,\block)\le \eblock.\weight.$$

    $\\$

    As a summary, $\potspe'(\advs,\advs\l,\block)-\potspe(\advs\l,\block)\le \eblock.\weight$ always holds. And $\potspe'(\advs,\advs\l,\block)-\potspe(\advs\l,\block)>0$ only if $\mainchild\l\neq\bot\wedge \mainchild\l\preceq \eblock$ or $\mainchild\l=\bot\wedge \block\preceq\eblock$. 
\end{proof}

\begin{lemma}\label{lma:potcase:3.2.1}
    If $\mainchild\l\neq\bot$ and $\mainchild=\bot$, we have 
    $$ (\potspe(\advs,\block)+\spevalue)-(\potspe'(\advs,\advs\l,\block)+\spevalue\l)\le\kam.$$
    For other cases, 
    $$ (\potspe(\advs,\block)+\spevalue)-(\potspe'(\advs,\advs\l,\block)+\spevalue\l)\le 0.$$
\end{lemma}
\begin{proof}
    We prove this lemma under four cases partitioned by whether $\mainchild=\bot$ and whether $\mainchild\l=\bot$. According to claim~\ref{clm:spesv}.1, we have $\spevalue-\spevalue\l\le \tw(\speset\backslash\speset\l)$.

    \textbf{Case 1:} $\mainchild\l=\bot$ and $\mainchild=\bot$. 

    In this case, according to claim~\ref{clm:spesv}.4, we have $\tree{\gdta,\tip(\vcmp)}\cap\mathbf{M}\subseteq \speset$. Since $\mainchild\l=\bot$, we have $\tip(\vcmp)=\block$. Thus
    \begin{align*}
        & (\potspe(\advs,\block)+\spevalue)-(\potspe'(\advs,\advs\l,\block)+\spevalue\l) \\
        = & \spevalue-\spevalue\l-\tw(\tree{\gdta,\block}\cap\mathbf{M}\backslash\speset\l) \\
        \le & \spevalue-\spevalue\l-\tw(\speset\backslash\speset\l) \\
        \le &0.
    \end{align*}

    \textbf{Case 2:} $\mainchild\l\neq\bot$ and $\mainchild=\bot$. 

    According to claim~\ref{clm:spesv}.2, we have $\spevalue-\spevalue\l\le \kam.$ Thus 
    \begin{align*}
        & (\potspe(\advs,\block)+\spevalue)-(\potspe'(\advs,\advs\l,\block)+\spevalue\l) \\
        = & \spevalue-\spevalue\l-\tw(\tree{\gdta,\mainchild\l}\cap\mathbf{M}\backslash\speset)\\
        \le & \kam.
    \end{align*}

    \textbf{Case 3:} $\mainchild\l=\bot$ and $\mainchild\neq\bot$. 

    According to claim~\ref{clm:spesv}.4, we have $\speset\l\subseteq\speset$ and $\speset\backslash\speset\l\subseteq\tree{\gdta,\tip(\vcmp)}\cap\mathbf{M}$. Since $\mainchild\in \vcmp$, we have $\tree{\gdta,\tip(\vcmp)}\subseteq \tree{\gdta,\mainchild}$. Thus 
    \begin{align*}
        & (\potspe(\advs,\block)+\spevalue)-(\potspe'(\advs,\advs\l,\block)+\spevalue\l) \\
        = & \spevalue-\spevalue\l+\tw(\tree{\gdta,\mainchild}\cap\mathbf{M}\backslash\speset)-\tw(\tree{\gdta,\block}\cap\mathbf{M}\backslash\speset\l) \\
        \le & \spevalue-\spevalue\l+\tw(\tree{\gdta,\mainchild}\cap\mathbf{M}\backslash\speset)-\tw(\tree{\gdta,\mainchild}\cap\mathbf{M}\backslash\speset\l) \\
        = & \spevalue-\spevalue\l-\tw(\tree{\gdta,\mainchild}\cap\mathbf{M}\cap\left(\speset\backslash\speset\l\right)) \\
        = & \spevalue-\spevalue\l-\tw(\speset\backslash\speset\l) \\
        \le & 0.
    \end{align*}
    \textbf{Case 4:} $\mainchild\l\neq\bot$ and $\mainchild\neq\bot$. It will be $\mainchild=\mainchild\l$. For the same reason as case 3, we have
    \begin{align*}
        & (\potspe(\advs,\block)+\spevalue)-(\potspe'(\advs,\advs\l,\block)+\spevalue\l) \\
        = & \spevalue-\spevalue\l+\tw(\tree{\gdta,\mainchild}\cap\mathbf{M}\backslash\speset)-\tw(\tree{\gdta,\mainchild}\cap\mathbf{M}\backslash\speset\l) \\
        \le & 0.
    \end{align*}
\end{proof}

\subsubsection{Collect the case discussions}

Now we collect the previous results and gives the upper bound for block potential value $\pot(\advs,\block)-\pot(\advs\l,\block)$. Similar with the discussion in the second and the third component, we define an intermediate block potential value as 
$$ \pot'(\advs,\advs\l,\block):=\potwith(\advs,\block)+\potadv(\advs,\mainchild\l)+\potspe'(\advs,\advs\l,\block).$$

\begin{lemma}\label{lma:potsum:1}
    If $\flagb=\flagb\l$, we have 
    $$ \pot'(\advs,\advs\l,\block)-\pot(\advs\l,\block)\le\eventweight(\advs\l,\event).$$
\end{lemma}
\begin{proof}
    First, we define
    \begin{align*}
        \potwith^\Delta&:=\potwith(\advs,\block)-\potwith(\advs\l,\block), \\
        \potadv^{\sf \Delta1}&:=\potadv(\advs,\mainchild\l)-\potadv(\advs\l,\mainchild\l), \\
        \potspe^{\sf \Delta1}&:=\potspe'(\advs,\advs\l,\block)-\potspe(\advs\l,\block). 
    \end{align*}
    Thus we have
    $$ \pot'(\advs,\advs\l,\block)-\pot(\advs\l,\block)=\potwith^\Delta+\potadv^{\sf \Delta1}+ \potspe^{\sf \Delta1}$$

    With the assumption $\flagb\l=\flagb$, we category all the possible cases for $\event$ as follows.
    \begin{itemize}[nosep]
        \item \textbf{Case 1:} $\event.\type=\mgen$.
        \item \textbf{Case 2:} $\event.\type=\allrec$. 
        \item \textbf{Case 3.1:} $\event.\type=\mrls$, $\mainchild\l\neq\bot$ and $\event.\mathsf{block}\in \tree{\gmax,\block}\backslash\tree{\gmax,\mainchild\l}$.
        \item \textbf{Case 3.2:} $\event.\type=\mrls$, $\mainchild\l\neq\bot$ and $\event.\mathsf{block}\in \tree{\gmax,\mainchild\l}$.
        \item \textbf{Case 3.3:} $\event.\type=\mrls$, $\mainchild\l=\bot$ and $\event.\mathsf{block}\in \tree{\gmax,\block}$.
        \item \textbf{Case 3.4:} $\event.\type=\mrls$ and $\event.\mathsf{block}\notin \tree{\gmax,\block}$.
        \item \textbf{Case 4.1:} $\event.\type=\hgen$ and $\event.\mathsf{block}.\weight= 0$. 
        \item \textbf{Case 4.2:} $\event.\type=\hgen$ and $\event.\mathsf{block}.\weight= 1$. 
        \item \textbf{Case 4.3:} $\event.\type=\hgen$, $\event.\mathsf{block}.\weight= \heavyw$ and $\mathrm{Spe}(\advs\l)=\true$. 
        \item \textbf{Case 4.4:} $\event.\type=\hgen$,  $\event.\mathsf{block}.\weight= \heavyw$ and $\mathrm{Spe}(\advs\l)=\false$. According to our rule in updating flag block, if $\flagb\l=\flagb$ in this case, there must be $\flagb\l=\flagb=\bot$ and that $\gdta\l$ has an honest block with block weight $\heavyw$.
    \end{itemize}
  
    \begin{table}[!htbp]
        \begin{center}
            Let $w:=\event.\mathsf{block}.\weight$. \\
            \begin{tabular}{|l|l|l|l||l|}
            \hline
                & $\potwith^\Delta$  & $\potadv^{\sf \Delta1}$   & $\potspe^{\sf \Delta1}$   & $\eventweight(\advs\l,\event)$ \\
            \hline
            \textbf{Case 1  }  & $w$   & 0     & 0     & $w$ \\
            \hline
            \textbf{Case 2  }  & 0     & 0     & 0     & 0 \\
            \hline
            \textbf{Case 3.1}  & $-w$  & $w$   & 0     & 0 \\
            \hline
            \textbf{Case 3.2}  & $-w$  & 0     & $w$   & 0 \\
            \hline
            \textbf{Case 3.3}  & $-w$  & 0     & $w$   & 0 \\
            \hline
            \textbf{Case 3.4}  & 0     & 0     & 0     & 0 \\
            \hline
            \textbf{Case 4.1}  & 0     & 0     & 0     & 0 \\
            \hline
            \textbf{Case 4.2}  & 0     & $-1$  & 0     & $-1$\\
            \hline
            \textbf{Case 4.3}  & 0     & 0     & 0     & 0 \\
            \hline
            \textbf{Case 4.4}  & 0     & 0     & 0     & 0 \\
            \hline
            \end{tabular}%
        \label{tab:addlabel}%
        \end{center}
        \caption{The upper bounds for each component under different cases (Lemma~\ref{lma:potsum:1})}\label{tab:potsum:1}
    \end{table}%

    Table~\ref{tab:potsum:1} shows the upper bounds for $\potwith^\Delta$, $\potadv^{\sf \Delta1}$ and $\potspe^{\sf \Delta1}$ under difference cases. $w$ denotes $\event.\mathsf{block}.\weight$. 
    The upper bounds for $\potwith^\Delta$ follow lemma~\ref{lma:potcase:1}. 
    The upper bounds for $\potadv^{\sf \Delta1}$ follow lemma~\ref{lma:potcase:2.1.2} except case 4.2, which follows lemma~\ref{lma:potcase:2.1.1}. 
    The upper bounds for $\potspe^{\sf \Delta1}$ follow lemma~\ref{lma:potcase:3.1.1}. The last column shows $\eventweight(\advs\l,\event)$ under different cases. For each case (each row in the table), we can check that $$\potwith^\Delta+\potadv^{\sf \Delta1}+\potspe^{\sf \Delta1}\le \eventweight(\advs\l,\event).$$ 
  
\end{proof}

\begin{lemma}\label{lma:potsum:2}
    If $\flagb\l=\flagb$, we have
    $$ (\pot(\advs,\block)+\spevalue)-(\pot'(\advs,\advs\l,\block)+\spevalue\l)\le 0.$$
\end{lemma}
\begin{proof}
    First we define 
    \begin{align*}
        \potadv^{\sf \Delta2}&:=\potadv(\advs,\mainchild)-\potadv(\advs,\mainchild\l), \\
        \potspe^{\sf \Delta2}&:=(\pot(\advs,\block)+\spevalue)-(\pot'(\advs,\advs\l,\block)+\spevalue\l). 
    \end{align*}

    If $\mainchild\l\neq\bot$ and $\mainchild=\bot$, we have $\potadv^{\sf \Delta2}\le-\kam$ (lemma~\ref{lma:potcase:2.2.1}) and $\potspe^{\sf \Delta2}\le\kam$ (lemma~\ref{lma:potcase:3.2.1}). 
    
    For the other cases, we have $\potadv^{\sf \Delta2}\le0$ (lemma~\ref{lma:potcase:2.2.1}) and $\potspe^{\sf \Delta2}\le0$ (lemma~\ref{lma:potcase:3.2.1}).
    
    Thus we have 
    
    $$  \pot(\advs,\block)-\pot'(\advs,\advs\l,\block)=\potadv^{\sf \Delta2}+\potspe^{\sf \Delta2}\le 0.$$
\end{proof}

\begin{lemma}\label{lma:potsum:3}
    If $\flagb\l\neq\flagb$, we have
    $$ (\pot(\advs,\block)+\spevalue)-(\pot'(\advs,\advs\l,\block)+\spevalue\l)\le \eventweight(\advs\l,\event).$$
\end{lemma}
\begin{proof}
    First we define 
    \begin{align*}
        \potwith^\Delta&:=\potwith(\advs,\block)-\potwith(\advs\l,\block), \\
        \potadv^\Delta&:=\potadv(\advs,\mainchild)-\potadv(\advs\l,\mainchild\l), \\
        \potspe^{\sf \Delta1}&:=\potspe'(\advs,\advs\l,\block)-\potadv(\advs\l,\block), \\
        \potspe^{\sf \Delta2}&:=(\potspe(\advs,\block)+\spevalue)-(\potspe'(\advs,\advs\l,\block)+\spevalue\l). 
    \end{align*}

    According to our rule in updating flag block, if $\flagb\l\neq\flagb$, $\event$ has three possible cases. 
    \begin{itemize}[nosep]
        \item \textbf{Case 1:} $\event.\type=\hgen$, $\event.\mathsf{block}.\weight=\heavyw$, $\flagb\l=\bot$ and $\flagb=\event.\mathsf{block}$. According to lemma~\ref{lma:tipmove2}, in this case, it cannot be $\mainchild\l\neq\bot \wedge \mainchild=\bot$.  
        \item \textbf{Case 2:} $\event.\type=\hgen$, $\event.\mathsf{block}.\weight=\heavyw$, $\flagb\l\neq\bot$ and $\flagb=\bot$. Notice that $\flagb\l\in \gdta\l$ is an honest block with block weight $\heavyw$.
        \item \textbf{Case 3:} $\event.\type=\allrec$, $\event.\mathsf{block}.\weight=\heavyw$, $\flagb\l=\event.\mathsf{block}$ and $\flagb=\bot$.  According to lemma~\ref{lma:tipmove2}, in this case, we have $\vcmp\l=\vcmp$ and thus $\mainchild\l=\mainchild$. 
    \end{itemize}

    \begin{table}[htbp]
        \centering
        \begin{tabular}{|l|l|l|l|l||l|}
        \hline
                & $\potwith^\Delta$  & $\potadv^\Delta$   & $\potspe^{\sf \Delta1}$ & $\potspe^{\sf \Delta2}$   & $\eventweight(\advs\l,\event)$ \\
        \hline
        \textbf{Case: 1  }  & 0     & $2\kah+2\kam-\heavyw$    & 0     & 0     & $2\kah+2\kam-\heavyw$  \\
        \hline
        \textbf{Case: 2  }  & 0     & $\heavyw$     & 0     & $\kam$    & $\heavyw+\kam$ \\
        \hline
        \textbf{Case: 3  }  & 0     & 0   & 0     & 0   & 0 \\
        \hline
        \end{tabular}%
        \caption{The upper bounds for each component under different cases (Lemma~\ref{lma:potsum:3})}\label{tab:potsum:3}
  \end{table}%
  
  Table~\ref{tab:potsum:3} shows the upper bounds for $\potwith^\Delta$, $\potadv^\Delta$, $\potspe^{\sf \Delta1}$ and $\potspe^{\sf \Delta2}$ under difference cases. 
  The upper bounds for $\potwith^\Delta$ follow lemma~\ref{lma:potcase:1}. 
  The upper bounds for $\potadv^\Delta$ follow lemma~\ref{lma:potcase:2.3.1} (case 1), lemma~\ref{lma:potcase:2.3.2} (case 2) and lemma~\ref{lma:potcase:2.3.3} (case 3).
  The upper bounds for $\potspe^{\sf \Delta1}$ follow lemma~\ref{lma:potcase:3.1.1}. 
  The upper bounds for $\potspe^{\sf \Delta2}$ follow lemma~\ref{lma:potcase:3.2.1}. 
  The last column shows $\eventweight(\advs\l,\event)$ under different cases. For each case (each row in the table), we can check that $$\potwith^\Delta+\potadv^\Delta+\potspe^{\sf \Delta1}+\potspe^{\sf \Delta2}\le \eventweight(\advs\l,\event).$$ 

\end{proof}

Now we reach the result for the section~\ref{sec:potcase}. 
\begin{lemma}\label{lma:potcase:final}
    \statecondm. For any block $\block$, if $\pot(\advs\l,\block)\neq\bot$ and $\pot(\advs,\block)\neq\bot$, we have 
    $$(\pot(\advs,\block)+\spevalue)-(\pot(\advs\l,\block)+\spevalue\l)\le\eventweight(\advs\l,\event).$$ 
\end{lemma}
\begin{proof}
    If $\flagb\l=\flagb$, we have $\pot'(\advs,\advs\l,\block)-\pot(\advs,\block\l)\le \eventweight(\advs\l,\event)$ (lemma~\ref{lma:potsum:1}) and $(\pot(\advs,\block)+\spevalue)-(\pot'(\advs,\advs\l,\block)+\spevalue\l)\le 0$ (lemma~\ref{lma:potsum:2}). 
    If $\flagb\l\neq\flagb$,  we have $(\pot(\advs,\block)+\spevalue)-(\pot(\advs\l,\block)+\spevalue\l)\le \eventweight(\advs\l,\event)$ (lemma~\ref{lma:potsum:3}). So we always have 
    $$ (\pot(\advs,\block)+\spevalue)-(\pot(\advs\l,\block)+\spevalue\l)\le\eventweight(\advs\l,\event).$$
\end{proof}
\subsection{Case discussions for potential value (Part 2)}\label{sec:cpotcase}

\paragraph{Common settings}
\statecondm. 
In this sub-section, we study the upper bound of block potential value for the case not covered in section~\ref{sec:potcase}. If $\pot(\advs\l,\block)=\bot$ and $\pot(\advs,\block)\neq\bot$, we can not estimate the upper bound of $\pot(\advs,\block)$ by $\pot(\advs,\block)-\pot(\advs\l,\block)$. We try to estimate $\pot(\advs,\block)$ by $\pot(\advs,\block)-\pot(\advs\l,\tip(\vcmp\l))$ instead.

In this sub-section, we assume  $\pot(\advs\l,\block)=\bot$, $\pot(\advs,\block)\neq\bot$ and $\mathrm{Old}(\gmin\l,\block)=\true$. All the lemmas in section~\ref{sec:cpotcase} are discussed under these assumptions. We will not repeat them in each lemma (except lemma~\ref{lma:cpotcase:final}). 

We define symbol $\mainchild:=\nxt(\vcmp,\block)$ for given $\vcmp$ and $\block$ in the context. 
According to lemma~\ref{lma:tipmove1}, one of $\vcmp\l$ and $\vcmp$ must be the prefix of another. 
Since $\pot(\advs\l,\block)=\bot$ and ${Old}(\gmin\l,\block)=\true$, we have $\block\notin\vcmp\l$. Since $\block\in\vcmp$, $\vcmp\l$ must be a prefix of $\vcmp$ and it is strictly shorter than $\vcmp$. It must be 
$$ \tip(\vcmp\l)\prec\block\preceq\tip(\vcmp).$$
According to lemma~\ref{lma:tipmove2}, we have the following claims. All the proofs will refer this claim implicitly.

\begin{claim}\label{clm:cpotcase}
    Under the common assumption of section~\ref{sec:cpotcase}, there could be one of following two cases for event $\event$. 
    \begin{itemize}
        \item $\event.\type=\allrec$ and $\flagb\l=\flagb$
        \item $\event=\hgen$, $\flagb\l=\bot$ and $\flagb=\event.\mathsf{block}$
    \end{itemize}
\end{claim}

\begin{lemma}\label{lma:cpotcase:1}
    $\potwith(\advs,\block)-\potwith(\advs,\tip(\vcmp\l))\le 0$. 
\end{lemma}
\begin{proof}
    Since $\event.\type\notin\{\mgen,\mrls\}$, it can be verified that $\ggen\l\backslash\gmax\l=\ggen\backslash\gmax.$ Thus $\potwith(\advs,\tip(\vcmp\l))=\potwith(\advs\l,\tip(\vcmp\l))$. Since $\tip(\vcmp\l)\prec\block$, $\tree{\ggen\l\backslash\gmax\l,\block}\subseteq \tree{\ggen\l\backslash\gmax\l,\tip(\vcmp\l)}$. So we have 
    $$\potwith(\advs,\block)-\potwith(\advs,\tip(\vcmp\l))\le 0.$$
\end{proof}

\begin{lemma}\label{lma:cpotcase:2}
    If $\event=\hgen$, $\flagb\l=\bot$ and $\flagb=\event.\mathsf{block}$, we have 
    $$ \potadv(\advs,\mainchild)-\potadv(\advs\l,\bot)\le 2\kah+2\kam-\heavyw.$$
    If $\event.\type=\allrec$ and $\flagb\l=\flagb$, we have 
    $$ \potadv(\advs,\mainchild)-\potadv(\advs\l,\bot)\le 0.$$
\end{lemma}
\begin{proof}
    We discuss two cases respectively.

    \textbf{Case 1:} $\event=\hgen$, $\flagb\l=\bot$ and $\flagb=\event.\mathsf{block}$.
    Let $\eblock:=\event.\mathsf{block}$. 
    If $\event.\type=\hgen$, $\flagb\l=\bot$ and $\flagb=\event.\mathsf{block}$, 
    according to lemma~\ref{lma:flagadv}, 
    for all the block $\block'$ in $\chain{\eblock}$, 
    we have $\mathrm{Adv}(\advs,\block')> \heavyw-\kah-\kam\ge \kah+\kam$. 
    According to claim~\ref{clm:chainc}.3, $\chain{\eblock}$ is a prefix of $\vcmp$.
    Since $\eblock\notin \gmin$, $\eblock$ cannot be an old enough block and thus $\pot(\advs,\eblock)=\bot$. 
    So $\block$ cannot be the last block in $\vcmp$. 
    We claim $$\mainchild\neq\bot\text{ and }\mainchild\in \vcmp.$$ 

    Since $\mainchild\in\vcmp$, according to lemma~\ref{lma:flagadv}, $\mathrm{Adv}(\advs,\mainchild)\ge \heavyw-\kah-\kam.$ Thus 
    $$ \potadv(\advs,\mainchild)-\potadv(\advs\l,\bot)\le \kah+\kam-\mathrm{Adv}(\advs,\mainchild)\le 2\kah+2\kam-\heavyw.$$

    \textbf{Case 2:} $\event=\allrec$ and $\flagb\l=\flagb$. 

    If $\mainchild=\bot$, then $\potadv(\advs,\mainchild)-\potadv(\advs\l,\bot)=0$ holds trivially. 

    If $\mainchild\neq\bot$, since $\mainchild\in \vcmp$ and $\tip(\vcmp\l)\prec\block\prec\mainchild$, according to claim~\ref{clm:chainc}.1, $\mathrm{Adv}(\advs,\mainchild)>\kah+\kam$. Thus 
    $$ \potadv(\advs,\mainchild)-\potadv(\advs\l,\bot)\le \kah+\kam-\mathrm{Adv}(\advs,\mainchild)<0.$$
\end{proof}

\begin{lemma}\label{lma:cpotcase:3}
    $(\potspe(\advs,\block)+\spevalue)-(\potspe(\advs\l,\tip(\vcmp\l))+\spevalue\l)\le 0$. 
\end{lemma}
\begin{proof}
    Since $\nxt(\vcmp\l,\tip(\vcmp\l))=\bot$, we have $\potspe(\advs,\tip(\vcmp\l))=\spevalue\l$. According to claim~\ref{clm:spesv}, 
    $$\tree{\gdta\l,\tip(\vcmp\l)}\cap\mathbf{M}\l\subseteq\speset\l.$$ 

    If $\event.\type=\hgen$, since $\eblock$ is not malicious block, $\tree{\gdta\l,\tip(\vcmp\l)}\cap\mathbf{M}\l$ must equal to \\$\tree{\gdta,\tip(\vcmp\l)}\cap\mathbf{M}$. If $\event.\type=\allrec$, then $\gdta\subseteq\gdta\l$, we have $\tree{\gdta,\tip(\vcmp\l)}\cap\mathbf{M}\l\subseteq\tree{\gdta\l,\tip(\vcmp\l)}\cap\mathbf{M}$. Recalling that $\tip(\vcmp\l)\prec\tip(\vcmp)$,  
    $$\tree{\gdta,\tip(\vcmp)}\cap\mathbf{M}\subseteq\tree{\gdta,\tip(\vcmp\l)}\cap\mathbf{M}\subseteq\speset\l.$$ 
    
    According to the rule in updating $\speset$, we have $\speset=\speset\l\cup \tree{\gdta,\tip(\vcmp)}\cap\mathbf{M}=\speset\l$. Thus $\spevalue-\spevalue\l=0$ according to claim~\ref{clm:spesv}.3. 
    
    If $\mainchild=\bot$, we have $$(\potspe(\advs,\block)+\spevalue)-(\potspe(\advs\l,\tip(\vcmp\l))+\spevalue\l)=0.$$
    
    If $\mainchild\neq\bot$, since $\tip(\vcmp\l)\prec\mainchild$, we have $$\tree{\gdta,\mainchild}\cap\mathbf{M}\subseteq\tree{\gdta,\tip(\vcmp\l)}\cap\mathbf{M}\subseteq\speset\l.$$ 
    Thus 
    $$(\potspe(\advs,\block)+\spevalue)-(\potspe(\advs\l,\tip(\vcmp\l))+\spevalue\l)= \tw(\tree{\gdta,\mainchild}\cap\mathbf{M}\backslash\speset\l)=0.$$
\end{proof}

\begin{lemma}\label{lma:cpotcase:final}
    \statecondm. For any block $\block$ with $\pot(\advs\l,\block)=\bot$, $\pot(\advs,\block)\neq\bot$ and $\mathrm{Old}(\gmin\l,\block)=\true$, we have 
    $$ (\pot(\advs,\block)+\spevalue)-(\pot(\advs\l,\tip(\vcmp\l))+\spevalue\l)\le \eventweight(\advs\l,\event).$$
\end{lemma}
\begin{proof}
    First we define
    \begin{align*}
        \potwith^\Delta&:=\potwith(\advs,\block)-\potwith(\advs\l,\tip(\vcmp\l)), \\
        \potadv^\Delta&:=\potadv(\advs,\mainchild)-\potadv(\advs\l,\bot), \\
        \potspe^\Delta&:=(\potspe(\advs,\block)+\spevalue)-(\potspe(\advs\l,\tip(\vcmp\l))+\spevalue\l). 
    \end{align*}

    Then we have 
    $$ (\pot(\advs,\block)+\spevalue)-(\pot(\advs\l,\tip(\vcmp\l))+\spevalue\l) = \potwith^\Delta+\potadv^\Delta+\potspe^\Delta.$$

    According two claim~\ref{clm:cpotcase}, we have two possible cases

    \textbf{Case 1:} $\event.\type=\allrec$ and $\flagb\l=\flagb$
    
    In this case, we have $\eventweight(\advs\l,\event) = 0$. According to lemma~\ref{lma:cpotcase:1}, \ref{lma:cpotcase:2} and \ref{lma:cpotcase:3}, we have $\potwith^\Delta\le 0$, $\potadv^\Delta\le 0$ and $\potspe^\Delta\le 0$. So we claim
    $$ \potwith^\Delta+\potadv^\Delta+\potspe^\Delta \le \eventweight(\advs\l,\event).$$

    \textbf{Case 2:} $\event=\hgen$, $\flagb\l=\bot$ and $\flagb=\event.\mathsf{block}$
    In this case, we have $\eventweight(\advs\l,\event) = 2\kah+2\kam-\heavyw$. According to lemma~\ref{lma:cpotcase:1}, \ref{lma:cpotcase:2} and \ref{lma:cpotcase:3}, we have $\potwith^\Delta\le 0$, $\potadv^\Delta\le 2\kah+2\kam-\heavyw$ and $\potspe^\Delta\le 0$. So we claim
    $$ \potwith^\Delta+\potadv^\Delta+\potspe^\Delta \le \eventweight(\advs\l,\event).$$
\end{proof}

\subsection{Proof of Theorem~\ref{thm:hisunchange}}

\begin{proof}
    For any $n\ge N(r_1)$, we have assumed the following two conditions. 
    \begin{itemize}
        \item $\cpot(\advs_n,\gmin_{N(r_0)})<-\heavyw$
        \item For any block $\block\in \ggen_n$ with $\gmin_{N(r_0)}\nsubseteq\block.\mathsf{past}$, it will be $\mathrm{Old}(\gmin_n,\block)=\true$. 
    \end{itemize}

    Let $\block_n'\in \vcmp_n$ be the last block satisfying $\gmin_{N(r_0)}\nsubseteq\block_n'.\mathsf{past}$ in $\vcmp_n$. According to the second assumption, $\mathrm{Old}(\gmin_n,\block_n')=\true$. According to the definition of global potential value, the block potential value of all the blocks in $\chain{\block_n'}$ are taken into considered. Thus $\cpot(\advs_n,\gmin_{N(r_0)})\ge \max_{\block\in \chain{\block_n'}} \pot(\advs_n,\block)$. So we have $\pot(\advs_n,\block_n')< -\heavyw$. According to the definition of block potential value, $\pot(\advs_n,\block_n')<0$ only if $\block_n'\neq\tip(\vcmp_n)$. So we can let $\block_n\eqdef\mathrm{Next}(\vcmp_n,\block_n')$. Since $\block_n'$ is the last block satisfying $\gmin_{N(r_0)}\nsubseteq\block_n'.\mathsf{past}$, there must be $\gmin_{N(r_0)}\subseteq\block_n.\mathsf{past}$.

    Notice that $\max_{\block\in \chain{\block_n'}} \pot(\advs_n,\block)\le \cpot(\advs_n,\gmin_{N(r_0)})<-\heavyw$, thus we have 
    \begin{align*}
        \forall \block \in \chain{\block_n'}, \pot(\advs_n,\block)<-\heavyw.
    \end{align*}
   
    Since $\pot(\advs_n,\block)\ge -\mathrm{Adv}(\advs_n,\mathrm{Next}(\vcmp_n,\block))$,
    we have
    \begin{align*}
        \forall \block \in \chain{\block_n}\backslash\{\genesisblock\},\; \mathrm{Adv}(\advs_n,\block)<-\heavyw.
    \end{align*}

    Since 
    \begin{align*}
         \mathrm{Adv}(\advs_n,\block)= & \treew{\gmin_n\cup\{\flagb_n\},\block}-\sibtreew{\gmax_n,\block} \\
        \le & \heavyw+\treew{\gmin_n,\block}-\sibtreew{\gmax_n,\block},
    \end{align*}
    we have 
    \begin{align*}
        \forall \block \in \chain{\block_n}\backslash\{\genesisblock\}, \treew{\gmin_n,\block}-\sibtreew{\gmax_n,\block}>0.
    \end{align*}
    According to lemma~\ref{lma:ghost}, $\block_n\in \mathrm{Pivot}(\state)$ for all the local state $\gmin_n\subseteq\state\subseteq\gmax_n$. It means that for any honest participant, $\block_n$ is in its pivot chain. Since $\gmin_{N(r_0)}\subseteq \block_n.\past$, for any block $\ti\block\in \gmin_{N(r_0)}$, its history is determined by block $\block_n$. Formally, $\mathrm{Prefix}(\order(\state),\ti\block)$ is a prefix of $\order(\block_n.\past)$. 

    Then we will show that for any $n\ge N(r_1)$, $\block_{n}$ and $\block_{n+1}$ are the same block. Since $\block_n$ and $\block_{n+1}$ is the first block $\block'$ satisfying $\gmin_{N(r_0)}\subseteq\block'.\past$ in $\vcmp_n$ and $\vcmp_{n+1}$. So there can not be $\block_n\prec \block_{n+1}$ or $\block_{n+1}\prec \block_{n}$. According to lemma~\ref{lma:tipmove1}, one of $\vcmp_n$ and $\vcmp_{n+1}$ must be the prefix of another, thus there must be $\block_n=\block_{n+1}$.
    
    So for all the $n\ge N(r_1)$, block $\block_n$ refers the same block and it is in the pivot chain of all the honest participants. According to the block ordering algorithm $\order_{\sf GHAST}$ (recalling that $\order_{\sf GHAST}$ is defined the same as $\order_{\sf TG}$ except the block weight), the history of blocks in $\block_n.\past$ must be a prefix of $\order(\block_n.\past)$. Notice that $\gmin_{N(r_0)}\subseteq \block_n.\past$, we have
    $$ \forall \ti\block\in \gmin_{N(r_0)},\left|\bigcup\nolimits_{\substack{r\in\{r_1,\cdots,\rmax\}\\\state\in \uniongraphs_r}} \mathrm{Prefix}(\order_{\sf GHAST}(\state),\ti\block)\right|=1.$$

\end{proof}

\subsection{Proof of Theorem~\ref{thm:potcase:final}}

\begin{proof}
    Combining the conclusions in lemma~\ref{lma:potcase:final} and lemma~\ref{lma:cpotcase:final}, we proved this lemma.
\end{proof}

\subsection{Proof of Theorem~\ref{thm:cpotdiff}}

\begin{proof}
    Recalling that $\cpot(\advs,\graph)$ picks the maximum block potential value of blocks in $\vcmp':=\{\block'\in\vcmp|\graph\nsubseteq \block'.\past\}$. Since genesis block $\genesisblock$ must be in $\vcmp'$ and $\mathrm{Old}(\gmin\l,\genesisblock)=\true$, there must exists block in $\vcmp'$ whose block potential value is not $\bot$. Let $\block$ be the block with maximum block potential value in $\vcmp'$. (A.k.a. $\block:=\argmax_{\block'\in\vcmp'}\pot(\advs,\block')$). Thus 
    $$ \cpot(\advs,\graph)=\pot(\advs,\block).$$
    
    Since $\pot(\advs,\block)\neq\bot$, we have $\mathrm{Old}(\gmin,\block)=\true$. In the assumptions of this theorem, $\mathrm{Old}(\gmin,\block')=\false \vee \mathrm{Old}(\gmin\l,\block')=\true \vee  \cpot(\advs,\graph)=\pot(\advs,\block')$ holds for all the block $\block'\in \vcmp'$. So we have 
    $$ \mathrm{Old}(\gmin\l,\block)=\true.$$
    
    Since $\mathrm{Old}(\gmin\l,\block)=\true$, it must be $\block\in\gmin\l$. 
    Thus $\block.\past\subseteq\gmin\l\subseteq\gmax\l$. 
    Since we assume $\gmax\l\cap\graph\l=\gmax\l\cap\graph$, thus $\block.\past\cap\graph=\block.\past\cap\graph\l$. If $\graph\l\subseteq \block.\past$, then it will be $\graph\subseteq \block.\past$ and thus $\block\notin\vcmp'$. This contradicts $\block\in\vcmp'$. Thus we have
    $$\graph\l\nsubseteq \block.\past.$$

    \textbf{Case 1:} $\pot(\advs\l,\block)\neq\bot$.

    Since $\pot(\advs\l,\block)\neq\bot$ and $\block\in\vcmp\l$, we have $\block\in \vcmp'\l$. Recalling that $\cpot(\advs\l,\graph\l)$ picks the maximum block potential value of blocks in $\vcmp'\l:=\{\block'\in\vcmp\l|\block'.\past\nsubseteq\graph\l\}$, there must be 
    $ \pot(\advs\l,\block)\le \cpot(\advs\l,\graph\l).$
    According to lemma~\ref{lma:potcase:final}, $(\pot(\advs,\block)+\spevalue)-(\pot(\advs\l,\block)+\spevalue\l)\le \eventweight(\advs\l,\event)$. Thus 
    $$ (\cpot(\advs,\graph)+\spevalue)-(\cpot(\advs\l,\graph\l)+\spevalue\l)\le (\pot(\advs,\block)+\spevalue)-(\pot(\advs\l,\block)+\spevalue\l)\le \eventweight(\advs\l,\event).$$

    \textbf{Case 2:} $\pot(\advs\l,\block)=\bot$.

    Since $\pot(\advs\l,\block)=\bot$ and $\mathrm{Old}(\gmin\l,\block)=\true$, we have $\block\notin\vcmp\l$. According to lemma~\ref{lma:tipmove1}, $\vcmp$ must be a prefix of $\vcmp\l$. Thus $\tip(\vcmp\l)\prec\block$. So we have $\tip(\vcmp\l).\past\subseteq\block.\past.$ Thus 
    $$\graph\l\nsubseteq \tip(\vcmp\l).\past.$$
    
    Recalling that $\cpot(\advs\l,\graph\l)$ picks the maximum block potential value of blocks in $\vcmp'\l:=\{\block'\in\vcmp\l|\block'.\past\nsubseteq\graph\l\}$, there must be 
    $ \pot(\advs\l,\tip(\vcmp\l))\le \cpot(\advs\l,\graph\l).$
    According to lemma~\ref{lma:cpotcase:final}, $(\pot(\advs,\block)+\spevalue)-(\pot(\advs\l,\tip(\vcmp\l))+\spevalue\l)\le \eventweight(\advs\l,\event)$. Thus 
    $$ (\cpot(\advs,\graph)+\spevalue)-(\cpot(\advs\l,\graph\l)+\spevalue\l)\le (\pot(\advs,\block)+\spevalue)-(\pot(\advs\l,\tip(\vcmp\l))+\spevalue\l)\le \eventweight(\advs\l,\event).$$

\end{proof}

\section{The Summation of Event Values}

\subsection{Decompose Event Values}

The definition of potential value is not friendly for analysis probability distribution, we elaborate the decomposition of the event value into several components:
$\eventweight_{\sf H}(\advs\l,\event)$, $\eventweight_{\sf M}(\advs\l,\event)$, $\eventweight_{\sf F}(\advs\l,\event)$ and $\eventweight_{\sf T}(\advs\l,\event)$. And shows that the event value is always no more than the sum of components in lemma~\ref{lma:event components}. 
\begin{itemize}
    \item When $\event$ is a $\mgen$ event, $\eventweight_{\sf M}(\advs\l,\event):=\event.{\sf block}.\weight$. For other cases, $\eventweight_{\sf M}(\advs\l,\event):=0$.
    \item When $\event$ is a $\hgen$ event and $\mathrm{Spe}(\advs\l)=\false$, $\eventweight_{\sf H}(\advs\l,\event):=-(\heavyw-2\kah-2\kam)/\heavyw\cdot\event.{\sf block}.\weight$. 
    For other cases. It equals to 0 for other cases. 
    \item $\eventweight_{\sf F}(\advs\l,\event):=2\heavyw-2\kah-\kam$ if $\event$ is an $\hgen$ event of block with block weight $\heavyw$ and $\gdta\l$ has an honest block with block weight $\heavyw$. It equals to 0 for other cases. 
    \item $\eventweight_{\sf T}(\advs\l,\event):=-\kam$ if $\event$ is an $\hgen$ event of block with block weight $\heavyw$ and $\gdta\l$ has at least two honest blocks with block weight $\heavyw$. It equals to 0 for other cases. 
\end{itemize}

\begin{lemma}\label{lma:event components}
    \statecondm. We have 
    $$ \eventweight(\advs\l,\event)\le \eventweight_{\sf M}(\advs\l,\event)+\eventweight_{\sf H}(\advs\l,\event)+\eventweight_{\sf F}(\advs\l,\event)+\eventweight_{\sf T}(\advs\l,\event).$$
\end{lemma}
\begin{proof}
    For the case $\event.\type\in\{\mrls,\allrec\}$, the event value and all its components always be 0. So the inequality holds trivially.
    
    If $\event.\type=\mgen$, the event value components $\eventweight_{\sf H},\eventweight_{\sf F},\eventweight_{\sf T}$ are all 0. So we have $\eventweight(\advs\l,\event)=\event.{\sf block}.\weight=\eventweight_{\sf M}(\advs\l,\event).$ 

    If $\event.\type=\hgen$ and $\event.{\sf block}.\weight=0$,  the event value and all its components must be 0. 

    If $\event.\type=\hgen$ and $\event.{\sf block}.\weight=1$, the event value components $\eventweight_{\sf M},\eventweight_{\sf F},\eventweight_{\sf T}$ are all 0. When $\mathrm{Spe}(\advs\l,\event)=\false$, we have $\eventweight(\advs\l,\event)=-1$ and $\eventweight_{\sf H}(\advs\l,\event)=-1+(2\kah+2\kam)/\heavyw\ge -1$. When $\mathrm{Spe}(\advs\l,\event)=\true$, we have $\eventweight(\advs\l,\event)=0$ and $\eventweight_{\sf H}(\advs\l,\event)=0$. So the inequality holds for this case.

    If $\event.\type=\hgen$ and $\event.{\sf block}.\weight=\heavyw$, there must be $\eventweight_{\sf M}(\advs\l,\event)=0$. Table~\ref{tab:event components} lists the value of $\eventweight_{\sf H}(\advs\l,\event)$, $\eventweight_{\sf F}(\advs\l,\event)$, $\eventweight_{\sf T}(\advs\l,\event)$ and $\eventweight(\advs\l,\event)$ under all the possible cases. We can check that $\eventweight_{\sf H}(\advs\l,\event)+\eventweight_{\sf F}(\advs\l,\event)+\eventweight_{\sf T}(\advs\l,\event)\ge \eventweight(\advs\l,\event)$ holds for all the cases.

    \begin{table}[htbp]
        \begin{center}
            \begin{tabular}{|l|l|l|l|l||l|}
                \multicolumn{6}{l}{Let $w_1:=2\kam+2\kah-\heavyw$ and $w_2:=2\heavyw-2\kah-\kam$.}\\
                \hline
                    &       & $\eventweight_{\sf H}$     & $\eventweight_{\sf F}$     & $\eventweight_{\sf T}$     & $\eventweight$ \\
                \hline
                \multirow{2}[4]{*}{$|\mathbf{H}\l|=0^{\; 1}$} & $\mathrm{Spe}(\advs\l)=\false$ & $w_1$ & 0     & 0     & $w_1$ \\
            \cline{2-6}          & $\mathrm{Spe}(\advs\l)=\true$  & 0     & 0     & 0     & 0 \\
                \hline
                \multirow{2}[4]{*}{$|\mathbf{H}\l|=1^{\; 1}$} & $\flagb=\bot$   & $\ge w_1$ & $w_2$ & 0     & 0 \\
            \cline{2-6}          & $\flagb\neq\bot$  & $\ge w_1$ & $w_2$ & 0     & $\heavyw+\kam$ \\
                \hline
                \multirow{2}[4]{*}{$|\mathbf{H}\l|\ge 2^{\; 1}$} & $\flagb=\bot$   & $\ge w_1$ & $w_2$ & $-\kam$    & 0 \\
            \cline{2-6}          & $\flagb\neq\bot$  & \multicolumn{4}{c|}{Impossible$^{\;2}$} \\
                \hline 
                \multicolumn{6}{l}{\footnotesize
                    1. Let $\mathbf{H}\l$ includes all the honest blocks in $\gdta\l$ with block weight $\heavyw$. }\\
                \multicolumn{6}{l}{\footnotesize
                    2. If $\flagb\neq\bot$, there must be $|\mathbf{H}\l|=1$ according to our rule in updating flag block. }\\
            \end{tabular}%
            \caption{Event value and its components under different cases.}\label{tab:event components}
        \end{center}
    \end{table}%
\end{proof}

\subsection{Probability for each Component}\label{sec:prob:component}

Recalling that $\vec{\eta}:=(\difficulty,\heavyw,\advan,\timerw,\timerdiff)$ and $(m,\beta,\delay,\adv,\env)$ admissible w.r.t. $(\proto_{\sf GHAST}^{\vec{\eta}},\order_{\sf GHAST})$ in our analysis. 
\\
$\view^{(\proto_{\sf GHAST},\order_{\sf GHAST})}(\env,\adv,\secp)$ is the random variable denote the joint view of all the participant nodes and the adversary in all rounds. We denote it as $\view$ in section~\ref{sec:prob:component}. Now we define several random variables determined by $\view$. Let $\view_r$ denote the joint view before round $r$. 

Let $\event_n$ denote the $n^{th}$ event since the ghast protocol launched and $\advs_{n-1}$ denote the adversary state when event $\event_n$ happens. Similarly, symbols $\ggen_n$, $\gmax_n$, $\gmin_n$, $\gdta_n$, $\mathbf{M}_n$, $\flagb_n$, $\vcmp_n$, $\speset_n$, $v_n$ denote corresponding components of $\advs_n$ in the context for any subscript.
For any round $r$, let $N(r)$ be the index of last event before round $r$. We define random variables $ \varm_r$,  $\varh_r$, $\varf_r$ as follows:
\begin{equation}\label{eq:event rv}
    \begin{aligned}
        \varm_r  &:=\sum_{i=N(r)+1}^{N(r+1)} \eventweight_{\sf M}(\advs_{i-1},\event_i) &
        \varh_r  &:=\sum_{i=N(r)+1}^{N(r+1)} \eventweight_{\sf H}(\advs_{i-1},\event_i) &
        \varf_r  &:=\sum_{i=N(r)+1}^{N(r+1)} \eventweight_{\sf F}(\advs_{i-1},\event_i) 
    \end{aligned}
\end{equation}

For the case no event happens in round $r$, it will be $N(r)=N(r+1)$, so all the three random variables equal to 0 in this case.

We use valued random variable $\vars_r$ to denote if there exists an adversary state $\advs$ with $\mathrm{Spe}(\advs)=\true$ in phase 3 of round $r$. Futhermore, we use boolean-valued random variables $\varsm_r,\varsh_r,\varsf_r$ to distinguish the reason in triggering special status. $\varsm_r$ denotes there exists adversary state $\advs$ in phase 3 or round $r$ which satisfies the first rule in the definition of special status (definition~\ref{def:special}). $\varsh_r$ and $\varsf_r$ correspond to the second rule and the third rule. The random variables equal to 1 for the ``$\true$'' statement and equal to 0 otherwise. 

\begin{claim}\label{clm:probc:vars}
    For any given round $r$, $\vars_r\le \varsh_r+\varsf_r+\varsm_r$. 
\end{claim}

\begin{lemma}\label{lma:probc:varsm}
    For any given round $r_1<r_2$, we have 
    $$ (v_{N(r_2)}-v_{N(r_1)})/\kam\ge \sum_{i=r_1}^{r_2-1} \varsm_i/(\delay+1)-1.$$
\end{lemma}
\begin{proof}
    For any round $r$ with $\varsm_r=1$, we have adversary state $\advs_n$ with $N(r)<n\le N(r+1)$ such that $$\tw(\tree{\gdta_n,\tip(\vcmp_n)}\cap \mathbf{M}_n)\ge \kam.$$
    Let $\mathbf{T}_n:=\tree{\gdta_n,\tip(\vcmp_n)}$. According to claim~\ref{clm:spesv}, we have $\mathbf{T}_n\cap \mathbf{M}_n\subseteq\speset_n$. Since $\gdta_n$ only contains blocks generated no earlier than round $r-\delay+1$. Let $n'=N(r-\delay)$, we have $\gdta_n\cap \ggen_{n'}=\emptyset$. Thus $\left(\mathbf{T}_n\cap \mathbf{M}_n\right)\cap\speset_{n'}=\emptyset$. So we have 
    $$ \tw(\speset_{n}\backslash\speset_{n'})\ge \kam.$$

    According to claim~\ref{clm:spesv}, $\tw(\speset_{i+1}\backslash\speset_{i})\le \min\{\kam,v_{i+1}-v_i\}$ and $\speset_{i}\subseteq\speset_{i+1}$ and $\spevalue_i-\spevalue_{i-1}\ge 0$ hold for all the $i$. So we have
    \begin{align*}
        \tw(\speset_{n}\backslash\speset_{n'}) \le &\sum_{i=N(r-\delay)+1}^{N(r+1)} \tw(\speset_i\backslash\speset_{i-1}) \\ 
        \le &\sum_{i=N(r-\delay)+1}^{N(r+1)} \min\{\kam,\spevalue_i-\spevalue_{i-1}\} \\
        \le &\min\{\kam,\spevalue_{N(r+1)}-\spevalue_{N(r-\delay)}\} 
    \end{align*}

    Thus we claim $\spevalue_{N(r+1)}-\spevalue_{N(r-\delay)}\ge \kam$ holds if $\varsm_r=1$. Recalling that $\spevalue_i$ is non-decreasing in terms of $i$, we have 
    \begin{align*}
        \sum_{i=r_1+\delay}^{r_2-1} \varsm_i\le &\left(\sum_{i=r_1+\delay}^{r_2-1} \spevalue_{N(i+1)}-\spevalue_{N(i-\delay)}\right)/\kam \\ 
        \le &(\delay+1)/\kam\cdot\left(\spevalue_{N(r_2)}-\spevalue_{N(r_1)}\right) \\
    \end{align*}

    The rest part $\sum_{i=r_1}^{r_1+\delay-1} \varsm_i\le \delay$ holds trivially. Thus 
    $$ (v_{N(r_2)}-v_{N(r_1)})/\kam\ge \sum_{i=r_1}^{r_2-1} \varsm_i/(\delay+1)-1.$$
\end{proof}

\begin{lemma}\label{lma:probc:varf}
    For any given round $r_1<r_2$, we have 
    $$ \sum_{i=N(r_1)+1}^{N(r_2)}-\eventweight_{\sf T}(\advs_{i-1},\event_i)/\kam \ge \sum_{j=r_1}^{r_2-1} \varsf_j/(\delay+1)-1.$$
\end{lemma}
\begin{proof}
    For any round $r$ with $\varsf_r=1$, we have adversary state $\advs_n$ with $N(r)<n\le N(r+1)$ such that $\gdta_n$ contain at least three honest blocks with block weight $\heavyw$. Suppose event $\event_{n'}$ be the latest $\hgen$ event of these three blocks. So $\gdta_{n'-1}$ contain at least two honest blocks with block weight $\heavyw$. Thus $\eventweight_{\sf T}(\gdta_{n'-1},\event_{n'})=-\kam$. Since all the blocks in $\gdta_n$ should be generated no earlier than round $r-\delay+1$, we have $n'>N(r-\delay)$. Recalling that $\eventweight_{\sf T}(\advs_{i-1},\event_i)\le 0$ for all $i$, we have 
    $$\varsf_r\le \sum_{i=N(r-\delay)+1}^{N(r+1)}-\eventweight_{\sf T}(\advs_{i-1},\event_i)/\kam.$$
    Thus we have 
    \begin{align*}
        \sum_{i=r_1}^{r_2-1} \varsf_i= & \sum_{i=r_1}^{r_1+\delay-1} \varsf_i + \sum_{i=r_1+\delay}^{r_2-1} \varsf_i \\
        \le &\delay + \sum_{i=r_1+\delay}^{r_2-1}\sum_{j=N(i-\delay)+1}^{N(i+1)} -\eventweight_{\sf T}(\advs_{i-1},\event_i)/\kam\\ 
        \le &\delay + (\delay+1)\cdot \sum_{i=N(r_1)+1}^{N(r_2)} -\eventweight_{\sf T}(\advs_{i-1},\event_i)/\kam.
    \end{align*}
\end{proof}

\begin{lemma}\label{lma:probc:1}
    For any round $r_1<r_2$, 
    let $u:=(\heavyw-2\kam-2\kah)/\heavyw$, 
    $p_1(t)=\exp\left(\frac{(e^{t\heavyw}-1)\cdot \beta m}{\heavyw\difficulty}\right)$,  
    $p_2(t)=\exp\left(\frac{(e^{-tu\heavyw}-1)\cdot (1-\beta) m}{\heavyw\difficulty}\right)$, 
    $p_3(t)=\exp(-t\kam/(\delay+1))$ 
    and $p(t):=p_1(t)\cdot\max\{p_2(t),p_3(t)\}$. 
    For any round $r$, let $X_r:=\varm_r+\varh_r-\kam/(\delay+1)\cdot\vars_r$.
    For any $t>0$, any $\view_r$ and any $k\in\mathbb{R}$, we have
    $$ \Pr\left[\sum_{i=r_1}^{r_2-1}X_i\ge k\middle|\view_{r_1}\right]\le p(t)^{r_2-r_1}/e^{tk}.$$
    and 
    $$\Pr\left[\sum_{i=r_1}^{r_2-1}M_i\ge k \middle|\view_{r_1}\right]\le p_1(t)^{r_2-r_1}/e^{tk}.$$
\end{lemma}
\begin{proof}
    Let $\view'_r$ denote the joint view in $\view$ before the phase 3 of round $r$. For any $r\ge r_1$ and $t>0$, we have the following discussion. 

    In phase 3 of round $r$, suppose the adversary queries the oracle $\mine(\cdot)$ $x_1$ times with the block $\block_{\sf new}$ satisfying $\adp(\block_{\sf new}.\past)=\false$ and queries the oracle $x_2$ times the block $\block_{\sf new}$ satisfying $\adp(\block_{\sf new}.\past)=\true$. (Note that $\block_{\sf new}$ is not a valid block. But its past set $\block_{\sf new}.\past$ is fixed before querying random oracle.) Since we only allow adversary control $\beta m$ malicious nodes, $x_1+x_2\le \beta m$.
    
    For any given $x_1$, $x_2$ and $\view'$, let $A_i$ ($i\in[x_1]$) and $B_j$ ($j\in[x_2]$) denote the block weight if the adversary finds a solution to make the block valid and equals to 0 otherwise. So $$\varm=\sum_{i\in[x_1]}A_i+\sum_{j\in[x_2]}B_j.$$ 
    
    $B_j$ corresponds to the queries satisfying $\adp(\block_{\sf new}.\past)=\true$. So $B_j=\heavyw$ with probability $1/(\difficulty\heavyw)$ and $B_i=0$ otherwise. Thus $E\left[e^{tB_j}|\view'_r,x_1,x_2\right]=(\difficulty\heavyw-1+e^{-t\heavyw})/(\difficulty\heavyw).$ 

    $A_i$ corresponds to the queries satisfying $\adp(\block_{\sf new}.\past)=\false$. So $A_i=1$ with probability $1/\difficulty$ and $A_i=0$ otherwise. Thus $E\left[e^{tA_i}|\view'_r,x_1,x_2\right]=(\difficulty-1+e^{-t})/\difficulty.$ 
    According to AM-GM inequality, $e^{-t\heavyw}/\heavyw+(\heavyw-1)/\heavyw\ge e^{-t}$. Thus $(\difficulty-1+e^{-t})/\difficulty\le (\difficulty\heavyw-1+e^{-t\heavyw})/(\difficulty\heavyw)$. 
    
    Since all the $A_i$ and $B_j$ are independent conditioned on $\view'_r$, we have 

    $$ E\left[e^{t\varm_r}\middle|\view'_r\right]\le\left((\difficulty\heavyw-1+e^{t\heavyw})/(\difficulty\heavyw)\right)^{x_1+x_2}.$$

    Notice that $(\difficulty\heavyw-1+e^{t\heavyw})/(\difficulty\heavyw)\le \exp((e^{t\heavyw}-1)/(\difficulty\heavyw))$, recalling $x_1+x_2\le \beta m$, we have 
    $$ E\left[e^{t\varm_r}\middle|\view'_r\right]\le\exp((e^{t\heavyw}-1)\beta m/(\difficulty\heavyw))=p_1(t).$$

    $M_r$ is the sum of block weight for malicious blocks generated in round $r$. If $\vars_r=0$, then $H_r$ is $-u$ times the sum of block weight for honest blocks generated in round $r$. And the sum of block weight is independent with $\vars_r$, Similar with the previous claim, for all the $\view'_r$ and $t>0$,
    $$ E\left[e^{t\varh_r}\middle|\view'_r,\vars_r=0\right]\le\exp((e^{-tu\heavyw}-1)(1-\beta) m/(\difficulty\heavyw))=p_2(t).$$

    If $\vars_r=1$, we have $\varh_r\le 0$. Thus
    $$ E\left[e^{t(\varh_r-\kam/(\delay+1)\cdot \vars_r)}\middle|\view'_r\right]\le\max\{p_2(t),p_3(t)\}.$$

    Since $\varm_r,\varh_r,\vars_r$ are independent under given $\view'_r$, thus for any given $r\ge r_1$, $\view'_r$ and $t>0$
    $$ E\left[e^{tX_r}\middle|\view'_r\right]\le p_1(t)\cdot\max\{p_2(t),p_3(t)\}=p(t).$$

    When $r\ge r_1$, $\view_{r_1}$ and $X_i$ for $i\in[r_1,r)$ are determined only by $\view'_r$. So for any $\view_{r_1}$, $r\ge r_1$, $X_i$ for $i\in[r_1,r)$ and  $t>0$, we have
    $$ E\left[e^{tX_r}\middle|e^{t\sum_{i=r_1}^{r-1}X_i},\view_{r_1}\right]\le p(t).$$
    Thus
    $$ E\left[e^{t\sum_{i=r_1}^{r}X_i}\middle|\view_{r_1}\right]\le p(t)\cdot E\left[e^{t\sum_{i=r_1}^{r-1}X_i}\middle|\view_{r_1}\right].$$
    By induction, we have
    $$ E\left[e^{t\sum_{i=r_1}^{r_2}X_i}\middle|\view_{r_1}\right]\le p(t)^{r_2-r_1}.$$
    According to the Markov’s inequality, for any $t>0$, any $\view_r$ and any $k\in\mathbb{R}$, we have
    $$ \Pr\left[\sum_{i=r_1}^{r_2-1}X_i\ge k\middle|\view_{r_1}\right]\le p(t)^{r_2-r_1}/e^{tk}.$$
    Similarly, since we have $E\left[e^{t\varm_r}\middle|\view'_r\right]\le p_1(t)$, thus
    $$ \Pr\left[\sum_{i=r_1}^{r_2-1}\varm_i\ge k\middle|\view_{r_1}\right]\le p_1(t)^{r_2-r_1}/e^{tk}.$$
\end{proof}

\begin{lemma}\label{lma:probc:2}
    For any given round $r_1<r_2$, let $q:=(1-\beta)m\delay/(\heavyw\difficulty)$, $\mathrm{B}(n,p)$ denote the binomial distribution with experiment times $n$ and success probability $p$, $Y_1$ follows the probability distribution $\mathrm{B}((1-\beta)m (r_2- r_1),1/(\heavyw\difficulty))$. 
    For any $\view_{r_1}$, $k_2\in \mathbb{N}$ and $y\in \mathbb{N}$, we have 
    $$ \Pr\left[\sum_{i=r_1}^{r_2-1}\varf_i\ge k_2\cdot(2\heavyw-2\kah-\kam)\middle|\view_{r_1}\right]\le \Pr[Y_1\ge y+1] + \Pr_{Y_2\sim \mathrm{B}(y,q)}[Y_2\ge k_2].$$
\end{lemma}
\begin{proof}
    We refer the blocks satisfy $\prf^{\sf weight}(\block.\digest)\ge 2^\secp/\heavyw$ as \emph{tagged blocks} in this proof. For any given $y\in \mathbb{N}$, let $n_j$ ($j\in[y]$) denote the index of the $j^{th}$ $\hgen$ event of honest tagged block $\block$ no earlier than round $r_1$. Let $r'_j$ denote the round index of event $n_j$. 
    
    Recalling that only a tagged block can have block weight $\heavyw$. 
    If $r'_{y+1}\ge r_2$, we claim for any $r_1\le i < r_2$, $F_i=2\heavyw-2\kah-\kam$ only if there exists $n_j$ such that $r'_j-r'_{j-1}<\delay$. (We set $r'_0=r_1$.) 
    So we have $\sum_{i=r_1}^{r_2-1}\varf_i \ge k_2\cdot(2\heavyw-2\kah-\kam)$ only if there are at least $k_2$ different $j\in[y]$ satisfy $r'_j-r'_{j-1}<\delay$. 
    For any $j\in [y]$, $r'_j-r'_{j-1}<\delay$ only if the honest nodes queries oracle $\mine(\cdot)$ at least $(1-\beta)\cdot\delay m$ times between event event $\event_{n_j}$ (included) and event $\event_{n_{j-1}}$ (excluded). It will happens with probability no more than
    $ 1-(1-1/(\heavyw\difficulty))^{(1-\beta)m\delay}\le (1-\beta)m\delay/(\heavyw\difficulty)=q$ and 
    independent among different $j$. 
    So we have 
    $$\Pr\left[r'_{y+1}\ge r_2\wedge \sum_{i=r_1}^{r_2-1}\varf_i \ge k_2\cdot(2\heavyw-2\kah-\kam)\middle|\view_{r_1}\right]
    \le \Pr_{Y_2\sim \mathrm{B}(y,q)}[Y_2\ge k_2]
    .$$

    Notice that $r'_{y+1}\ge r_2$ represents the honest nodes generate at most $y$ tagged blocks between round $r_1$ (included) and round $r_2-1$ (included). The honest blocks query oracle $\mine(\cdot)$ $(1-\beta)\cdot m(r_2-r_1)$ times, and find a valid block satisfies this property with probability $1/(\heavyw\difficulty)$ in each query. Thus $r'_{y+1}< r_2$ holds with probability
    $$\Pr_{Y_1\sim \mathrm{B}((1-\beta)m (r_2- r_1),1/(\heavyw\difficulty))}[Y_1\ge y+1].$$
    Applying the union bound, we have 
    $$ \Pr\left[\sum_{i=r_1}^{r_2-1}\varf_i\ge k_2\cdot(2\heavyw-2\kah-\kam)\middle|\view_{r_1}\right]\le \Pr_{Y_1}[Y_1\ge y+1] + \Pr_{Y_2\sim \mathrm{B}(y,q)}[Y_2\ge k_2].$$
\end{proof}

\begin{lemma}\label{lma:probc:3}
    For any given round $r_1<r_2$, let $\mathrm{B}(n,p)$ denote the binomial distribution with experiment times $n$ and success probability $p$. For any $y\in\mathbb{Z}^+$, let $d_y:=\lceil\delay/y\rceil$, $Z_1\sim \mathrm{B}((1-\beta)m(y+1)\delay_y,1/\difficulty)$, $q:=\Pr[Z_1\ge \kah]$, $Z_2\sim \mathrm{B}(\lceil(r_2-r_1)/(\delay_y(y+1))\rceil,q)$. For any $k_3\in \mathbb{N}$ and $\view_{r_1}$, we have
    $$ \Pr\left[\sum_{i=r_1}^{r_2-1}\varsh_i\ge (y+1)\cdot(k_3+1)\cdot\delay_y\middle|\view_{r_1}\right]\le (y+1)\cdot \Pr[Z_2\ge k_3].$$
\end{lemma}
\begin{proof}
    We divide the rounds in $[r_1,r_2)$ into several intervals with length $\delay_y$ except the last one. The $j^{th}$ interval is $[r_1+(j-1)\cdot\delay_y,r_1+j\cdot\delay_y)$. Let $T_j=1$ denote if there exists a round $i$ with $\varsh_i=1$ in the $j^{th}$ interval and $T_j=0$ for other cases. Thus 
    $$\sum_{i=r_1}^{r_2-1}\varsh_i \le \sum_{j=1}^{\lceil(r_2-r_1)/\delay_y\rceil} \delay_y\cdot T_j.$$

    Since $\varsh_i=1$ only if honest nodes generates at least $\kah$ blocks in round $(i-\delay,i]$. So we claim $T_j=1$ ($j\ge y+1$) only if honest nodes generates at least $\kah$ blocks in round $[r_1+(j-y-1)\cdot\delay_y,r_1+j\cdot\delay_y)$. Since the honest nodes will query oracle $\mine(\cdot)$ in $(1-\beta)m(y+1)\cdot\delay_y$ times during round $[r_1+(j-y-1)\cdot\delay_y,r_1+j\cdot\delay_y)$, they will generate at least $\kah$ honest blocks with probability 
    $$q=\Pr_{Z_1\sim \mathrm{B}((1-\beta)m(y+1)\cdot\delay_y,1/\difficulty)}[Z_1\ge \kah].$$
    Notice that $T_{j_2}$ and $T_{j_1}$ are independent event if $j_2-j_1\ge y+1$. So for any $j'\in\{0,1,\cdots,y\}$, we construct a list of random variables for all the $T_j$ with $j\in[y+1,r_2-r_1)$ and $j$ modulo $y+1$ equals to $j'$. The list contains at most $\lceil(r_2-r_1)/(\delay_y(y+1))\rceil$ independent random variables. So the sum of these random variables is lager or equal $k$ with probability $\Pr[Z_2\ge k_3+1]$. \todo{Rewrite it!} Taking the union bound, we have 

    $$ \Pr\left[\sum_{j=y+1}^{\lceil(r_2-r_1)/\delay_y\rceil} T_j\ge (y+1)k_3 \middle|\view_{r_1}\right]\le (y+1)\Pr[Z_2\ge k_3]$$
    Since $\sum_{j=1}^{y} T_j< y+1$ holds trivially, we have 
    $$ \Pr\left[\sum_{i=r_1}^{r_2-1}\varsh_i\ge (y+1)\cdot(k_3+1)\cdot\delay_y\middle|\view_{r_1}\right]\le (y+1)\cdot \Pr[Z_2\ge k_3].$$
\end{proof}

\subsection{Summarize all the Components}

Now we summarize the previous lemmas. First we define a function $\probd(\lambda,\beta,\vec{\eta},\rd,\rho,\kam,\kah,\vec{t})$. $\vec{t}$ is a tuple of four parameters in $\mathbb{R}^+$, which are denoted by $t_1,t_2,t_3,t_4$. $\lambda$ equals to $m(\delay+1)/\difficulty$ and $\rd\in \mathbb{R}^+$. 

\begin{equation}\label{eq:probd}
    \begin{aligned}
        f(\mu,n) & :=\frac{e^{n-\mu}}{(n/\mu)^n}  \\
        x_1 & :=e^{t_1\heavyw}-1 \\
        x_2 & :=e^{t_1(2\kah+2\kam-\heavyw)}-1 \\
        k & :=\rho-(t_3+1)\cdot(2\heavyw-2\kah-\kam)-1.02\cdot\kam\cdot(t_4+2)-2\kam  \\
        q_1 & :=\exp( \rd \cdot \lambda\left(x_1\beta+x_2(1-\beta)\right)/\heavyw- t_1 k) \\
        q_2 & :=\exp(\rd  \cdot (\lambda x_1\beta/\heavyw - t_1 s_m)- t_1 k) \\
        q_3 & :=f(\rd\cdot (1-\beta)\lambda/\heavyw,t_2)+f(t_2\cdot (1-\beta)\lambda/\heavyw,t_3) \\
        x_3 & :=f(1.02\cdot \lambda(1-\beta),\kah) \\
        q_4 & :=101\cdot f((\rd+1)\cdot x_3,t_4) \\
        \probd(\lambda,\beta,\vec{\eta},\rd,\rho,\kam,\kah,\vec{t}) &:=\max\{q_1,q_2\}+q_3+q_4
    \end{aligned}
\end{equation}

\begin{lemma}\label{lma:eventdec}
    For any given round $r_1<r_2$, let $\rd:=(r_2-r_1)/(\delay+1)$, $\lambda = m(\delay+1)/\difficulty$. When $\delay>10^4$, for arbitrary $\kam\ge 0$, $\kah\ge 0$ and positive parameters in $\vec{t}$, we have 
    $$ \Pr\left[\sum_{i=N(r_1)+1}^{N(r_2)}\eventweight(\advs_{i-1},\event_i)-\left(\spevalue_{N(r_2)}-\spevalue_{N(r_1)}\right)\ge \rho\middle| \view_{r_1}\right]\le \probd(\lambda,\beta,\vec{\eta},\rd,\rho,\kam,\kah,\vec{t}).$$
\end{lemma}
\begin{proof}
    This proof inherits symbols in equation~(\ref{eq:probd}). 

    Notice that $m/\difficulty\cdot(r_2-r_1)=\lambda\cdot \rd$, for any round $r$, let $X_r:=\varm_r+\varh_r-\kam/(\delay+1)\cdot\vars_r$. According to lemma~\ref{lma:probc:1}, we have
    \begin{equation}\label{eq:probc:final:1}
        \Pr\left[\sum_{i=r_1}^{r_2-1}X_i\ge k\middle|\view_{r_1}\right]\le \max\{q_1,q_2\}.
    \end{equation}
    If random variable $T$ is in binomial distribution $\mathrm{B}(n,p)$, for any $z>np$, according to chernoff bound, we have 
    $$ \Pr[T\ge z]\le \frac{e^{z/(np)-1}}{(z/(np))^{z/(np)}}=f(np,z).$$

    According to lemma~\ref{lma:probc:2}, (let variable $k_2$ in lemma~\ref{lma:probc:2} equals to $\lceil t_3\rceil$ and variable $y$ in it equals to $\lfloor t_2 \rfloor$), we have

    \begin{equation} \label{eq:probc:final:2}
        \Pr\left[\sum_{i=r_1}^{r_2-1}\varf_i\ge (t_3+1)\cdot(2\heavyw-2\kah-\kam)\middle|\view_{r_1}\right]\le q_3.
    \end{equation}

    Let variable $y$ in lemma~\ref{lma:probc:3} equals to 100 and variable $k_3$ in it equals to $\lceil t_4\rceil$. Since $\delay>10^4$, we have $\lceil\delay/100\rceil\cdot(100+1)\le 1.02(\delay+1)$. According to lemma~\ref{lma:probc:3},
    \begin{equation} \label{eq:probc:final:3}
        \Pr\left[\sum_{i=r_1}^{r_2-1}\varsh_i\ge 1.02\cdot (t_4+2)\cdot(\delay+1)\middle|\view_{r_1}\right]\le q_4
    \end{equation}

    Summarize the previous lemmas, we can link the random variables to the summation of event values.

    \begin{align*}
        & \sum_{i=N(r_1)+1}^{N(r_2)}\eventweight(\advs_{i-1},\event_i)\\
        \le & \sum_{i=N(r_1)+1}^{N(r_2)}\left(
            \eventweight_{\sf M}(\advs_{i-1},\event_i) + 
            \eventweight_{\sf H}(\advs_{i-1},\event_i) + 
            \eventweight_{\sf F}(\advs_{i-1},\event_i) + 
            \eventweight_{\sf T}(\advs_{i-1},\event_i)
            \right) 
        & \text{Lemma~\ref{lma:event components}} \\
        \le & \sum_{i=r_1}^{r_2-1} \left(\varm_i+\varh_i+\varf_i-\frac{\kam}{\delay+1}\cdot \varsf_i\right) +\kam 
        & \text{Eq.~(\ref{eq:event rv})\;\&\; Lemma~\ref{lma:probc:varf}}\\
        \le & \sum_{i=r_1}^{r_2-1} \left(X_i+\varf_i+\frac{\kam}{\delay+1}\cdot \left(\varsm_i+\varsh_i\right)\right) +\kam & \text{Claim~\ref{clm:probc:vars}}\\
        \le & \sum_{i=r_1}^{r_2-1} \left(X_i+\varf_i+\frac{\kam}{\delay+1}\cdot\varsh_i\right) +2\kam + \left(\spevalue_{N(r_2)}-\spevalue_{N(r_1)}\right)& \text{Lemma~\ref{lma:probc:varsm}}
    \end{align*}
    
    Since $\rho=k+(t_3+1)\cdot(2\heavyw-2\kah-\kam)+1.02\cdot(t_4+2)+2\kam$, according to equations~\ref{eq:probc:final:1}, \ref{eq:probc:final:2} and \ref{eq:probc:final:3}, for the parameters $\kam\ge 0$, $\kah\ge 0$ and $t_1,t_2,t_3,t_4$, we claim 
    $$ \Pr\left[\sum_{i=N(r_1)+1}^{N(r_2)}\eventweight(\advs_{i-1},\event_i)-\left(\spevalue_{N(r_2)}-\spevalue_{N(r_1)}\right)\ge \rho\middle| \view_{r_1}\right]\le \max\{q_1+q_2\}+q_3+q_4.$$
\end{proof}

We define function $\probf(\lambda,\beta,\vec{\eta},\rd,\rho)$ as the minimum value that $\probd(\lambda,\beta,\vec{\eta},\rd,\rho,\kam,\kah,\vec{t})$ can achieve by picking parameters $\kah,\kam$ and $\vec{t}$. Then we have the following lemma

\begin{lemma}\label{lma:probf}
    When $\lambda \ge {0.8\log(500/\delta)}$ and $\heavyw=30\lambda/\delta$, for any $\varepsilon>0$, if 
    $$\rd\ge \max\left\{(3+\rho/\heavyw)\cdot\frac{600}{\delta^2},\log\left(\frac{4}{\varepsilon}\right)\cdot\frac{3000}{\delta^3},\log\left(\frac{404}{\varepsilon}\right)\cdot\frac{200}{\delta}\right\},$$
    then 
    $$\probf(\lambda,\beta,\vec{\eta},\rd,\rho)\le \varepsilon.$$
\end{lemma}

\begin{proof}
    In this proof, we try to find a feasible solution for $\probd(\lambda,\beta,\vec{\eta},\rd,-\rho,\kam,\kah,\vec{t})<\varepsilon.$ This proof inherits symbols in equation~(\ref{eq:probd}). 
    Let $\kam=1.5\lambda$, $\kah=3\lambda$ and $t_1=\delta^2/(150\lambda)$. Then we have 
    $x_1=e^{\delta/5}-1$ and $x_2=e^{-\delta/5\cdot (1-7\delta/10)}-1\le (e^{-\delta/5}-1)\cdot(1-7\delta/10)$. Thus 
    \begin{align*}
        & \lambda\left(x_1\beta+x_2(1-\beta)\right)/\heavyw \\
        = & \frac{\delta}{30(2-\delta)}\cdot \left((e^{\delta/5}+e^{-\delta/5}-2)\cdot(1-7\delta/10)-0.7\delta\cdot(e^{\delta/5}-1)\right) \\
        \le & \frac{\delta}{30(2-\delta)}\cdot \left(\frac{\delta^2}{25}-\frac{7\delta^2}{50}\right) \\
        \le  & -\frac{\delta^3}{600}
    \end{align*}
    
    So $$q_1\le \exp(-\rd\cdot\delta^3/600-t_1\cdot k_1).$$ 

    Since it can be verified that $\lambda x_1\beta/\heavyw\le \frac{\delta^2}{300}$ and $-t_1*\kam=-\frac{\delta^2}{100}$, so we have 
    $$q_2\le  \exp(-\rd\cdot\delta^2/150-t_1\cdot k_1).$$

    Let $t_2=\sqrt{2}\cdot\rd\cdot (1-\beta)\lambda/\heavyw$ and $t_3=\sqrt{2}\cdot (1-\beta)\lambda/\heavyw\cdot t_2$, we skim the detailed computation and claim
    $$ q_3< 2\exp(-\rd\cdot(1-\beta)^2\delta^2/450)<2\exp(-\rd\cdot\delta^2/1800).$$

    Notice that $\kah\ge 3\cdot\lambda$. Thus 
    $$ x_3\le f(1.02\cdot\lambda\cdot(1-\beta),3\cdot\lambda)<\exp(-1.25\lambda)=\frac{\delta}{500}.$$

    Let $t_4=4\cdot \delta/500\cdot (\rd+1)$, notice that $\rd\ge \log(404)\cdot 200>1000$, thus $\rd+1<1.001\rd$. So we have 
    $$ q_4<101\cdot \exp(-\rd\cdot \delta/200).$$

    Notice that $t_3=2(1-\beta)^2\lambda^2/\heavyw^2\cdot \rd\le \delta^2/450\cdot \rd$ and $t_4< \delta/120$. We have 
    \begin{align*}
        k_1 & \ge -\rho-2\heavyw\cdot t_3-1.02\kam t_4-3.04\kam-2\heavyw\\ 
        & \ge -\rho-2.26\heavyw -2\heavyw\cdot t_3-1.02\kam t_4 \\ 
        & > -\rho-2.26\heavyw - \rd\cdot\lambda(2\delta/15+\delta/60) \\
        & > -\rho-2.26\heavyw - \rd\cdot\lambda\cdot 3\delta/20
    \end{align*}

    Applying the lower bound of $k_1$ to $q_1$ and $q_2$, we can get
    \begin{align*}
        q_1&< \exp(-\rd\cdot \delta^3/1500+t_1(\rho+3\heavyw))\\
        q_2&< \exp(-\rd\cdot 17\delta^3/3000+t_1(\rho+3\heavyw))
    \end{align*}

    When $\rd\ge (3+\rho/\heavyw)\cdot\frac{600}{\delta^2}$, we have $t_1(\rho+3\heavyw)\le \rd\cdot \delta^3/3000$, thus $q_1< \exp(-\rd\cdot \delta^3/3000)$ and $q_2< \exp(-\rd\cdot 16\delta^3/3000)$. Since $\rd\ge \log\left(\frac{4}{\varepsilon}\right)\cdot\frac{3000}{\delta^3}$, we have $q_1<\varepsilon/4$, $q_2<\varepsilon/4$ and $q_3<\varepsilon/2$. Since $\rd\ge \log\left(\frac{404}{\varepsilon}\right)\cdot\frac{200}{\delta}$, $q_4<\varepsilon/4$. Thus we have 
    $$\probd(\lambda,\beta,\vec{\eta},\rd,-\rho,\kam,\kah,\vec{t})<\varepsilon.$$
\end{proof}

\subsection{Proof of Theorem~\ref{thm:eventdec}}

\begin{proof}
    For any given round $r_1<r_2$, let $\rd:=(r_2-r_1)/(\delay+1)$. According to lemma~\ref{lma:eventdec}, we have 
    $$ \Pr\left[\sum_{i=N(r_1)+1}^{N(r_2)}\eventweight(\advs_{i-1},\event_i)-\left(\spevalue_{N(r_2)}-\spevalue_{N(r_1)}\right)\ge \rho\middle| \view_{r_1}\right]\le \probf(\lambda,\beta,\vec{\eta},\rd,\rho).$$

    According to lemma~\ref{lma:probf}, when $\lambda \ge {0.8\log(500/\delta)}$, $\heavyw=30\lambda/\delta$ and $$\rd\ge \max\left\{(3+\rho/\heavyw)\cdot\frac{600}{\delta^2},\log\left(\frac{4}{\varepsilon}\right)\cdot\frac{3000}{\delta^3},\log\left(\frac{404}{\varepsilon}\right)\cdot\frac{200}{\delta}\right\},$$
    we have 
    $$  \probf(\lambda,\beta,\vec{\eta},\rd,\rho)\le \varepsilon. $$

    So this theorem is proved.
\end{proof}


\section{Timer Chain and Old Enough Blocks}

\begin{lemma}\label{thm:chaindiffgrowth}
    Let $\graph_{n,r}\eqdef\left\{\block\in \ggen_{n}\middle|\gmin_{N(r)}\nsubseteq\block.\past \vee \block\in \gmin_{N(r)}\right\}$. 
    For any round $r$ and $\tmpi\ge 0$, if $\beta\ge 0.1$, $\timerw\ge 2\lambda/\delta$ and $r_\Delta$ satisfying 
    $$r_\Delta\ge\frac{\timerw\difficulty}{m}\cdot 
    \max\left\{
    \frac{128}{\delta^2}\cdot \log(\frac{8400}{\varepsilon\delta^2}),
    \frac{8(\tmpi+2)}{\delta}
    \right\}
    $$
    we have 
    $$\Pr\left[\exists n\ge N(r+r_\Delta), \mathrm{MaxTH}(\gmin_n)-\mathrm{MaxTH}(\graph_{n,r})\le \tmpi\right]\le \varepsilon.$$
\end{lemma}

\begin{proof}
Let $\delta\eqdef 1-\beta/(1-\beta)$, $\tau\eqdef \sqrt[4]{(1-\beta)/\beta}$, $z\eqdef \lfloor\timerw\difficulty/(m\beta\tau)\rfloor$, $c\eqdef (\delay+1)\cdot m/(\timerw\difficulty)$ and $f(x)\eqdef \frac{e^x}{(1+x)^{1+x}}$. Notice that $\tau \le 2$ when $\beta\ge 0.1$. 

Recalling that $N(r)$ denotes the index of latest event before round $r$, $\event_n$ denotes the $n^{th}$ event and $\advs_{n-1}$ denotes the adversary state before event $\event_n$. $\ggen_n,\gmax_n,\gmin_n$ denote the corresponding component in $\advs_n$. 

For any given round $r$ and positive integer $k\in\mathbb{Z}^+$ with $k\ge 10$, we define real number $r_+$ and random variables $X(r,k)$ and $Y(r,k)$ as follows
\begin{align*}
    r_+ &\eqdef r+ z\cdot k \\
    X(r,k) &\eqdef  \mathrm{MaxTH}(\gmin_{N(r_+)})-\mathrm{MaxTH}(\gmax_{N(r)})\\ 
    Y(r,k) &\eqdef \mbox{the number of malicious timer blocks generated in rounds $[r,r_+)$}.
\end{align*}

We try to study probability distribution for $X(r,k)$ and $Y(r,k)$.
In an admissible environment, if an honest node generates or receives a block in round $r$, all the honest nodes will receive such a block before phase 2 of round $r+\delay$. So if an honest node constructs a timer block $\block_1$ in round $r$, for any honest timer block $\block_2$ generated no earlier than round $r+\delay$, it will be $\mathrm{TimerHeight}(\block_2)\ge \mathrm{TimerHeight}(\block_1)+1$. 
We construct an event list as the following steps:
\begin{enumerate}[nosep]
    \item Find the first $\hgen$ event of timer block started with round $r+\delay$. 
    \item Skip the subsequent $(1-\beta)m\delay$ queries for oracle $\mine(\cdot)$ from honest nodes after this event. (Recall that honest nodes query this oracle $(1-\beta)m$ times in each round.) Then find the next $\hgen$ event of timer block. 
    \item Repeat step 2 until reaching the end of round $r_+-\delay-1$. 
\end{enumerate}

We claim the timer height of the first element in this event list will be no less than $\mathrm{MaxTH}(\gmax_{N(r)})+1$. And for any two consecutive events, the timer height of the latter one must be strict larger than the timer height of the former one. So $\mathrm{MaxTH}(\gmin_{N(r_+)})-\mathrm{MaxTH}(\gmax_{N(r)})$ is no less than the length of this event list.

If $X(r,k)\le \tau k$, the length of such event list is no larger than $\tau k$ and at most $(1-\beta)m\delay\cdot \tau k$ queries are skipped when contracting such event list. So it only happens when honest nodes find at most $\tau k$ timer blocks in $(r_+-r-(\tau k+2)\delay+1)(1-\beta)m$ queries. An honest node will find a valid timer block with probability $1/(\difficulty\timerw)$ in each query and the outcomes are independent. 
Notice that 
\begin{align*}
     & (r_+-r-(\tau k+2)\delay+1)(1-\beta)m \\
    \ge & \left(\frac{\timerw\difficulty}{m\beta\tau}-(\tau k+0.2k)(\delay+1)\right)\cdot (1-\beta)m \\
    = & (1-\beta)k\cdot \timerw\difficulty\cdot\left(\frac{1}{\beta\tau}-(\tau+0.2)\cdot c\right) 
\end{align*}

It can be verified that $\frac{1/(\beta\tau)-\tau^2/(1-\beta)}{\tau+0.2}<\delta/2$. Since $c>\delta/2$, 
$$ (1-\beta)k\cdot \timerw\difficulty\cdot\left(\frac{1}{\beta\tau}-(\tau+0.2)\cdot c\right)>\tau^2 k\difficulty\timerw.$$

Let $W$ be a random variable with probability distribution $\mathrm{B}(\tau^2 k\difficulty\timerw,1/(\difficulty\timerw))$. According to the chernoff bound (lemma~\ref{lma:chernoff}), we have 
$$ \forall k'<\tau^2\cdot k, \quad \Pr[X(r,k)\le k']\le \Pr[W\le \tau \cdot k]\le f\left(\frac{k'}{\tau^2\cdot k}-1\right)^{\tau^2k}.$$ 

$Y(r,k)\ge k$ only if the adversary generates at least $k$ timer blocks in round $[r,r_+)$, which implies the adversary finds $k$ timer blocks in $(r_+-r)\beta m\le k/\tau$ queries. Similarly, we can get
$$ \forall k'>k/\tau, \quad \Pr[Y(r,k)\ge k']\le f\left(\frac{\tau\cdot k'}{k}-1\right)^{k/\tau}.$$ 

%

%

Now we will study the probability that 
$$ \exists n\ge N(r+r_\Delta), \mathrm{MaxTH}(\graph_{n,r})-\mathrm{MaxTH}(\gmin_n)\le \tmpi.$$

Let $t$ denote the largest timer height such that no malicious block in $\ggen_{N(r)}$ has timer height $t$. So there must be an honest block with timer height $t$ in $\gmax_{N(r)}$. 
We denote the earliest one by $\block_{t}$. 
Notice that all the honest blocks $\block$ generated no earlier than round $r$ satisfy $\gmin_{N(r)}\subseteq \block.\past$ in an admissible environment and thus they can not appear in  $\graph_{n,r}$ for any $n\ge N(r)$.

Let $n_1$ be the index of $\hgen$ event for $\block_t$ (notice that $n_1<N(r)$). If there exists $n_2\ge N(r+r_\Delta)$ such that 
$\mathrm{MaxTH}(\gmin_{n_2})-\mathrm{MaxTH}(\graph_{n_2,r})\le \tmpi$, we find integers $k_1$ and $k_2$ which satisfy
\begin{align*}
    N(r-z\cdot (k_1-1))\le &n_1< N(r-z\cdot k_1)\\ 
    N(r+z\cdot k_2)\le &n_2 < N(r+z\cdot (k_2+1)) 
\end{align*}

Let $k_{\Delta}=k_{1}+k_{2}$. Since $\tau-1> \delta/4$ and $r_\Delta\ge 4\cdot z(\tmpi+2)/\delta$, we have $k_2 \ge r_\Delta/z = 4\cdot (\tmpi+2)/\delta>(\tmpi+2)/(\tau-1).$
Thus $$\mathrm{MaxTH}(\gmin_{n_2})-\mathrm{MaxTH}(\graph_{n_2,r})\le \tmpi<(\tau-1)\cdot k_2-2 \le(\tau-1)\cdot k_\Delta-2.$$
Then one of the following two inequalities must holds
\begin{align*}
    \mathrm{MaxTH}(\gmin_{n_2})&\le t+\tau\cdot k_\Delta \\
    \mathrm{MaxTH}(\graph_{n_2,r})&\ge t+2+k_\Delta
\end{align*}

For the case $\mathrm{MaxTH}(\gmin_{n_2})\le t+\tau\cdot k_\Delta$, let $r_s=r-z\cdot k_1$, $r_e=r+z\cdot k_2$. Then $\mathrm{MaxTH}(\gmin_{N(r_e)})\le t+\tau\cdot k_\Delta$ because $N(r_e)<n_2$. Since $\block_t$ is the first timer block with height $t$ and it is generated earlier than round $r_s$, we have $\mathrm{MaxTH}(\gmax_{N(r_s)})\ge t$. So it will be 
$$X(r-z\cdot k_1,k_\Delta)=\mathrm{MaxTH}(\gmin_{N(r_e)})-\mathrm{MaxTH}(\gmax_{N(r_s)})\le \tau \cdot k_\Delta.$$

For the case $\mathrm{MaxTH}(\graph_{n_2})\ge t+2+k_\Delta$, let $r_s=r-z\cdot (k_1-1)$, $r_e=r+z\cdot (k_2+1)$.  Since $\block_t$ is the first timer block with height $t$ and it is generated no earlier than round $r_s$ and $\graph_{n_2,r}\subseteq\graph_{N(r_e),r}$, the adversary generates at least $2+k_\Delta$ blocks in rounds $[r_s,r_e)$. It means 
$$Y(r+z\cdot (k_1-1),k_\Delta+2)\ge k_\Delta+2.$$

Notice that $k_1,k_2$ may be dependent with random variable $X(r,k)$ and $Y(r,k)$. So we take a union bound over all the possible $k_1,k_2$. Since $r_\Delta> z\cdot \frac{64}{\delta^2}\cdot \log(\frac{8400}{\varepsilon\delta^2})$, we have $k_2\ge \frac{64}{\delta^2}\cdot \log(\frac{8400}{\varepsilon\delta^2})$. Notice that $1/(1-e^{-x})<1.01/x$ for $0<x<1/64$ and $\tau\le 2$. It can be verified that $f(\tau-1)\le \exp(-\delta^2/32)$ and $f(1/\tau-1)\le \exp(-\delta^2/32)$. Let $y\eqdef\frac{64}{\delta^2}\cdot \log(\frac{8400}{\varepsilon\delta^2})$, we have

\begin{align*}
    & \Pr\left[\exists n\ge N(r+r_\Delta), \mathrm{MaxTH}(\gmin_n)-\mathrm{MaxTH}(\graph_{n,r})\le \tmpi\right]\\
    \le & \sum_{k_1=0}^{\infty}\sum_{k_2=\lceil y\rceil}^{\infty} \left(\Pr[X(r-z\cdot k_1,k_\Delta)\le \tau k_\Delta]+ \Pr[Y(r-z\cdot (k_1-1),k_\Delta+2)\ge k_\Delta+2]\right) \\
    = & \sum_{k_1=0}^{\infty}\sum_{k_2=\lceil y\rceil}^{\infty} \left(f(1/\tau-1)^{\tau^2k_{\Delta}}+f(1-\tau)^{(k_{\Delta}+2)/\tau}\right) \\
    \le & \sum_{k_1=0}^{\infty}\sum_{k_2=\lceil y\rceil}^{\infty} \left(\exp\left(-\delta^2/32\cdot \tau^2k_{\Delta}\right)+\exp(-\delta^2/32\cdot(k_{\Delta}+2)/\tau)\right) \\
    \le & \sum_{k_1=0}^{\infty}\sum_{k_2=\lceil y\rceil}^{\infty} 2\exp(-\delta^2/64\cdot (k_1+k_2)) \\ 
    < & \frac{8400}{\delta^4}\cdot \exp\left(-\delta^2/64\cdot \frac{64}{\delta^2}\cdot \log\left(\frac{8400}{\varepsilon\delta^2}\right)\right)\\
    = & \varepsilon
\end{align*}

Next we will study the probability that 
$$ \mathrm{MaxTH}(\graph_{n,n})-\mathrm{MaxTH}(\gmin_n)\ge \tmpi.$$

Similarly, let $t$ denote the largest timer height such that no malicious block in $\ggen_n$ has timer height $t$. So there must be an honest block with timer height $t$ in $\gmax_{N(r)}$. 
We denote the earliest one by $\block_{t}$. 
Notice that all the honest blocks $\block$ generated no earlier than round $r$ satisfy $\gmin_{N(r)}\subseteq \block.\past$ in an admissible environment. So any honest block generated no earlier than round $r$ can not be in $\graph_{n,r}$ for any $n$.

\end{proof}

\subsection{Proof of Theorem~\ref{thm:cpotwhenold}}

\begin{proof}
    Let $\graph_{n,r}\eqdef\left\{\block\in \ggen_{n}\middle|\gmin_{N(r)}\nsubseteq\block.\past \vee \block \in \gmin_{N(r)}\right\}$,
    $r_\Delta\eqdef\frac{\timerw\difficulty}{m}\cdot \max\left\{
        \frac{129}{\delta^2}\cdot \log(\frac{9000}{\varepsilon\delta^2}),
        \frac{8(\timerdiff+3)}{\delta}
        \right\}$ and $r_1\eqdef r_2-r_\Delta+\delay+1$.

    All the blocks generated earlier than round $r_1-\delay$ belong to $\gmin_{N(r_1)}$. So they also belong to $\graph_{N(r_2),r_1}$. 
    Notice that $\delay+1=\lambda\difficulty/m\le \timerw\difficulty\delta/(2m),$ thus
        $$ r_2-r_1=r_\Delta-\delay-1\ge \frac{\timerw\difficulty}{m}\cdot \max\left\{
            \frac{128}{\delta^2}\cdot \log\left(\frac{9000}{\varepsilon\delta^2}\right),
            \frac{8(\timerdiff+2)}{\delta}
            \right\}.$$
        
    According to lemma~\ref{thm:chaindiffgrowth} , we have 
    $$\Pr\left[\mathrm{MaxTH}(\gmin_{N(r_2)})-\mathrm{MaxTH}(\graph_{N(r_2),r_1})\le \timerdiff \right]\le \frac{14}{15}\cdot \varepsilon.$$

    Notice that $\mathrm{Old}(\gmin_{N(r_2)},\block)=\false$ holds only if $\mathrm{MaxTH}(\gmin_{N(r_2)})-\mathrm{TimerHeight}(\block)< \timerdiff$. As long as $\mathrm{MaxTH}(\gmin_{N(r_2)})-\mathrm{MaxTH}(\graph_{N(r_2),r_1})\le \timerdiff$, all the blocks generated earlier than round $r_1-\delay$ will be old enough given $\gmin_{N(r_2)}$. This holds with exception probability $14/15\cdot \varepsilon$. 
    $\\$

    By the definition of potential value, for any $\advs$ and $\block$, we have $\potwith(\advs,\block) \le \treew{\ggen \backslash \gmax,\block}$, $\potadv(\advs,\block)\le \kah+\kam\le \heavyw$  and $\potspe(\advs,\block)\le \treew{\gdta \cap \mathbf{M},\block}$. 
    Notice that $\ggen \backslash \gmax$ only contains malicious blocks (a.k.a. $\ggen \backslash \gmax \subseteq \mathbf{M}$) and $\gdta = \gmax\backslash\gmin$. We claim $\ggen \backslash \gmax$ and $\gdta \cap \mathbf{M}$ are disjoint sets and there union is subset of or equal to $\ggen \cap \mathbf{M}$. In all, 
    $$ \pot(\advs,\block)\le \heavyw+\treew{\ggen \cap \mathbf{M},\block}.$$

   It implies that $\pot(\advs_{N(r_2)},\block)-\heavyw$ is no more than the total weight of malicious blocks after $\block$'s generation. Let $M_i$ denote the total block weight of malicious blocks generated in round $i$. Then, for any block $\block$ generated no earlier than round $r_1-\delay$, for any $n$ with $N(r_2)<n\le N(r_2+1)$, we have
    $$\pot(\advs_n,\block)-\heavyw\le \sum_{i=r_1-\delay}^{r_2} M_i.$$
    Let $p_1(t):=\exp\left(\frac{(e^{t\heavyw}-1)\cdot \beta m}{\heavyw\difficulty}\right)$, for any $t>0$ and any $k\in\mathbb{R}$, according to lemma~\ref{lma:probc:1}, we have 
    $$ \Pr\left[\sum_{i=r_1-\delay}^{r_2} M_i\ge k\right]\le p_1(t)^{r_\Delta}/e^{tk}.$$
    Let $t\eqdef\delta^2/(150\lambda)$, $k\eqdef 2\lambda/(\delay+1)\cdot r_\Delta$, we have 
    $$  p_1(t)^{r_\Delta}/e^{tk} \le \exp(-\delta^2/150\cdot r_\Delta/(\delay+1)).$$
    Note that $\timerw\difficulty/m=2(\delay+1)$. Thus $r_\Delta/(\delay+1)\ge \frac{300}{\delta^2}\cdot \log(\frac{9000}{\varepsilon\delta^2})> \frac{150}{\delta^2}\cdot \log(\frac{9000}{\varepsilon})$. So we claim 
    $$\exp(-\delta^2/150\cdot r_\Delta/(\delay+1))\le \frac{\varepsilon}{9000}.$$
    It implies with exception probability $\frac{\varepsilon}{9000}$, for any block $\block$ generated no earlier than round $r_1-\delay$, for any $n$ with $N(r_2)<n\le N(r_2+1)$, $\pot(\advs_n,\block)-\heavyw\le 2\lambda/(\delay+1)\cdot r_\Delta$ holds for any block $\block$ generated no earlier than round $r_1-\delay$ and $N(r_2)<n\le N(r_2+1)$. 
    $\\$

    Now we have showed that for any block generated earlier than round $r_1-\delay$, they all become old enough at the beginning of round $r_2$ with exception probability $\frac{14\varepsilon}{15}$. For the other blocks, their block potential value will never exceed $2\lambda/(\delay+1)\cdot r_\Delta+\heavyw$ in round $r_2$ with exception probability $\frac{\varepsilon}{9000}$. Notice that 
    \begin{align*}
          & 2\lambda/(\delay+1)\cdot r_\Delta+\heavyw \\
        = & \frac{30\lambda}{\delta}+4\lambda\cdot \max\left\{\frac{129}{\delta^2}\cdot\log\left(\frac{9000}{\varepsilon\delta^2}\right),\frac{8(\timerdiff+3)}{\delta}\right\}\\
        < & 4\lambda\cdot \max\left\{\frac{140}{\delta^2}\cdot\log\left(\frac{9000}{\varepsilon\delta^2}\right),\frac{8(\timerdiff+4)}{\delta}\right\}
    \end{align*}

    In all, for any $r_2\ge 0$ and $\varepsilon>0$, we have
    $$ \Pr\left[ \exists N(r_2)<n\le N(r_2+1), \exists \block \in {\ti \graph}_{r_2}, \pot(\advs_n,\block)\ge w(\varepsilon) \right]\le \frac{14\varepsilon}{15}+\frac{\varepsilon}{9000}<\varepsilon.$$
    
\end{proof}

\subsection{Proof of Theorem~\ref{thm:mustoldenough}}

\begin{proof}
    Let $\graph_{n,r}\eqdef \left\{\block\in \ggen_{n}\middle|\gmin_{N(r)}\nsubseteq\block.\past \vee \block\in \gmin_{N(r)}\right\}$. 
    If there exists $n$ and $\block$ satisfying $\block \in \ggen_n$, $\gmin_{N(r)}\nsubseteq\block.\past$ and $\mathrm{Old}(\gmin_n,\block)=\false$, we claim $\block \in \graph_{n,r}$ and $\mathrm{MaxTH}(\gmin_n)-\mathrm{TimerHeight}(\block)\le \timerdiff$. Thus, $\mathrm{MaxTH}(\gmin_n)-\mathrm{TimerHeight}(\graph_{n,r})\le \timerdiff$. According to lemma~\ref{thm:chaindiffgrowth}, it will happen with probability $\varepsilon$. 
\end{proof}

\section{Chernoff bound}

\begin{lemma}[Multiplicative Chernoff bound]\label{lma:chernoff}
    Let $X$ be a random variable with binomial distribution $\mathrm{B}(n,p)$, then for any $\delta>0$, we have 
    $$ \Pr\left[X\ge (1+\delta)np\right]\le \left(\frac{e^\delta}{(1+\delta)^{1+\delta}}\right)^{np}.$$
    $$ \Pr\left[X\le (1-\delta)np\right]\le \left(\frac{e^{-\delta}}{(1-\delta)^{1-\delta}}\right)^{np}.$$
\end{lemma}


\end{document}